\documentclass[sigconf,nonacm]{acmart}

\setcopyright{acmcopyright}
\copyrightyear{2018}
\acmYear{2018}
\acmDOI{XXXXXXX.XXXXXXX}
\acmConference[Conference acronym 'XX]{Make sure to enter the correct conference title from your rights confirmation email}{June 03--05, 2018}{Woodstock, NY}
\acmISBN{978-1-4503-XXXX-X/2018/06}

\usepackage{graphicx}
\usepackage{subcaption}
\usepackage{amsmath}
\usepackage{amsfonts}
\usepackage{mathtools}
\usepackage{tabularx}
\usepackage{booktabs}
\usepackage[inline]{enumitem}
\usepackage{amsthm}
\usepackage{multirow}
\usepackage{multicol}
\usepackage{makecell}
\usepackage{algorithm}
\usepackage[noend]{algpseudocode}
\usepackage[capitalize,noabbrev]{cleveref}
\usepackage{bm}
\usepackage{xspace}
\usepackage{boxedminipage}
\usepackage{colortbl}
\usepackage{thmtools}
\usepackage{thm-restate}

\newcommand{\name}{Clownfish\xspace}

\newcommand\StateX{\Statex\hspace{\algorithmicindent}}
\newcommand\StateXX{\StateX\hspace{\algorithmicindent}}
\newcommand\StateXXX{\StateXX\hspace{\algorithmicindent}}
\algrenewcommand\textproc{}

\definecolor{yescolor}{HTML}{026378}
\def\yes{\textcolor{yescolor}{\boldsymbol{\checkmark}}}
\def\no{\textcolor{red}{\boldsymbol{\times}}}

\newcommand{\rmnum}[1]{\romannumeral #1}

\newcommand{\N}{\mathbb{N}}
\newcommand{\M}{\mathcal{M}}
\newcommand{\ML}{\mathcal{ML}}
\newcommand{\MLV}{\mathcal{MLV}}
\newcommand{\bigS}{\mathcal{S}}
\newcommand{\CLS}{\mathcal{CLS}}
\newcommand{\CMV}{\mathcal{CMV}}

\newcommand{\True}{\mathsf{true}}
\newcommand{\False}{\mathsf{false}}

\newcommand{\node}[1]{\ensuremath{p_{#1}}\xspace}
\newcommand{\tuple}[1]{\langle #1 \rangle}

\newcommand{\NoVote}{\mathsf{no\mbox{-}vote}}
\newcommand{\Timeout}{\mathsf{timeout}}
\newcommand{\FastVote}{\mathsf{fast\mbox{-}vote}}
\newcommand{\NewRound}{\mathsf{new\mbox{-}round}}

\newtheorem{lemma}{Lemma}
\newtheorem{theorem}{Theorem}
\newtheorem{property}{Property}

\newtheorem{fact}{Fact}
\theoremstyle{definition}
\newtheorem{definition}{Definition}



\algdef{SE}[EVENT]{Upon}{EndUpon}[1]{\textbf{upon} #1 \textbf{do}}{\textbf{end upon}}
\algtext*{EndUpon}

\title[\name: Scaling DAG-based BFT Consensus via Sparse Edges]{\name: Scaling DAG-based BFT Consensus via Sparse Edges}

\author{Feifan Wang}
\authornote{Both authors contributed equally to this research.}
\affiliation{%
  \institution{Tsinghua University}
  \city{Beijing}
  \country{China}}
\email{wff25@mails.tsinghua.edu.cn}

\author{Jingfan Yu}
\authornotemark[1]
\affiliation{%
  \institution{Tsinghua University}
  \city{Beijing}
  \country{China}}
\email{yujf20@mails.tsinghua.edu.cn}

\author{Zixi Cai}
\affiliation{%
  \institution{Tsinghua University}
  \city{Beijing}
  \country{China}}
\email{caizx22@mails.tsinghua.edu.cn}

\author{Zhixuan Fang}
\authornote{Corresponding author.}
\affiliation{%
  \institution{Tsinghua University}
  \city{Beijing}
  \country{China}}
\email{zfang@mail.tsinghua.edu.cn}

\begin{document}

\begin{abstract}

Directed Acyclic Graph (DAG) based BFT protocols have demonstrated the capability to achieve significantly high throughput in practice. Recent advancements focused on minimizing the good-case latency of these protocols, approaching the theoretical lower bound. However, the high communication complexity inherent in existing DAG-based protocols limits their scalability. This primarily arises because each vertex in the DAG must include a linear number of edges (references) to vertices from previous rounds.

We present \name, a partially synchronous DAG-based BFT protocol designed to address the scalability bottleneck. \name achieves lower communication complexity by selectively reducing the number of edges in DAG vertices. When using a communication-optimal consistent broadcast, \name attains quadratic total communication complexity per round, outperforming prior DAG-based protocols. \name also reduces the additional latency in failure cases by optimizing the round advancement rule. Additionally, \name supports multiple leaders per round to reduce average latency while maintaining its lower communication complexity. Our experimental evaluation demonstrates that \name provides significantly better scalability than existing DAG-based protocols.

\end{abstract}




\maketitle

\section{Introduction}

\begin{table*}[t]
    \footnotesize
    \centering
    \caption{Theoretical performance of partially synchronous DAG-based BFT protocols (after GST)}
    \label{table:comparison}
    \setlength\tabcolsep{12pt}
    \def\arraystretch{1}
    \begin{tabularx}{\textwidth}{l c c c c c c}
    \toprule
    \makecell{\textbf{Protocol}} &
    \makecell{\textbf{Broadcast} \\ \textbf{Primitive} \\ \textbf{Used}} &
    \makecell{\textbf{LV} \\ \textbf{Commit} \\ \textbf{Latency}} &
    \makecell{\textbf{NLV}$^{(1)}$ \\ \textbf{Commit} \\ \textbf{Latency}} &
    \makecell{\textbf{Communication}$^{(2)}$ \\ \textbf{Complexity} \\ \textbf{(Good / Bad Case)}} &
    \makecell{\textbf{NLV Latency}$^{(3)}$ \\ \textbf{Under Failure} \\ \textbf{(Single / Consecutive)}} &
    \makecell{\textbf{Multiple} \\ \textbf{Leaders}}
    \\
    \midrule

    \multirow{4}{*}{Bullshark{~\cite{spiegelman2022bullshark}}} & 
    {Bracha's RBC\cite{bracha1987asynchronous}} &
    {$6\delta$} &
    {$+6\delta$} &
    {$O(n^4)$}&
    {$+(5\Delta+3\delta)$} &
    \multirow{4}{*}{$\no$}
    \\
    &
    {CCBRB\cite{alhaddad2022balanced}} &
    {$6\delta$} &
    {$+6\delta$} &
    {$O(\lambda n^3)^*$}&
    {$+(5\Delta+3\delta)$} &
    \\
    &
    {Latency-optimal RBC\cite{abraham2021good}} &
    {$4\delta$} &
    {$+4\delta$} &
    {$O(n^4)$}&
    {$+(3\Delta+2\delta)$} &
    \\
    &
    {Narwhal's CBC\cite{danezis2022narwhal}} &
    {$6\delta$} &
    {$+6\delta$} &
    {$O(\lambda n^3) / O(\lambda n^4)$}&
    {$+(4\Delta+3\delta)$} &
    \\

    \arrayrulecolor{black!30}
    \midrule
    \arrayrulecolor{black}

    \multirow{4}{*}{Shoal++{~\cite{arun2025shoal++}}} & 
    {Bracha's RBC\cite{bracha1987asynchronous}} &
    {$4\delta$} &
    {$+3\delta$} &
    {$O(n^4)$}&
    {$+(5\Delta+3\delta)$} &
    \multirow{4}{*}{$\yes$}
    \\
    &
    {CCBRB\cite{alhaddad2022balanced}} &
    {$4\delta$} &
    {$+3\delta$} &
    {$O(\lambda n^3)^*$}&
    {$+(5\Delta+3\delta)$} &
    \\
    &
    {Latency-optimal RBC\cite{abraham2021good}} &
    {$3\delta$} &
    {$+2\delta$} &
    {$O(n^4)$}&
    {$+(3\Delta+2\delta)$} &
    \\
    &
    {Narwhal's CBC\cite{danezis2022narwhal}} &
    {$4\delta$} &
    {$+3\delta$} &
    {$O(\lambda n^3)$ / $O(\lambda n^4)$}&
    {$+(4\Delta+3\delta)$} &
    \\

    \arrayrulecolor{black!30}
    \midrule
    \arrayrulecolor{black}
    
    \multirow{4}{*}{Sailfish{~\cite{shrestha2025sailfish}}} & 
    {Bracha's RBC\cite{bracha1987asynchronous}} &
    {$4\delta$} &
    {$+3\delta$} &
    {$O(n^4)$}&
    {$+(5\Delta+2\delta)$ / $+(7\Delta+2\delta)$} &
    \multirow{4}{*}{$\yes$}
    \\
    &
    {CCBRB\cite{alhaddad2022balanced}} &
    {$4\delta$} &
    {$+3\delta$} &
    {$O(\lambda n^3)^*$}&
    {$+(5\Delta+2\delta)$ / $+(7\Delta+2\delta)$} &
    \\
    &
    {Latency-optimal RBC\cite{abraham2021good}} &
    {$3\delta$} &
    {$+2\delta$} &
    {$O(n^4)$}&
    {$+(3\Delta+2\delta)$ / $+(4\Delta+2\delta)$} &
    \\
    &
    {Narwhal's CBC\cite{danezis2022narwhal}} &
    {$4\delta$} &
    {$+3\delta$} &
    {$O(\lambda n^3)$ / $O(\lambda n^4)$}&
    {$+(4\Delta+2\delta)$ / $+(5\Delta+2\delta)$} &
    \\

    \arrayrulecolor{black!30}
    \midrule
    \arrayrulecolor{black}

    {Sparse Bullshark{~\cite{anoprenko2025dags}}} &
    {Narwhal's CBC\cite{danezis2022narwhal}} &
    {$6\delta$} &
    {$+6\delta$} &
    {$O(\lambda^2 n^2)$ / $O(\lambda^2 n^3)$} &
    {$+(4\Delta+3\delta)$} &
    {$\no$}
    \\

    {Cordial Miners{~\cite{keidar2023cordial}}} &
    {Best-effort Broadcast} &
    {$3\delta$} &
    {$+3\delta$} &
    {$O(\lambda n^3)$ / $O(\lambda n^4)$} &
    {$+6\Delta$} &
    {$\no$}
    \\

    {Mysticeti{~\cite{babel2025mysticeti}}} &
    {Best-effort Broadcast} &
    {$3\delta$} &
    {$+3\delta$} &
    {$O(\lambda n^3)$ / $O(\lambda n^4)$} &
    {$+(4\Delta+2\delta)$} &
    {$\yes$}
    \\

    \arrayrulecolor{black!30}
    \midrule
    \arrayrulecolor{black}

    \multirow{4}{*}{\textbf{\name}} & 
    {Bracha's RBC\cite{bracha1987asynchronous}} &
    {$4\delta$} &
    {$+3\delta$} &
    {\color{blue}$\bm{O(n^3\log n)}$}&
    {\color{blue}$\bm{+(5\Delta+\delta)}$} &
    \multirow{4}{*}{$\yes$}
    \\
    &
    {CCBRB\cite{alhaddad2022balanced}} &
    {$4\delta$} &
    {$+3\delta$} &
    {$O(\lambda n^3)^*$}&
    {\color{blue}$\bm{+(5\Delta+\delta)}$} &
    \\
    &
    {Latency-optimal RBC\cite{abraham2021good}} &
    {$3\delta$} &
    {$+2\delta$} &
    {\color{blue}$\bm{O(\lambda n^3)}$}&
    {\color{blue}$\bm{+(3\Delta+\delta)}$} &
    \\
    &
    {Narwhal's CBC\cite{danezis2022narwhal}} &
    {$4\delta$} &
    {$+3\delta$} &
    {\color{blue}$\bm{O(\lambda n^2)/O(\lambda n^3)}$} &
    {\color{blue}$\bm{+(4\Delta+\delta)}$} &
    \\

    \bottomrule
    \end{tabularx}
    \begin{flushleft}
    \textbf{LV} and \textbf{NLV} denote the leader vertex and non-leader vertex, respectively. (1)~This column represents the additional commit latency required for an NLV relative to the LV of the same round. (2)~This column represents the total metadata communication complexity per round. The terms ``good case'' and ``bad case'' indicate whether fetching missing data incurs additional communication overhead (relevant only when using Narwhal's CBC or best-effort broadcast). (3)~This column represents the additional latency imposed on an NLV by a Byzantine leader (relative to column (1)). The terms ``single'' and ``consecutive'' refer to scenarios involving a single or multiple consecutive Byzantine leaders, respectively (relevant only to Sailfish). $*$~Due to the use of erasure codes, the broadcast protocol incurs additional computational cost and larger vertex size. \textcolor{blue}{\textbf{Blue text highlights where \name demonstrates an advantage over other protocols (under specific broadcast primitives).}}
    \end{flushleft}
\end{table*}

Byzantine Fault Tolerant (BFT) consensus enables a set of replicas to consistently commit a sequence of values even under adversarial conditions. With the proliferation of decentralized systems such as blockchains, BFT consensus protocols—serving as their fundamental building blocks—have garnered extensive research attention~\cite{castro1999practical, yin2019hotstuff, spiegelman2022bullshark}. Under specific network and adversarial assumptions, the key metrics used to evaluate BFT consensus protocols include latency, throughput, and communication complexity. Among these, latency and throughput directly reflect the execution efficiency of a protocol, whereas communication complexity serves as a crucial indicator of system scalability~\cite{alqahtani2021bottlenecks}.

Following the seminal work of PBFT~\cite{castro1999practical}, traditional BFT consensus protocols have widely adopted the leader-based paradigm~\cite{buchman2018latest, yin2019hotstuff, gelashvili2022jolteon}. These protocols rely on a single, rotating leader to propose blocks containing transactions (data), while the remaining replicas participate in multi-phase voting to commit these blocks. In partial synchrony networks~\cite{dwork1988consensus}, protocols following this paradigm can achieve optimal latency~\cite{abraham2021good} (3 message delays in PBFT) or optimal communication complexity (linear complexity in HotStuff~\cite{yin2019hotstuff}) in the good case. However, the single leader becomes a system bottleneck that significantly constrains the achievable throughput~\cite{danezis2022narwhal}.

The key to enhancing throughput lies in fully utilizing the bandwidth of all replicas for data transmission. Effective approaches to achieve this include running multiple leader-based BFT instances concurrently with distinct leaders~\cite{stathakopoulou2019mir, stathakopoulou2022state}, as well as decoupling data dissemination from consensus logic~\cite{yang2022dispersedledger, giridharan2024autobahn}. A recently emerging paradigm is Directed Acyclic Graph (DAG)-based BFT~\cite{gkagol2019aleph, keidar2021all, danezis2022narwhal, spiegelman2022bullshark}. This paradigm naturally integrates the aforementioned ideas and has been successfully deployed in modern blockchains~\cite{blackshear2024sui} due to its simplicity and practical efficiency.

In DAG-based BFT protocols, replicas concurrently construct a DAG in a structured manner. Typically, these protocols proceed in \textit{rounds}. In each round, every replica disseminates a \textit{vertex} packed with transactions. Each vertex contains references to at least $n-f$ vertices from the previous round via \textit{edges}, where $n$ is the total number of replicas and $f$ is the maximum number of Byzantine replicas. Collectively, these vertices and edges form a DAG. The consensus logic is intrinsic to the DAG structure, allowing all replicas to commit a consistent DAG prefix through local interpretation.

DAG-based BFT protocols have demonstrated significantly high throughput in practice~\cite{danezis2022narwhal, spiegelman2022bullshark}. To achieve better efficiency, many recent DAG protocols have adopted the partial synchrony network assumption~\cite{spiegelman2024shoal, spiegelman2022bullshark_p, shrestha2025sailfish}. By pre-designating leader vertices (also referred to as \textit{anchors}) and introducing timeouts into the rounds, these protocols facilitate a total ordering of the DAG through simple commit rules. State-of-the-art protocols leverage this approach to achieve near-optimal latency in the good case~\cite{arun2025shoal++, shrestha2025sailfish, babel2025mysticeti}.

However, these performance gains often come at the cost of high communication complexity, which significantly constrains the scalability of DAG-based protocols. The high communication complexity in existing protocols stems from the structural requirement that each vertex must carry at least $n-f$ references (counted as \textit{metadata}). Given that $O(n)$ vertices are disseminated per round---typically via Reliable Broadcast (RBC)~\cite{bracha1987asynchronous} or Consistent Broadcast (CBC)~\cite{danezis2022narwhal}---the per-round metadata communication overhead amounts to at least $\Omega(n^3)$ times the size of a single reference. When references are represented by signatures or hashes, this overhead scales to $\Omega(\lambda n^3)$, where $\lambda$ denotes the security parameter.

Existing DAG-based protocols usually amortize this high communication overhead by batching $\Omega(n)$ transactions within each vertex. Leveraging erasure codes~\cite{das2021asynchronous}, this achieves amortized linear complexity per transaction. However, this approach faces two critical limitations. First, in many modern blockchain systems~\cite{aptoscan,suiscan}, a lot of blocks (vertices) only contain few or even zero transactions. As $n$ increases, waiting to accumulate $\Omega(n)$ transactions imposes prohibitive queueing latency, which is impractical. Second, even assuming sufficient transactions, since block sizes cannot increase indefinitely, this method remains unsustainable as $n$ scales~\cite{anoprenko2025dags}. 

Consequently, reducing the metadata communication complexity of DAG-based protocols is crucial for system scalability. This leads to a key research question: \textit{Can we reduce the metadata communication complexity of DAG-based protocols while maintaining the desired latency and throughput?}

\noindentparagraph{\textbf{Our solution.}}
To achieve this goal, we design \name, a partially synchronous DAG-based protocol built upon the state-of-the-art Sailfish~\cite{shrestha2025sailfish} protocol. \name's key insight is that in DAG-based protocols, only leader vertices are directly committed and responsible for establishing paths to the ordered history, whereas non-leader vertices serve primarily to reference the leader vertex to enable its commitment. Based on this observation, \name introduces a novel reference format termed the \textit{leader edge}. This allows a non-leader vertex to reference only a single leader vertex via a leader edge, while the standard $n-f$ references is maintained exclusively by leader vertices. Since each round contains one leader vertex and $O(n)$ non-leader vertices, this approach effectively reduces the overall communication complexity.

Leveraging this core design, \name further incorporates several enhancements. (\rmnum{1}) \name requires replicas to immediately broadcast a ``no-vote'' message upon a timeout. Coupled with an optimized round advancement rule, this reduces additional latency under failure cases. (\rmnum{2}) We extend \name to Multi-leader \name, which supports multiple leaders per round. By redesigning leader edges, Multi-leader \name further reduces average latency while maintaining lower communication complexity. (\rmnum{3}) In addition to the standard RBC-based protocol, we also provide a CBC-based variant of \name\footnote{More precisely, this variant is built on the Narwhal's CBC~\cite{danezis2022narwhal}, which combines CBC with a random pulling mechanism.}. This addresses a theoretical gap in prior DAG-based protocols, whose correctness under this weaker broadcast primitive is often not formally proved.

In summary, \name offers the following advantages:
\begin{itemize}[noitemsep]
    \item\textbf{Lower communication complexity.} When using a latency-optimal RBC~\cite{abraham2021good}, \name achieves a per-round communication complexity of $O(\lambda n^3)$, yielding a reduction factor of $\frac{n}{\lambda}$ over state-of-the-art protocols. When using Narwhal's CBC~\cite{danezis2022narwhal}, it achieves $O(\lambda n^2)$ communication per round, reducing the complexity by a factor of $n$.
    \item\textbf{Reduced additional latency under failures.} By optimizing round advancement rule, \name effectively minimizes the additional commit latency induced by faulty leaders or network asynchrony.
    \item\textbf{Broader compatibility.} We design \name and Multi-leader \name based on both RBC and CBC. This compatibility stems from the fact that \name's core design is generically applicable to various broadcast primitives.
\end{itemize}

\Cref{table:comparison} presents a detailed comparison of the theoretical performance of \name against other partially synchronous DAG-based protocols. Under the same broadcast primitive, \name demonstrates advantages in both communication complexity and latency over existing protocols. 

Our empirical evaluation consists of both simulation and deployment studies. The simulation results show that \name provides better scalability than existing DAG-based protocols at large system sizes. The deployment results under medium-scale systems and bandwidth-limited networks demonstrate that \name achieves lower latency by reducing metadata communication.

\noindentparagraph{\textbf{Organization.}}
The rest of this paper is organized as follows. \Cref{s:preliminary} introduces the model and background information on DAG-based BFT. \Cref{s:overview} provides a technical overview of \name. \Cref{s:protocol} and \Cref{s:protocol_ml} detail and analyze \name and Multi-leader \name, respectively. \Cref{s:eval} presents the results of our evaluation. We summarize the paper and provide further discussion in \Cref{s:discussion}. Finally, we review related work in \Cref{s:related_work}.
\section{Preliminaries}\label{s:preliminary}

\subsection{Model}\label{ss:model}

We consider a system consisting of a fixed set of $n=3f+1$ replicas. Let \node{i} denote a specific replica, where $i \in \{1, 2, \dots, n\}$. At most $f$ replicas are \textit{Byzantine} and can act arbitrarily. The remaining replicas are referred to as \textit{honest}. We assume the existence of an adversary capable of controlling all Byzantine replicas.

We consider the standard partial synchrony model~\cite{dwork1988consensus}. Specifically, there exists an unknown \textit{Global Stabilization Time} (GST) and a known upper bound $\Delta$ on network delay, such that any message sent by an honest replica at time $t$ is guaranteed to arrive at its recipient by time $\max(GST, t) + \Delta$. We further assume that after GST, the actual network delay $\delta$ satisfies $\delta \leq \Delta$. 

Regarding cryptographic primitives, we assume the availability of a Public Key Infrastructure (PKI), cryptographic hash functions, and threshold/aggregate signatures~\cite{boneh2004short}. We denote a message $m$ signed by \node{i} as $\tuple{m}_i$. 
We assume a computationally bounded adversary and let $\lambda$ denote the security parameter for these primitives. Throughout this paper, we consider the setting where $\log n < \lambda < n$.

\subsection{DAG-based BFT Consensus}\label{ss:dag_bft}

We focus exclusively on certified DAG protocols under the partial synchrony model. The term ``certified'' implies that a vertex in the DAG must be delivered via Reliable Broadcast (RBC) or Consistent Broadcast (CBC), these primitives are introduced subsequently. 

Typically, DAG-based BFT protocols operate in \textit{rounds}~\cite{danezis2022narwhal}. In each round, each replica can create a \textit{vertex} containing a batch of transactions and a set of \textit{edges}. To be deemed valid, a vertex in round $r$ is required to reference at least $2f+1$ delivered vertices from round $r-1$. Upon the completion of RBC or CBC, the corresponding vertex is added to the DAG. Each replica maintains a local DAG view, which may differ from that of other replicas. However, both RBC and CBC guarantee non-equivocation, implying that vertices appearing at the same position within the DAG are identical. The edges in the DAG are utilized to commit and order vertices. We define the \textit{causal history} of a vertex as the subgraph originating from it, encompassing all predecessor vertices reachable via a path.

Under the partial synchrony model, a pre-defined \textit{leader vertex} is designated every few rounds (e.g., every round in Sailfish~\cite{shrestha2025sailfish} and every two rounds in Bullshark~\cite{spiegelman2022bullshark_p}). Only leader vertices can be directly committed. The remaining non-leader vertices are ordered as part of the causal history of the committed leader vertices. 

To \textit{directly commit} a leader vertex, a replica must observe sufficient ``votes'' for it. In the context of a DAG, edges from vertices in round $r+1$ to a leader vertex in round $r$ are interpreted as votes for it. Distinct protocols impose varying requirements regarding the quantity and format of these votes. For instance, Sailfish requires $2f+1$ \textit{first messages}\footnote{Here, a ``first message'' denotes the message broadcast in the first stage of the underlying RBC or CBC. At this stage, the vertex is not yet considered delivered.} for $2f+1$ vertices, whereas Bullshark requires $f+1$ delivered vertices. Upon committing a leader vertex, a replica traverses its causal history based on the local DAG view and recursively checks for the existence of any uncommitted leader vertices reachable via a path. If such vertices exist, the replica must first \textit{indirectly commit} the corresponding leader vertices. This ensures that all replicas derive the same committed-leader sequence and thus the same transaction order.

\subsection{Problem Definition}

In the context of DAG-based BFT, we focus on Byzantine Atomic Broadcast (BAB) problem. We use $a\_bcast(m,r)$ to denote the event of a replica broadcasting a message $m$ with sequence number $r$. We use $a\_deliver(m,r,\node{i})$ to denote the event of a replica delivering a message $m$ with sequence number $r$ originating from \node{i}.

\begin{definition}[Byzantine atomic broadcast~\cite{keidar2021all}]
\label{def:BAB}
Each honest replica $\node{i}$ can call $a\_bcast_i(m,r)$ and output $a\_deliver_i(m,r,\node{k})$. A Byzantine atomic broadcast protocol satisfies the following properties:
\begin{itemize}[noitemsep,leftmargin=*]
    \item[-] \textbf{Agreement.} If an honest replica $\node{i}$ outputs $a\_deliver_i(m, r, \node{k})$, then every honest replica $\node{j}$ eventually outputs $a\_deliver_j(m, r, \node{k})$.
    
    \item[-] \textbf{Integrity.} For every $r\in\mathbb{N}$ and replica $\node{k}$, an honest replica $\node{i}$ outputs $a\_deliver_i(m,r,\node{k})$ at most once regardless of $m$.

    \item[-] \textbf{Validity.} If an honest replica $\node{k}$ calls $a\_bcast_k(m,r)$, then every honest replica eventually outputs $a\_deliver(m,r,\node{k})$.
    
    \item[-] \textbf{Total order.} If an honest replica $\node{i}$ outputs $a\_deliver_i(m, r,\node{k})$ before $a\_deliver_i(m', r', \node{k'})$, then no honest replica $\node{j}$ outputs $a\_deliver_j(m', r', \node{k'})$ before $a\_deliver_j(m, r, \node{k})$.
\end{itemize}
\end{definition}

We also introduce the two fundamental broadcast primitives employed in DAG construction:

\noindentparagraph{\textbf{Reliable broadcast (RBC).}}
Let $r\_bcast(m,r)$ and $r\_deliver(m,r,\node{k})$ denote the events of broadcasting and delivering a message $m$ with round number $r$, respectively. The RBC primitive satisfies the \textit{Agreement}, \textit{Integrity}, and \textit{Validity} properties defined in \Cref{def:BAB}.

We outline the properties satisfied by an RBC protocol after GST in \Cref{property:RBC}. Note that $k_1$ and $k_2$ are RBC-related parameters that vary across different RBC protocols~\cite{bracha1987asynchronous, abraham2021good, alhaddad2022balanced, das2021asynchronous}.

\begin{property}\label{property:RBC}
Let $t$ and $t'$ be times after GST. (\rmnum{1}) If an honest replica reliably broadcasts a vertex $v$ at time $t$, then all honest replicas will deliver $v$ by time $t + k_1\Delta$. (\rmnum{2}) If an honest replica delivers a vertex $v'$ at time $t'$, then all honest replicas will deliver $v'$ by time $t' + k_2\Delta$.
\end{property}

\noindentparagraph{\textbf{Consistent broadcast (CBC).}}
CBC is a broadcast primitive weaker than RBC~\cite{cachin2011introduction}. It satisfies the \textit{Validity} and \textit{Integrity} properties defined in \Cref{def:BAB}, but substitutes the \textit{Agreement} property with the following \textit{Consistency} property:
\begin{itemize}[noitemsep,leftmargin=*]
    \item[-] \textbf{Consistency.} If an honest replica \node{i} outputs $c\_deliver_i(m, r, \node{k})$ and another honest replica \node{j} outputs $c\_deliver_j(m',r,\node{k})$, then $m=m'$.
\end{itemize}

A CBC protocol satisfies only the part (\rmnum{1}) of \Cref{property:RBC}.

\subsection{Efficiency Measure}

\begin{figure*}[t]
	\centering
	\includegraphics[width=0.86\textwidth]{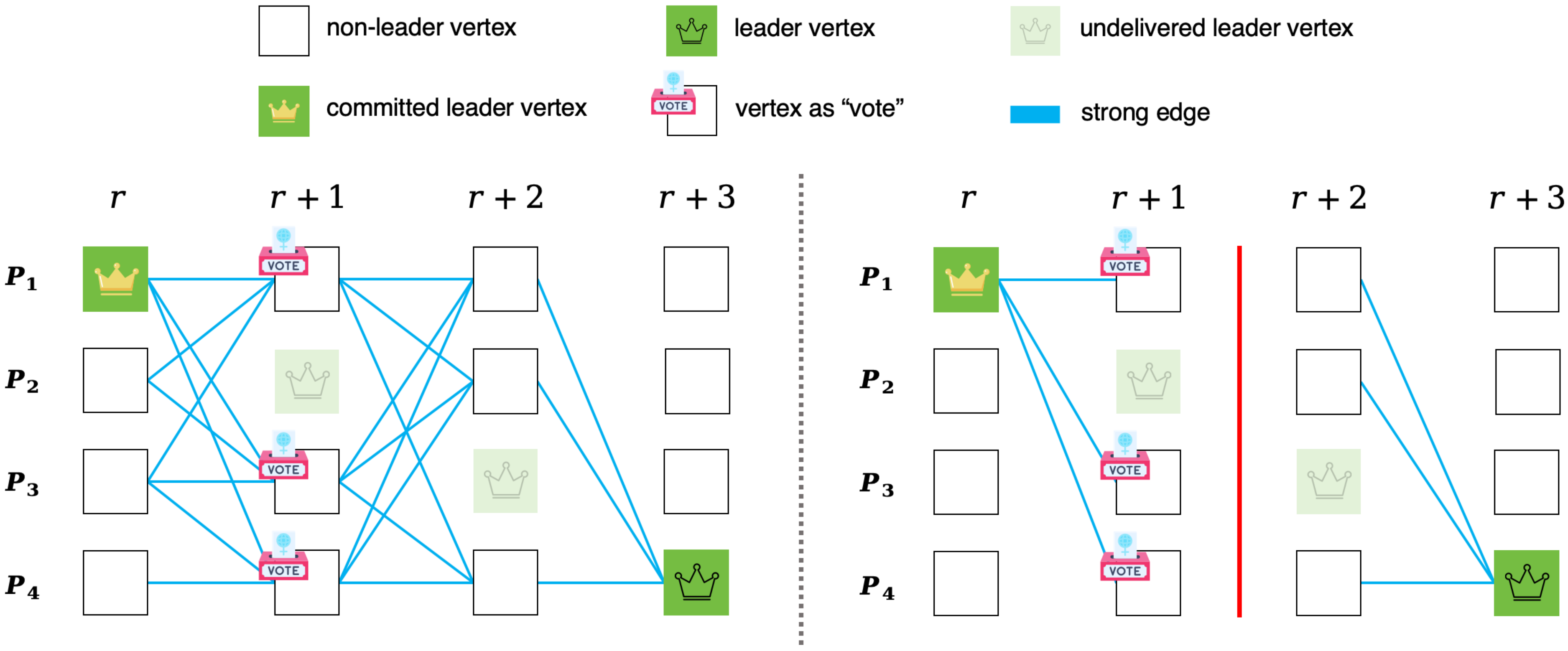}
    \vspace{-1em}
	\caption{Illustration of the challenge. Here, $n=4$ and $f=1$, with a pre-designated leader assigned for each round. A strong edge denotes a reference to a vertex in the previous round. According to the commit rule in Sailfish~\cite{shrestha2025sailfish}, $LV_r$ is committed as it receives $2f+1$ votes. (1) The left side depicts an execution in Sailfish. Since all vertices reference at least $2f+1$ vertices from the previous round, a path exists from $LV_{r+3}$ to $LV_r$. (2) The right side depicts an execution where non-leader vertices are restricted to reference \textit{only} the leader vertex of the previous round. In this case, $LV_{r+3}$ fails to establish a path to $LV_r$.}
\label{fig:main_challenge}
\end{figure*}

\noindentparagraph{\textbf{Communication complexity.}}
We focus on \textit{metadata} communication complexity, where metadata consists of DAG edges and other consensus messages, excluding the transaction payload. The reasons are twofold. First, metadata communication complexity asymptotically dominates the overall communication complexity and thus reflects the protocol's scalability. Second, since blocks in many practical systems contain only a small transaction payload~\cite{aptoscan, suiscan} or merely lightweight transaction availability certificates or batch digests~\cite{danezis2022narwhal,giridharan2024autobahn}, metadata still accounts for a substantial portion of communication and computation overhead. Therefore, we evaluate the total number of metadata bits transmitted by all honest replicas per round, which more fundamentally reflects the protocol's scalability.

\noindentparagraph{\textbf{Latency.}}
We focus on the latency measured from the time a vertex is broadcast until it is committed by an honest replica. We exclude the queuing latency as it depends on whether using a worker layer and assumptions regarding transaction batch sizes and arrival rates.
\section{Technical Overview of \name}\label{s:overview}

Building upon the foundation of Sailfish~\cite{shrestha2025sailfish}, \name employs several techniques to improve protocol performance. We outline the key idea, the challenges, and our solutions in this section.

\noindentparagraph{\textbf{The key idea.}}
The high communication complexity of DAG-based BFT protocols stems from the requirement that every vertex---whether a leader vertex or not---needs to carry at least $n-f$ references to vertices from the previous round. However, a key observation is that the commit and ordering process in DAG-based BFT protocols is driven primarily by leader vertices. This implies that the edges in the vast majority of non-leader vertices are redundant. Leveraging this insight, \name decreases the number of edges in non-leader vertices from $O(n)$ to $O(1)$. Since there is at most one leader vertex among $n$ vertices in each round, this strategy effectively reduces the communication complexity.

\noindentparagraph{\textbf{The main technical challenge.}}
The safety of DAG-based BFT protocols hinges on the guarantee that a path exists between any committed leader vertex $LV_r$ in round $r$ and a subsequent leader vertex $LV_{r'}$ in round $r'>r$. In the safety proofs of Sailfish~\cite{shrestha2025sailfish}, once the base cases where $r' \leq r+2$ are established, the existence of paths for cases where $r' > r+2$ follows straightforwardly from transitivity. However, this property ceases to hold when we naively restrict non-leader vertices to reference \textit{only} the leader vertex of the previous round. To illustrate this, consider a possible execution depicted in \Cref{fig:main_challenge}. Suppose that due to network asynchrony or Byzantine behavior, the leader vertices $LV_{r+1}$ and $LV_{r+2}$ are not delivered in the view of all replicas. This implies that the vertices created by all non-leader replicas in round $r+2$ will not contain any edges\footnote{This arises because we require vertices to reference exclusively the previous leader.
}. This causes the DAG to become structurally ``disconnected'' (indicated by the red line in the figure). Consequently, the leader vertices in round $r+3$ and subsequent rounds are unable to establish a path to $LV_r$ (except potentially via non-essential weak edges).

\noindentparagraph{\textbf{Our main solution: Ensuring connectivity via leader edge.}}
To address the above challenge, \name introduces the concept of \textit{leader edge}. Specifically, a leader edge allows a replica to reference the leader vertex that was most recently delivered locally before the creation of its current vertex, rather than being restricted to the leader vertex of the previous round. With each non-leader vertex referencing a single leader via a leader edge, \name efficiently guarantee paths between leader vertices (see \Cref{fig:protocol} for an illustration). Intuitively, the leader edge plays a role analogous to the ``highest lock'' in leader-based BFT protocols~\cite{yin2019hotstuff, malkhi2023hotstuff}.


\noindentparagraph{\textbf{Reducing additional latency under failure cases.}}
In Sailfish~\cite{shrestha2025sailfish}, replicas are required to sequentially transmit $\Timeout$ and $\NoVote$ messages upon a timeout, incurring two additional message delays. To overcome this limitation, \name consolidates the $\Timeout$ and $\NoVote$ messages found in Sailfish. In \name, upon a timeout, a replica can immediately broadcast a $\NoVote$ message to signal that it will not reference the corresponding leader vertex via a leader edge. However, such a direct merger poses a threat to liveness (see \Cref{ss:protocol} for further discussion). \name addresses this issue through the \textit{fast-vote} message. Intuitively, fast-vote permits replicas to cast lightweight votes for leader vertices in skipped rounds to help ensure they are committed. \name also provides more efficient round-skipping rules that enable lagging replicas to leverage messages from fast replicas to advance to the latest round.

\noindentparagraph{\textbf{Utilizing different broadcast primitives for protocol design.}}
We first design the \name prsotocol built upon standard RBC. It is intuitive and allows for flexible tradeoffs between latency and communication depending on the specific RBC implementation used. We further design a CBC-based \name, which fills the theoretical gap in existing DAG-based protocols that lack a corresponding correctness proof. It utilizes an explicit round synchronization mechanism to ensure that the critical path of consensus is not obstructed by data fetching. In conjunction with a delayed timer design, it maintains a short timeout duration. It can achieve amortized linear complexity without batching $\Omega(n)$ transactions in vertices. 

\noindentparagraph{\textbf{Supporting multiple leaders per round.}}
Consistent with recent protocols~\cite{shrestha2025sailfish, arun2025shoal++, yu2025angelfish}, we extend \name to a multi-leader variant (Multi-leader \name) to reduce average latency. We adopt a strategy that differentiates between \textit{main leader} and \textit{secondary leader} to balance waiting times against leader quantity. Under this strategy, only a main leader is responsible for establishing paths to previous leader vertices, while secondary leaders share the same structure as non-leader vertices. This allows for a seamless adaptation of the leader edge design. By interpreting a single leader edge as a vote for multiple leader vertices, we modify the commit rules to enable committing multiple leader vertices within the same round. Consequently, Multi-leader \name improves average latency while preserving low communication complexity.
\section{The \name Protocol}\label{s:protocol}

In this section, we present the \name protocol. We provide a detailed description of \name in \Cref{ss:protocol}, analyze its theoretical properties in \Cref{ss:analysis}, and finally discuss the CBC-based variant of \name in \Cref{ss:protocol_CBC}.

\subsection{Protocol Description}\label{ss:protocol}

In \name, the DAG is partitioned into a sequence of numbered rounds. Each round is assigned a designated leader, which is selected via a deterministic mechanism. We denote the leader of round $r$ as $L_r$, and the vertex created by $L_r$ as $LV_r$. 

\begin{algorithm*}[!t]
    \caption{\name's data structures and basic utilities for replica \node{i}}
    \label{alg:data_structure}
    \begin{algorithmic}[1]
    \footnotesize
    \Statex \textbf{Local variables:}
    \StateX struct vertex $v$:
    \Comment{The struct of a vertex in the DAG}
    \StateXX $v.round$ - the round of $v$ in the DAG
    \StateXX $v.source$ - the replica that broadcasts $v$
    \StateXX $v.block$ - a block of transactions
    \StateXX $v.strongEdges$ - a set of vertices in $v.round-1$ that represent strong edges
    {\color{magenta}\Comment{Only leader vertices need to contain}}
    \StateXX $v.weakEdges$ - a set of vertices in rounds $<$ $v.round-1$ that represent weak edges
    {\color{magenta}\Comment{Only leader vertices need to contain}}
    \color{magenta}
    \StateXX $v.leaderEdge$ - a leader vertex in round $\le$ $v.round-1$ that represents a leader edge
    \StateXX $v.selfEdges$ - a set of vertices in rounds $\le$ $v.round-1$ that represent self edges (created by \node{i})
    \StateXX $v.nvc$ - a $\NoVote$ certificate for $v.round-1$ (if any)
    \Comment{Only leader vertices need to contain}
    \StateX struct fast-vote $fv$: 
    \Comment{The struct of a fast-vote for a previous leader vertex}
    \StateXX $fv.round$ - the round of $fv$
    \StateXX $fv.source$ - the replica that broadcasts $fv$
    \StateXX $fv.leader$ - leader vertex that the replica votes for
    \normalcolor
    \StateX $DAG_i[]$ - An array of sets of vertices (indexed by rounds)
    \StateX {\color{magenta}$FV_i[]$ - An array of sets of fast-votes (indexed by rounds)}
    \StateX $blocksToPropose$ - A queue, initially empty, $\node{i}$ enqueues valid blocks of transactions from clients
    \color{magenta}
    \StateX $lastLeader$ - The most recent delivered leader vertex for which a no-vote has not been sent
    \StateX $oldVertices$ - The own vertices (created by \node{i}) that have been delivered but not yet referenced
    \normalcolor

    \Procedure{path}{$v,u$}
    \label{line:path}
    \Comment{Check if exists a path consisting of \textcolor{magenta}{all kinds of} edges in the DAG}
    \State \Return exists a sequence of $k \in \N$, vertices $v_1,\ldots, v_k$ s.t. 
    \StateXX $v_1 = v$, $v_k = u$, and $\forall j \in [2,..,k]$:
    $v_j \in \bigcup_{r\ge 1} DAG_i[r] \land (v_j \in v_{j-1}.weakEdges \cup v_{j-1}.strongEdges  \cup  \textcolor{magenta}{v_{j-1}.leaderEdge \cup v_{j-1}.selfEdges})$
    \EndProcedure

    \color{magenta}
    \Procedure{leader\_path}{$v,u$}
    \label{line:leader_path}
    \Comment{Check if exists a path consisting of leader edges and strong edges from leader vertex $v$ to leader vertex $u$}
	\State \Return exists a sequence of $k \in \N$, vertices $v_1,\ldots, v_k$ s.t. \StateXX $v_1 = v$, $v_k = u$, and $\forall j \in [2,..,k]$:
    $v_j \in \bigcup_{r\ge 1} DAG_i[r] \land (v_j \in  v_{j-1}.leaderEdge \cup v_{j-1}.strongEdges)$ 
	\EndProcedure
    \normalcolor

    \vspace{-1em}
    \setlength{\columnsep}{25pt}
    \begin{multicols}{2}
    \Procedure{set\_weak\_edges}{$v, r$}
    \State $v.weakEdges \gets \{\}$
    \For{$r' = r-2 \text{ down to } 1$}
        \For{\textbf{every} $u \in DAG_i[r']$ s.t.  $\neg$\Call{path}{$v,u$}}
            \State $v.weakEdges \gets v.weakEdges \cup \{u\}$
        \EndFor
    \EndFor
    \EndProcedure

    {\color{magenta}
	\Procedure{set\_self\_edges}{$v$}
    \State $v.selfEdges \gets \{\}$
	\For{every $u \in oldVertices$ s.t. $\neg$\Call{path}{$v,u$}}
        \State $v.selfEdges \gets v.selfEdges \cup \{u\}$
        \State $oldVertices \gets oldVertices \backslash \{u\}$
    \EndFor
	\EndProcedure
    }

    \columnbreak

    {\color{magenta}
    \Procedure{set\_leader\_edge}{$v$}
    \State $v.leaderEdge \gets lastLeader$
    \EndProcedure
    }

    \Procedure{get\_vertex}{$p, r$}
    \If{$\exists v \in DAG_i[r]$ s.t. $v.source = p$}
        \State \Return $v$
    \EndIf
    \State \Return $\bot$
    \EndProcedure

    \Procedure{get\_leader\_vertex}{$r$}
    \State \Return \Call{get\_vertex}{$L_r, r$} 
    \EndProcedure

    \Procedure{a\_bcast$_i$}{$b$}
    \State $blocksToPropose.$enqueue($b$)
    \EndProcedure
    \end{multicols}
    \vspace{-2em}

    \algstore{break_R}
    \end{algorithmic}
\end{algorithm*}

\noindentparagraph{\textbf{DAG components.}}
The data structures and basic utilities of \name are presented in \Cref{alg:data_structure}, with our modifications over Sailfish~\cite{shrestha2025sailfish} highlighted in magenta. In each round, each replica is permitted to propose one vertex containing a block of transactions (which may be empty). Each vertex must also include a set of edges to integrate into the DAG. All vertices are disseminated via RBC. Crucially, \textit{only} leader vertices are required to reference at least $2f+1$ vertices from the previous round via \textit{strong edges}. A leader vertex may also references up to $f$ vertices from earlier rounds to which a path has not yet been established via \textit{weak edges}. All vertices must reference a leader vertex via a \textit{leader edge}. Optionally, a vertex may also reference vertices previously created by its creator via \textit{self edges}. In particular, a leader vertex may additionally be required to carry a $\NoVote$ certificate ($NVC$) regarding the leader of the previous round via $v.nvc$. We elaborate on these elements in the DAG construction subsection. 

\begin{figure}[h]
	\centering
	\includegraphics[width=0.4\textwidth]{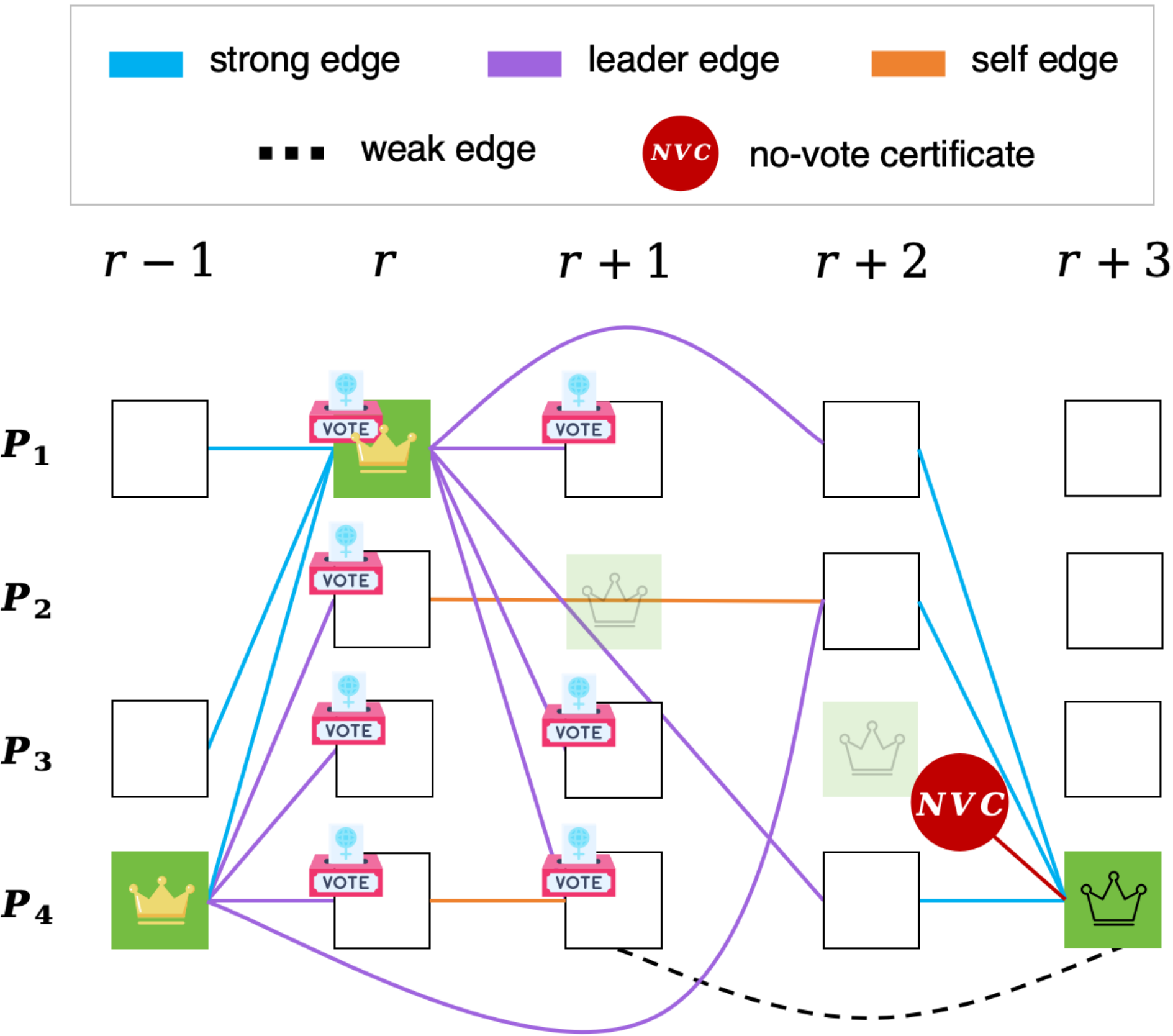}
	\caption{Illustration of \name. $LV_{r-1}$ and $LV_r$ are directly committed. The vertices created by replicas \node{1} and \node{4} in round $r+2$ reference $LV_r$ via leader edges. In addition to $2f+1$ strong edges, $LV_{r+3}$ provides $NVC_{r+2}$ to prove that $LV_{r+2}$ cannot be directly committed. Together, strong edges and leader edges constitute a \textit{leader path} between $LV_{r+3}$ and $LV_r$.}
\label{fig:protocol}
\end{figure}

We define a \textit{path} as a connection between two vertices established via any of the aforementioned four types of edges (Line~\ref{line:path}). A \textit{leader path} specifically denotes a path between two leader vertices consisting of leader edges and strong edges (Line~\ref{line:leader_path}).

Each replica maintains a local view of the DAG. Upon the completion of the RBC, the replica adds the vertex and its edges to its local DAG based on the corresponding source and round. Although the DAG views of all replicas may differ at any specific point in time, the RBC primitive guarantees their eventual consistency.

In addition to vertices, \name introduces a struct called \textit{fast-vote}. A fast-vote specifies the round, the source, and the target leader vertex for which the vote is intended. As detailed in the subsequent subsection, this type of message is created exclusively during round-skipping scenarios. Their objective is to help committing the leader vertices from preceding rounds. 

\begin{algorithm*}[!t]
    \caption{\name's DAG construction protocol for replica \node{i}}
    \label{alg:dag_construction}
    \begin{algorithmic}[1]
    \footnotesize
    \algrestore{break_R}

    \vspace{-2em}
    \setlength{\columnsep}{25pt}
    \begin{multicols}{2}
    \Statex \textbf{Local variables:}
    \StateX $round \gets 1$; $\textit{buffer} \gets \{\}$
    
    \Upon {\Call{r\_deliver$_i$}{$v, r, p$}} \label{line:r_deliver}
    \If {{\color{magenta}$v.source = p \land v.round = r \land$ is\_valid$(v)$}}
    \If {$\neg$\Call{try\_add\_to\_dag}{$v$}} \label{line:add_to_dag}
    \State $\textit{buffer} \gets \textit{buffer} \cup \{v\}$
    \Else
    \For {$v' \in \textit{buffer} : v'.round \geq r$}
    \State \Call{try\_add\_to\_dag}{$v'$}
    \EndFor
    \EndIf
    \EndIf
    \EndUpon

    {\color{magenta}
    \Upon {receiving $\tuple{fv, r,}_p$}
    \If {$fv.round=r \land$ is\_valid$(fv)$} 
    \State $FV_i[r]\gets FV_i[r]\cup \{fv\}$
    \EndIf
    \EndUpon
    }

    \Upon {$|DAG_i[r]|\geq 2f+1 \land (\exists v' \in DAG_i[r]$: $ v'.source=L_r$ $\lor$
    \StateX ${\color{magenta}NVC_r \text{ is received }})\text{ for } r \geq round$}
    \label{line:advance}
    {\color{magenta}
    \If {$r>round$}
    \Comment{Skip to a higher round (rule 1)}
    \State \Call{send\_fast\_vote}{$round+1,r$}
    \label{line:fast_vote_1}
    \EndIf
    }
    \State \Call{advance\_round}{$r+1$}
    \EndUpon

    {\color{magenta}
    \Upon {receiving a set $\M$ of $\geq f+1$ first messages for round $r+1$ vertices $\land$ 
    \StateX ($\exists v' \in DAG_i[r]: v'.source=L_r$ $\lor$ $NVC_r$ is received) for $r \geq round$}
    \label{line:jumping}
    \If {$r > round$}
    \Comment{Skip to a higher round (rule 2)}
    \State \Call{send\_fast\_vote}{$round+1,r$}
    \label{line:fast_vote_2}
    \EndIf
    \State \Call{advance\_round}{$r+1$}
    \EndUpon
    }

    {\color{magenta}
    \Upon {$\Timeout$ in $round$}
    \label{line:timeout}
    \If {$\nexists v'\in DAG_i[round]:v'.source=L_{round}$}
    \State broadcast $\tuple{\NoVote, round}_i$
    \EndIf
    \EndUpon

    \Upon {receiving $NVC_r$ for $r \geq round$}
    \State broadcast $NVC_r$
    \label{line:NVC_sync}
    \EndUpon

    \Upon {\Call{r\_deliver$_i$}{$LV_r,r,L_r$} for $r>lastLeader.round$} \label{line:receive_leader}
    \If {haven't sent $\tuple{\NoVote,r}_i$ before}
    \label{line:novote}
    \State $lastLeader \gets LV_r$
    \label{line:last_leader}
    \EndIf
    \EndUpon
    }

    \Procedure{broadcast\_vertex}{$r$}
    \State $v \gets$ \Call{create\_new\_vertex}{$r$}
    \State \Call{try\_add\_to\_dag}{$v$}
    \State \Call{r\_bcast$_i$}{$v, r$}
    \EndProcedure

    \columnbreak

    \Procedure{create\_new\_vertex}{$r$}
    \State $v.round \gets r$
    \State $v.source \gets \node{i}$
    \State $v.block \gets blocksToPropose$.dequeue()
    \color{magenta}
    \State \Call{set\_leader\_edge}{$v$}
    \label{line:leader_edge}
    \If {$\node{i} = L_r$}
    \State $v.strongEdges \gets DAG_i[r-1]$
    \label{line:strong_edge}
    \If {$\not\exists v' \in DAG_i[r-1]: v'.source=L_{r-1}$}
    \State $v.nvc \gets NVC_{r-1}$
    \label{line:no_vote_cert}
    \EndIf
    \State \Call{set\_weak\_edges}{$v, r$}
    \label{line:weak_edge}
    \EndIf
    \State \Call{set\_self\_edges}{$v$}
    \label{line:self_edge}
    \normalcolor
    \State \Return $v$
    \EndProcedure

    {\color{magenta}
    \Procedure{create\_fast\_vote}{$r$}
    \State $fv.round \gets r$
    \State $fv.source \gets \node{i}$
    \State $fv.leader \gets L_{r-1}$
    \State \Return $fv$
    \EndProcedure

    \Procedure{send\_fast\_vote}{$start, end$}
    \State start $timer$ for $\FastVote$
    \label{line:tau'}
    \Comment{Start fast-vote timer $\tau'$}
    \For {$r'=start$ up to $end$ simultaneously}
    \If {\Call{r\_deliver$_i$}{$LV_{r'-1},r'-1,L_{r'-1}$} before $\Timeout$ $\land$ $NVC_{r'-1}$ is not
    \StateXX received $\land$ haven't sent $\tuple{\NoVote,r}_i$ before} \label{line:fv_condition}
    \State $fv \gets$ \Call{create\_fast\_vote}{$r'$}
    \State broadcast $\tuple{fv, r'}_i$
    \label{line:broadcast_fv}
    \EndIf
    \EndFor
    \EndProcedure    
    }

    \Procedure {try\_add\_to\_dag}{$v$}
    \If {$\forall v' \in v.srongEdges \cup v.weakEdge$ $\textcolor{magenta}{\cup~v.leaderEdge}$
    \StateX $\textcolor{magenta}{\cup~v.selfEdges}: v' \in \bigcup_{k \geq 1} DAG_i[k]$}
    \State $DAG_i[v.round] \gets DAG_i[v.round] \cup \{v\}$
    \State $\textit{buffer}\gets \textit{buffer}\ \backslash\ \{v\}$
    \State \Return $\True$
    \EndIf
    \State \Return $\False$
    \EndProcedure

    \Procedure {advance\_round}{$r$}
    \If {\node{i}$=L_r$}
    \color{magenta}
    \State \textbf{wait until} $|DAG_i[r-1]| \geq 2f+1$ 
    \normalcolor
    \EndIf
    \State $round\gets r$; start $timer$ for $round$
    \label{line:tau}
    \Comment{Start round timer $\tau$}
    \State \Call{broadcast\_vertex}{$round$}
    \EndProcedure
    \end{multicols}
    \vspace{-2em}

    \algstore{break_R}
    \end{algorithmic}
\end{algorithm*}

\noindentparagraph{\textbf{DAG construction protocol.}}
The DAG construction protocol of \name is presented in \Cref{alg:dag_construction}. Upon delivering $2f+1$ vertices in round $r$ (Line~\ref{line:advance}), replica \node{i} advances to round $r+1$ and creates a new vertex if it has delivered $LV_r$ or received a $\NoVote$ certificate for round $r$ (denoted as $NVC_r$). A $\NoVote$ certificate is composed of $2f+1$ $\NoVote$ messages (which can be aggregated using threshold signature). Upon entering a new round, \node{i} starts a timer of duration $\tau$ (Line~\ref{line:tau}). If \node{i} times out while waiting for $LV_r$ in round $r$, it broadcasts a $\tuple{\NoVote, r}_i$ message (Line~\ref{line:timeout}). When \node{i} receives $NVC_r$, it forwards the certificate to all replicas (Line~\ref{line:NVC_sync}). 
Any replica that has broadcast $\tuple{\NoVote, r}$ is prohibited from referencing $LV_r$ via a leader edge (Line~\ref{line:novote}), as a $\NoVote$ represents a commitment to withhold its vote for $LV_r$.

To create a vertex $v$ for round $r$, \node{i} first include a block of transactions. Crucially, $v$ must reference the most recently delivered leader vertex via a \textit{leader edge} for which \node{i} has not previously broadcast a $\NoVote$ message (Line~\ref{line:leader_edge}). This state is tracked using the $lastLeader$ variable (Line~\ref{line:last_leader}). Specifically, if $v$ references $LV_{r-1}$ via leader edge, it constitutes a ``vote" for $LV_{r-1}$. The vote can be used to commit $LV_{r-1}$, as we will elaborate shortly. If $v$ references a leader vertex from a round prior to $r-1$, the leader edge serves solely to establish the leader path to subsequent leader vertices. 

The leader vertex $LV_r$ is required to include at least $2f+1$ \textit{strong edges} referencing vertices from round $r-1$ (Line~\ref{line:strong_edge}). To ensure safety, $LV_r$ must either reference $LV_{r-1}$ via a strong edge or provide $NVC_{r-1}$ to prove that $LV_{r-1}$ failed to obtain sufficient votes for being committed (Line~\ref{line:no_vote_cert}). In addition, $LV_r$ may reference vertices from rounds prior to $r-1$ via \textit{weak edges} to ensure their eventual inclusion in the DAG (Line~\ref{line:weak_edge}). We prohibit non-leader vertices from setting weak edges to maintain their size. This results in a sparser DAG, implying that straggling vertices may experience increased latency before being included in the causal history of a leader vertex. To mitigate this issue, we permit a replica to reference a constant number of delivered vertices created by itself via \textit{self edges} (Line~\ref{line:self_edge}). This ensures that once the new vertex is referenced, these historical vertices are included in the causal history. These vertices are tracked in the $oldVertices$ set and are removed upon being referenced. Note that \textit{self edges} serve solely as a practical optimization and are not required for the protocol's correctness.

\noindentparagraph{\textbf{Round synchronization.}}
To ensure liveness, partially synchronous DAG-based protocols require that, all honest replicas enter the same round within a bounded time~\cite{malkhi2024bbca, shrestha2025sailfish, babel2025mysticeti} after GST. We refer to this property as \textit{round synchronization}. To achieve this, we must enable straggling replicas to advance directly to the latest round, a mechanism we term \textit{round-skipping}.

In \Cref{alg:dag_construction}, we present two distinct rules for round-skipping. The first rule is the same as Sailfish~\cite{shrestha2025sailfish}. When a replica \node{i} delivers $2f+1$ vertices from round $r$, and has delivered $LV_r$ or received $NVC_r$, \node{i} can directly enter round $r+1$ (Line~\ref{line:advance}). The second rule is slightly different. When \node{i} receives the \textit{first message} of the RBC for round $r+1$ from at least $f+1$ distinct replicas, and either delivered $LV_r$ or reveived $NVC_r$, \node{i} can skip to round $r+1$ (Line~\ref{line:jumping}). The intuition is that among the earliest set of $f+1$ first messages, at least one must be sent by an honest replica. This honest replica must have necessarily delivered $2f+1$ vertices from round $r$ and possesses either $LV_r$ or $NVC_r$. Given that non-leader vertices in \name are only required to provide a leader edge, they can proceed to round $r+1$ immediately. This rule allows slow replicas to catch up faster, and in particular, it significantly reduces the timeout duration in the CBC-based version (see \Cref{ss:protocol_CBC}).

Fast-votes are triggered after round-skipping (Lines~\ref{line:fast_vote_1}, \ref{line:fast_vote_2}). Although a replica skipping from round $r'$ to $r$ does not propose vertices for the intermediate rounds, it needs to create and broadcast fast-votes for these rounds when necessary. Specifically, after skipping rounds, the replica starts a timer of duration $\tau'$ (Line~\ref{line:tau'}). For any leader vertex of a skipped round, the replica must await its delivery prior to the timeout, unless the corresponding $NVC$ has already been received or $\NoVote$ has been sent. If the replica delivers a $LV$ within this duration, it broadcasts a fast-vote for that leader (Lines~\ref{line:fv_condition}--\ref{line:broadcast_fv}). A fast-vote is also treated as a vote for the leader vertex and is utilized in the commit process. This mechanism is pivotal for guaranteeing liveness, as discussed later.

\noindentparagraph{\textbf{DAG commit rule.}}
The commit rule of \name is presented in \Cref{alg:commit_rule}. Only leader vertices are explicitly committed. Non-leader vertices appearing in the causal history of a committed leader are ordered according to a specified deterministic rule (Line~\ref{line:order}).

\begin{algorithm}[h]
    \caption{\name's commit rule for replica \node{i}}
    \label{alg:commit_rule}
    \begin{algorithmic}[1]
    \footnotesize
    \algrestore{break_R}
    \Statex \textbf{Local variables:}
    \StateX $committedRound \gets 0$
    \StateX $deliveredVertices \gets \{\}$
    \StateX $leaderStack \gets $ initialize empty stack
    
    \color{magenta}
    \Upon{receiving a set $\M$ of first messages for round $r+1$ vertices s.t.
    \Statex $\forall v'\in\M$: $(\exists v\in (v'.leaderEdge \cup v'.strongEdges) \land$ $v.round=r$
    \Statex $\land~v.source=L_r)$ $\land$ $(|\M|+|FV_i[r+1]|\ge 2f+1)$}
    \label{line:commit}
    \If {$committedRound < r$}
    \State $v \gets$ \Call{get\_leader\_vertex}{$r$}
    \State \Call{commit\_leader}{$v$}
    \EndIf
    \EndUpon
    \normalcolor

    \Procedure{commit\_leader}{$v$}
    \State $leaderStack.$\Call{push}{$v$}
    \State $r \gets v.round - 1$
    \State $v' \gets v$
    \While{$r > committedRound$}
    \label{line:indirect_start}
    \State $v_s \gets$ \Call{get\_leader\_vertex}{$r$}
    \If{{\color{magenta}\Call{leader\_path}{$v', v_s$}}}
    \State $leaderStack.$\Call{push}{$v_s$}
    \State $v' \gets v_s$
    \EndIf
    \State $r \gets r -1$
    \EndWhile
    \State $committedRound \gets v.round$
    \label{line:indirect_end}
    \State \Call{order\_vertices}{$ $}
    \EndProcedure

    \Procedure{order\_vertices}{$ $}
    \label{line:order}
    \While {$\neg leaderStack.$\Call{isEmpty}{$ $}}
    \State $v \gets leaderStack.$\Call{pop}{$ $}
    \State $toDeliver \gets \{v' \in \bigcup_{r>0}DAG_i[r]\, | \, path(v, v') \land$ 
    \StateXXX $v' \notin deliveredVertices\}$
    \label{line:start_traversal}
    \For {\textbf{every} $v'\in toDeliver$ in some deterministic order}
    \State \textbf{output} \Call{a\_deliver$_i$}{$v'.block, v'.round, v'.source$}
    \State $deliveredVertices \gets deliveredVertices \cup \{v'\}$
    \label{line:end_traversal}
    \EndFor
    \EndWhile
    \EndProcedure

    \algstore{break_R}
    \end{algorithmic}
\end{algorithm}

When \node{i} receives at least $2f+1$ ``votes'' for the leader vertex of round $r$, it directly commits $LV_r$. A valid vote comprises either the \textit{first message} (of the RBC) for a round $r+1$ vertex that references $LV_r$, or a \textit{fast-vote} from round $r+1$ (Line~\ref{line:commit}). Since \node{i} may deliver a vertex prior to its corresponding first message, we also recognize the delivered vertex as a valid vote. The intuition behind the threshold of $2f+1$ votes is that at least $f+1$ of these votes must originate from honest replicas. This sufficiently precludes the existence of $NVC_r$. Together with the constraints imposed on leader vertices, this guarantees that any subsequent leader vertex will establish a \textit{leader path} to the committed $LV_r$.

Prior to directly committing a leader vertex, \node{i} recursively commits all uncommitted leader vertices from earlier rounds that are connected via a leader path (Lines~\ref{line:indirect_start}--\ref{line:indirect_end}). This indirect committing rule is essential for preserving total order property.

\noindentparagraph{\textbf{The Necessity of fast-votes.}}
We illustrate how fast-votes resolve the liveness issue through the following example (\Cref{fig:fast_vote} in Appendix~\ref{a:illustration}). Consider an execution where $f+1$ honest replicas and $f$ Byzantine replicas have advanced to round $r+1$, while the remaining $f$ honest replicas lag behind at round $r' \leq r$. In \name, Byzantine replicas can deliberately omit references to $LV_r$ in their round $r+1$ vertices. Once the $f$ Byzantine replicas and the $f+1$ fast honest replicas complete the RBC for round $r+1$, they proceed to round $r+2$. At this time, if the Byzantine replicas help deliver these vertices to the lagging replicas, these stragglers will immediately skip to round $r+2$. Without fast-votes, $LV_r$ would fail to attain the required $2f+1$ votes for being committed. This process could repeat forever, preventing any leader vertex from being directly committed even after GST, thereby violating liveness. The fast-vote mechanism in \name allows the $f$ lagging honest replicas to cast votes for $LV_r$ in such scenarios. Together with our commit rule, this ensures that $LV_r$ receives sufficient votes to be committed.

\subsection{Correctness and Efficiency Analysis}\label{ss:analysis}

\noindentparagraph{\textbf{Correctness.}}
Due to space limitations, the full correctness proof of \name is deferred to Appendix~\ref{a:correctness_R}.

\noindentparagraph{\textbf{Communication complexity.}}
We analyze the per-round metadata communication complexity. \textit{Metadata} refers to the edges (references) in vertices and consensus messages such as $\NoVote$s and $\FastVote$s. The edges may take the form of indices, cryptographic hashes, or signatures. Under the assumption of threshold signatures, the broadcasting of consensus messages incurs at most $O(\lambda n^2)$ communication per round. In \name, only the leader vertex contains $O(n)$ edges and potentially a $NVC$, while each non-leader vertex contains only $O(1)$ edges. Consequently, the total number of edges per round is $O(n)$. By contrast, Sailfish have $O(n^2)$ edges per round. In both protocols, disseminating these edges through RBC or CBC dominates the overall communication complexity.

With Bracha's RBC~\cite{bracha1987asynchronous}, \name incurs $O(n^3\log n)$ communication complexity\footnote{The $\log n$ factor arises because each reference can be encoded as a replica index.} compared with $O(n^4)$ for Sailfish.
When using CCBRB~\cite{alhaddad2022balanced}, the communication complexity for both \name and Sailfish is $O(\lambda n^3)$. Although \name cannot utilize erasure coding to reduce the asymptotic complexity, the reduced number of references still saves both computation and communication overhead. When using latency-optimal RBC~\cite{abraham2021good}, \name incurs $O(\lambda n^3)$ communication complexity, compared with $O(n^4)$ for Sailfish. When using Narwhal's CBC~\cite{danezis2022narwhal}, \name incurs $O(\lambda n^2)$ communication complexity in the good case and $O(\lambda n^3)$ in the bad case, whereas the corresponding communication complexity of Sailfish is $O(\lambda n^3)$ and $O(\lambda n^4)$, respectively. 

\noindentparagraph{\textbf{Latency.}}
In \name and Sailfish, the commit latency for an honest leader vertex (after GST) is the latency of a single RBC instance plus $1\delta$. In good case, the commit latency for non-leader vertices incurs an additional RBC latency, as they need to be referenced by a subsequent leader vertex. Consequently, for a generic RBC protocol, these two latencies equal to $(k_1+1)\delta$ and $(2k_1+1)\delta$, respectively. 

We now evaluate the additional latency incurred by a Byzantine leader. In \name, the timeout parameter is configured as $\tau=(k_1+k_2)\Delta$. Consider a scenario where Byzantine $L_r$ remains silent and the first honest replica enters round $r$ at time $t$. By \cref{lemma:synchronization_R} and the responsiveness of the RBC protocol, all honest replicas are guaranteed to enter round $r$ within $k_2\delta$. Consequently, by time $t+\tau+(k_2+1)\delta$, all honest replicas will receive the $NVC_r$ and advance to round $r+1$. If $L_{r+1}$ is honest, all non-leader vertices from round $r-1$ referenced by $L_{r+1}$ will be committed within the subsequent $(k_1+1)\delta$. Therefore, the presence of a single Byzantine leader results in an increase in the commit latency for non-leader vertices of $\tau+(k_2+1)\delta+(k_1+1)\delta - [k_2\delta+(k_1+1)\delta] = \tau+\delta$\footnote{We subtract $k_2\delta$ here because the commit latency calculation in the good case does not account for the additional latency caused by round synchronization.
}. In comparison, Sailfish incurs an additional latency of $\tau+2\delta$ (for a single failure) or $\tau+k_2\Delta+2\delta$ (for consecutive failures), as it requires an extra round of $\Timeout$ message dissemination and collection.



\subsection{\name with Consistent Broadcast}\label{ss:protocol_CBC}

In this section, we present the CBC-based version of \name protocol. Specifically, we build on Narwhal's CBC~\cite{danezis2022narwhal}, which combines Signed Echo CBC~\cite{cachin2011introduction} with a randomized pull mechanism. Narwhal's CBC optimizes communication complexity in the good case by eliminating the all-to-all broadcast. The feasibility of this approach hinges on the existence of availability certificates, which moves data synchronization off the critical path~\cite{arun2025shoal++, danezis2022narwhal}.

To the best of our knowledge, existing proofs for Narwhal-based protocols often directly invoke RBC's property (~\Cref{property:RBC})~\cite{arun2025shoal++, yu2025angelfish, spiegelman2024shoal}, which do not align with the weaker properties provided by Narwhal's CBC. Our goal is to directly design a protocol whose security relies exclusively on the properties of Narwhal's CBC. 

The main challenge in designing the protocol is to guarantee round synchronization. We address this challenge through an explicit round synchronization mechanism. Specifically, upon entering a new round, a replica sends a set of certificates to the corresponding leader (which we call a $\NewRound$ message) to ensure that the leader can enter that round. The round-skipping rules must also be modified to avoid blocking round advancement. To prevent the timeout from becoming too long, we further optimize the protocol by delaying the start of the round timer. 

The details of protocol are presented in Appendix~\ref{a:protocol_CBC}.

\section{Multi-leader \name}\label{s:protocol_ml}

\begin{algorithm*}[!t]
    \caption{Multi-leader \name's pseudocode for replica \node{i}}
    \label{alg:protocol_ML}
    \begin{algorithmic}[1]
    \footnotesize
    \algrestore{break_R}

    \Statex \textbf{Local variables:}
    \StateX struct vertex $v$:
    \Comment{The struct of a vertex in the DAG}
    {\color{magenta}
    \StateXX $v.leaderEdge$ - a leader vertex in round $\leq v.round-1$ that represents a leader edge with highest index
    \StateXX $v.nvc$ - a $\NoVote$ certificate of a leader vertex in $v.round-1$
    \Comment{Only main leader vertices need to contain}
    \StateXX $v.auxEdges$ - a set of leader edges in round $< v.round-1$ that represent auxiliary leader edges
    \Comment{Only main leader vertices need to contain}
    \StateX $lastLeader$ - The most recent delivered leader vertex with the highest round and index
    }

    \vspace{-1em}
    \setlength{\columnsep}{25pt}
    \begin{multicols}{2}

    {\color{magenta}
    \Procedure{get\_leader}{$r, x$}
    \State \Return \Call{get\_vertex}{$\ML_r[x], r$} 
    \EndProcedure
    }

    \Procedure{leader\_path}{$v,u$}
    \State \Return exists a sequence of $k \in \N$, vertices $v_1,\ldots, v_k$ s.t.
    \StateXX $v_1 = v$, $v_k = u$, and $\forall j \in [2,..,k]$: $v_j \in \bigcup_{r\ge 1} DAG_i[r]~\land$
    \StateXX $(v_j \in  v_{j-1}.leaderEdge$ $\cup$ $v_{j-1}.strongEdges$ ${\color{magenta}\cup~v_{j-1}.auxEdges})$  
    \EndProcedure

    \Upon {$\Timeout$ in $round$}
    \If {$\nexists v'\in DAG_i[round]:v'.source=L_{round}$}
    \State broadcast $\tuple{\NoVote, {\color{magenta}L_{round}}, round}_i$
    \EndIf
    \EndUpon

    \Upon {\Call{r\_deliver$_i$}{$LV_r,r,L_r$} for $r>{\color{magenta}lastLeader.round}$} \label{line:receive_main}
    \If {haven't sent $\tuple{\NoVote,{\color{magenta}L_r},r}_i$ before}
    \label{line:novote_main}
    \State $lastLeader \gets LV_r$
    \label{line:last_leader_1}
    \EndIf
    \EndUpon

    \Procedure{create\_new\_vertex}{$r$}
    \State $v.round \gets r$
    \State $v.source \gets \node{i}$
    \State $v.block \gets blocksToPropose.\text{dequeue}$
    \State \Call{set\_leader\_edge}{$v$}
    \If {$\node{i} = L_r$}
    \State $v.strongEdges \gets DAG_i[r-1]$
    \If {$\exists v' \in DAG_i[r-1]: v'.source=L_{r-1}$}
    \label{line:start_exist_create}
    \For {$\ell \in \ML_{r-1}$}
    \If {$\nexists v' \in DAG_i[r-1]:v'.source=\ell$}
    \State $v.nvc\gets NVC_{r-1}^{\ell}$
    \State \textbf{break}
    \label{line:end_exist_create}
    \EndIf
    \EndFor
    \Else
    \label{line:start_nexist_create}
    \State $v.nvc \gets NVC_{r-1}$
    \color{magenta}
    \State $v^* \gets \arg\max_{v'\in DAG_i[r-1]}\{v'.leaderEdge\}$
    \label{line:start_highest}
    \State $(r^*, x^*) \gets (v^*.round,v^*.index)$
    \label{line:end_highest}
    \State $v.auxEdges \gets \{$\Call{get\_leader}{$r^*,x$} for $x\leq x^*\}$ \label{line:end_nexist_create}
    \EndIf
    \normalcolor
    \State \Call{set\_weak\_edges}{$v, r$}
    \EndIf
    \State \Call{set\_self\_edges}{$v$}
    \State \Return $v$
    \EndProcedure

    \Procedure{order\_vertices}{$ $}
    \label{line:order_ML}
    \While {$\neg leaderStack.$\Call{isEmpty}{$ $}}
    \State $\CMV \gets leaderStack.$\Call{pop}{$ $}
    \For {$v \in \CMV$}
    \Comment{iterate over $\CMV$ in order}
    \State $toDeliver \gets \{v' \in \bigcup_{r>0}DAG_i[r]\, | \, path(v, v')$ 
    \StateXXX $\land$ $v' \notin deliveredVertices\}$
    \For {\textbf{every} $v'\in toDeliver$ in some order}
    \State \textbf{output} \Call{a\_deliver$_i$}{$v'.block, v'.round, v'.source$}
    \State $deliveredVertices \gets$ $deliveredVertices \cup \{v'\}$ 
    \EndFor
    \EndFor
    \EndWhile
    \EndProcedure

    \columnbreak

    \Procedure {advance\_round}{$r$}
    \color{magenta}
    \For {$\ell\in \ML_{r-1}$}
    \Comment{iterative over $\ML_{r-1}$ in order}
    \If {$\nexists v' \in DAG_i[r-1]:v'.source=\ell$}
    \State send $\tuple{\NoVote, \ell, r-1}_i$ to $L_r$
    \label{line:novote_ML}
    \If {$lastLeader=LV_{r-1} \land \ell.index>2$}
    \State $lastLeader \gets$ \Call{get\_leader}{$r-1, \ell.index-1$}
    \label{line:last_leader_2}
    \EndIf
    \EndIf
    \EndFor
    \If {$\node{i} = L_{r}$  and $\exists v' \in DAG_i[r-1]: v'.source=L_{r-1}$}
    \label{line:start_exist}
    \State \textbf{wait until} $((\exists v' \in DAG_i[r-1]: v'.source=\ell)$ for $\forall \ell \in$ 
    \StateXX $\ML_{r-1}[:x])$ $\land$ $(NVC_{r-1}^{\ell'}$ is received for $\ell'=\ML_{r-1}[x+1])$ \label{line:end_exist}
    \EndIf
    \If {$\node{i} = L_r$ and $NVC_{r-1}$ is received}
    \label{line:start_nexist}
    \State $v^* \gets \arg\max_{v'\in DAG_i[r-1]}\{v'.leaderEdge\}$
    \State $(r^*,x^*)\gets (v^*.round, v^*.index)$
    \State \textbf{wait until} \Call{r\_deliver$_i$}{$*,r^*,\ell$} for $\forall \ell \in \ML_{r^*}[:x^*]$
    \label{line:end_nexist}
    \EndIf
    \normalcolor
    \State $round\gets r$; start $timer$ for $round$
    \State \Call{broadcast\_vertex}{$round$}
    \EndProcedure

    \Upon{receiving a set $\M$ of first messages for round $r+1$ vertices s.t.
    \StateX $\forall v'\in\M: (\exists v\in v'.leaderEdge$ $\cup$ $v'.strongEdges \land v.round=r)$
    \StateX $\land$ $(|\M|+|FV_i[r+1]|\ge 2f+1)$}
    \label{line:commit_main}
    \If {$committedRound < r$}
    \color{magenta}
    \State \Call{commit\_leaders}{$r$}
    \EndIf
    \EndUpon

    \Procedure{commit\_leaders}{$r$}
    \State $\CLS \gets [LV_r]$
    \label{line:CLS}
    \For {$x = 2$ to $|\ML_r|$}
    \State $v \gets$ \Call{get\_leader\_vertex}{$r,x$}
    \If {have received a set $\bigS$ of first messages for round $r+1$ vertices 
    \StateXX s.t. $(\forall v' \in \bigS: v'.leaderEdge.round=r$ $\land$ 
    \StateXX $v'.leaderEdge.index\ge x) \land (|\bigS|\ge 2f+1)$}
    \label{line:commit_secondary}
    \State $\CLS \gets \CLS \parallel v$
    \Else ~\textbf{break}
    \EndIf
    \EndFor
    \normalcolor
    \State $leaderStack.$\Call{push}{$\CLS$}; $v'\gets\CLS[0]$; $r'\gets v'.round-1$
    \label{line:start_indirect_commit}
    \While {$r'>committedRound$}
    \State $\CMV\gets[]$
    \label{line:CMV}
    \For {$x = 1$ to $|\ML_{r'}|$}
    \State $v \gets$ \Call{get\_leader}{$r',x$}
    \If {{\color{magenta}\Call{leader\_path}{$v',v$}}}
    \State $\CMV \gets \CMV \parallel v$
    \Else ~\textbf{break}
    \EndIf
    \EndFor
    \If {$\CMV \neq []$}
    \State $v'\gets \CMV[0]$
    \Comment{main leader vertex for round $r'$}
    \EndIf
    \State $leaderStack.$\Call{push}{$\CMV$}; $r' \gets r'-1$
    \label{line:end_indirect_commit}
    \EndWhile
    \State $committedRound \gets \CLS[0].round$
    \State \Call{order\_vertices}{$ $}
    \label{line:causal_history}
    \EndProcedure
    \end{multicols}
    \vspace{-2em}

    \end{algorithmic}
\end{algorithm*}

In \name, non-leader vertices incur at least one additional RBC latency over leader vertices. To further minimize the average latency, we extend \name to support multiple leaders per round. Our objective is to ensure that, in the good-case scenario, all leader vertices are committed within a latency of one RBC plus $1\delta$. This aligns with other state-of-the-art protocols~\cite{shrestha2025sailfish, yu2025angelfish, arun2025shoal++, babel2025mysticeti}. We designate this extended protocol as Multi-leader \name.

\noindentparagraph{\textbf{Multiple leaders with different types.}}
Drawing inspiration from Multi-leader Sailfish~\cite{shrestha2025sailfish}, we classify the leaders in each round into one \textit{main leader} and a set of \textit{secondary leaders}. The main leader is analogous to the leader in \name. It is responsible for establishing paths to potentially committed leader vertices or providing $NVC$s to certify that some are not committed. Secondary leaders are largely identical to non-leaders, with the exception that the new commit rule enables their vertices to be committed earlier.

We assume that the sequence of leaders for a given round is selected by a deterministic mechanism. Let $ML_r$ denote the leader sequence for round $r$. Specifically, $\ML_r[x]$ represents the $x$-th leader (main leader if $x=1$, secondary leader if $x>1$). Additionally, we define $\ML_r[:x]$ and $\ML_r[x+1:]$ to denote the first $x$ leaders and the set of all subsequent leaders in round $r$, respectively. For notational convenience, let $L_r = \ML_r[1]$ denote the main leader of round $r$ and $LV_r$ denote the main leader vertex. Let $\MLV_r[x]$ denote the delivered vertex corresponding to $\ML_r[x]$. We further define $\MLV_r[x].index = x$ (used in the pseudocode).

\noindentparagraph{\textbf{DAG construction protocol.}}
We begin by outlining the modifications in the DAG construction of Multi-leader \name relative to \name. The pseudocode for Multi-leader \name is presented in \Cref{alg:protocol_ML}. Text highlighted in magenta explicitly denotes the differences relative to both \name and multi-leader Sailfish~\cite{shrestha2025sailfish}.

Sharing the fundamental design philosophy of \name, Multi-leader \name aims to reduce the number of references in a non-main leader vertex to $O(1)$. However, the multi-leader setting necessitates that vertices cast ``votes'' for secondary leaders to directly commit them. To accomplish this, we extend the meaning of the \textit{leader edge}. Specifically, a leader edge references the ``highest'' leader vertex in the corresponding round. Here, the term ``highest'' implies that all leader vertices with lower indices have already been delivered (Lines~\ref{line:last_leader_1} and \ref{line:last_leader_2}). In essence, a leader edge referencing the $k$-th leader serves as a vote for all leaders with indices $1$ through $k$ in that round. This design circumvents the need for non-main leader vertices to include $O(n)$ references for voting, a requirement that would otherwise arise given the presence of $O(n)$ leaders.

Since a leader edge references only the highest leader vertex, the main leader needs to establish direct paths to other leader vertices from preceding rounds to guarantee safety. To facilitate this, we introduce \textit{auxiliary edges} within the main leader vertex (denoted as $auxEdges$ in \Cref{alg:protocol_ML}). We explain this design by considering the main leader $L_r$ of round $r$.

When $L_r$ enters round $r$ via the delivery of $LV_{r-1}$ (Lines~\ref{line:start_exist}--\ref{line:end_exist}), $auxEdges$ are omitted. In this scenario, we adopt the same constraints as in multi-leader Sailfish~\cite{shrestha2025sailfish}. Specifically, let $NVC_r^{\ell}$ denote the $\NoVote$ certificate for leader $\ell$ in round $r$. Additionally, we use $NVC_r$ to denote $NVC_r^{L_r}$. $LV_r$ is required to establish strong edges to all leader vertices corresponding to $\ML_{r-1}[:x]$ (for some $x>0$), and include the $NVC$ of $ML_{r-1}[x+1]$ (Lines~\ref{line:start_exist_create}--\ref{line:end_exist_create}). If all vertices of $ML_{r-1}$ are referenced, the inclusion of $NVC$ is omitted.

When $L_r$ enters round $r$ via $NVC_{r-1}$ (Lines~\ref{line:start_nexist}--\ref{line:end_nexist}), $L_r$ is responsible for establishing paths to leader vertices preceding round $r-1$. To achieve this, $L_r$ scans the leader edges of all vertices referenced by its strong edges and identifies the one with the highest round $r^*$ and subsequently the highest index $x^*$ (Lines~\ref{line:start_highest}--\ref{line:end_highest}). In the pseudocode, the $\arg\max$ notation succinctly captures this search process. Subsequently, $L_r$ references the vertices of $\ML_{r^*}[:x^*]$ via $auxEdges$ (Line~\ref{line:end_nexist_create}). By properties of RBC, the eventual delivery of these vertices is guaranteed.

Consistent with \name, a replica broadcasts a $\NoVote$ for $L_r$ if it fails to deliver $LV_r$ before timeout. Upon entering a new round, a replica sends $\NoVote$ messages for all undelivered leader vertices from the previous round to the main leader (Line~\ref{line:novote_ML}). This ensures that the main leader of the new round can successfully include an $NVC$ (if necessary) and create a valid vertex. We denote the $\NoVote$ message for leader $\ell$ in round $r$ as $\tuple{\NoVote, \ell, r}$.

\noindentparagraph{\textbf{DAG commit rule.}}
The primary distinction between the DAG commit rules of Multi-leader \name and \name stems from the interpretation of the leader edge. Specifically, a vertex in round $r+1$ referencing $\MLV_r[x]$ via its leader edge constitutes a vote for the entire set $\MLV_r[:x]$. For the main leader $LV_r$, any round $r+1$ vertices containing a leader edge to round $r$ is regarded as a vote (Line~\ref{line:commit_main}). After directly committing the main leader, replicas invoke commit\_leaders to attempt to directly commit secondary leaders within the same round. Given that fast-votes are issued exclusively for the main leader, we consider only first messages as votes for secondary leaders. $\MLV_r[x]$ is directly committed upon receiving at least $2f+1$ first messages from round $r+1$ with leader edge indices $\geq x$ (Line~\ref{line:commit_secondary}).

Upon completion of the aforementioned procedure, replicas recursively and indirectly commit all leader vertices $\MLV_{r'}[:x']$ (for some $x'>0$) in rounds $r' < r$ for which a leader path exists from $LV_r$ (Lines~\ref{line:start_indirect_commit}--\ref{line:end_indirect_commit}). Finally, replicas invoke order\_vertices to totally order the entire causal history in accordance with the sequence of committed leaders (Line~\ref{line:causal_history}). This indirect commit process leverages the existence of leader paths to ensure safety.

We present an illustration of Multi-leader \name in Appendix~\ref{a:illustration} and provide its correctness proof in Appendix~\ref{a:correctness_ML_R}.

\noindentparagraph{\textbf{Communication complexity.}}
In Multi-leader \name, the number of edges in non-main leader vertices remains unchanged. Since the number of auxiliary edges is at most $O(n)$, the main leader vertex still contain $O(n)$ edges. Within each round, all replicas are required to send up to $O(n)$ additional $\NoVote$ messages to the main leader, which results in a total communication overhead of $O(\lambda n^2)$. Consequently, across the various RBC implementations considered, Multi-leader \name maintains the same metadata communication complexity as \name.

\noindentparagraph{\textbf{Latency.}}
In the best-case scenario, all leader vertices are directly committed. Consequently, the commit latency for each leader vertex is the one $RBC$ plus $1\delta$. In a non-optimal case, the main leader vertex may require an additional $\delta$ to collect the $NVC$ of a previous secondary leader. Thus, the direct commit latency for leader vertices becomes one $RBC$ plus $2\delta$. Assuming there are $n-f$ delivered vertices per round, the average latency of Multi-leader \name outperforms that of \name as long as $x > \frac{n-f+k_1}{k_1}$, where $x$ denotes the number of directly committed leader vertices.

The CBC-based version of \name also supports a multi-leader extension. See Appendix~\ref{ss:multi_leader_C} for details.
\section{Evaluation}\label{s:eval}

We evaluate the performance of \name under varying numbers of replicas and different failure scenarios, and compare it against Sailfish~\cite{shrestha2025sailfish} and Sparse Bullshark~\cite{anoprenko2025dags}. Sailfish is the state-of-the-art certified-DAG protocol and serves as the basis of Clownfish. Sparse Bullshark is a scalable DAG protocol, which reduces the metadata communication overhead of Bullshark~\cite{spiegelman2022bullshark}. Our evaluation aims to demonstrate that (\rmnum{1}) \name provides better scalability and maintains strong performance at large numbers of replicas. (\rmnum{2}) \name achieves lower latency under failure case. (\rmnum{3}) Multi-leader \name can further reduce average latency while preserving scalability.

Our empirical evaluation is based on two complementary implementations. First, we use a Rust-based simulator to study protocol scalability at large system sizes. The simulator abstracts away implementation-specific overheads, allowing us to focus on metadata communication and capture the scalability gap between different protocols. Second, we build a Go-based prototype and deploy it on geo-distributed servers to measure real throughput-latency performance. The deployment results demonstrate the practical benefits of reducing metadata communication under medium-scale systems and bandwidth-limited settings.

\begin{figure*}[t]
    \centering
    \begin{subfigure}[t]{0.33\textwidth}
        \centering
        \includegraphics[width=\linewidth]{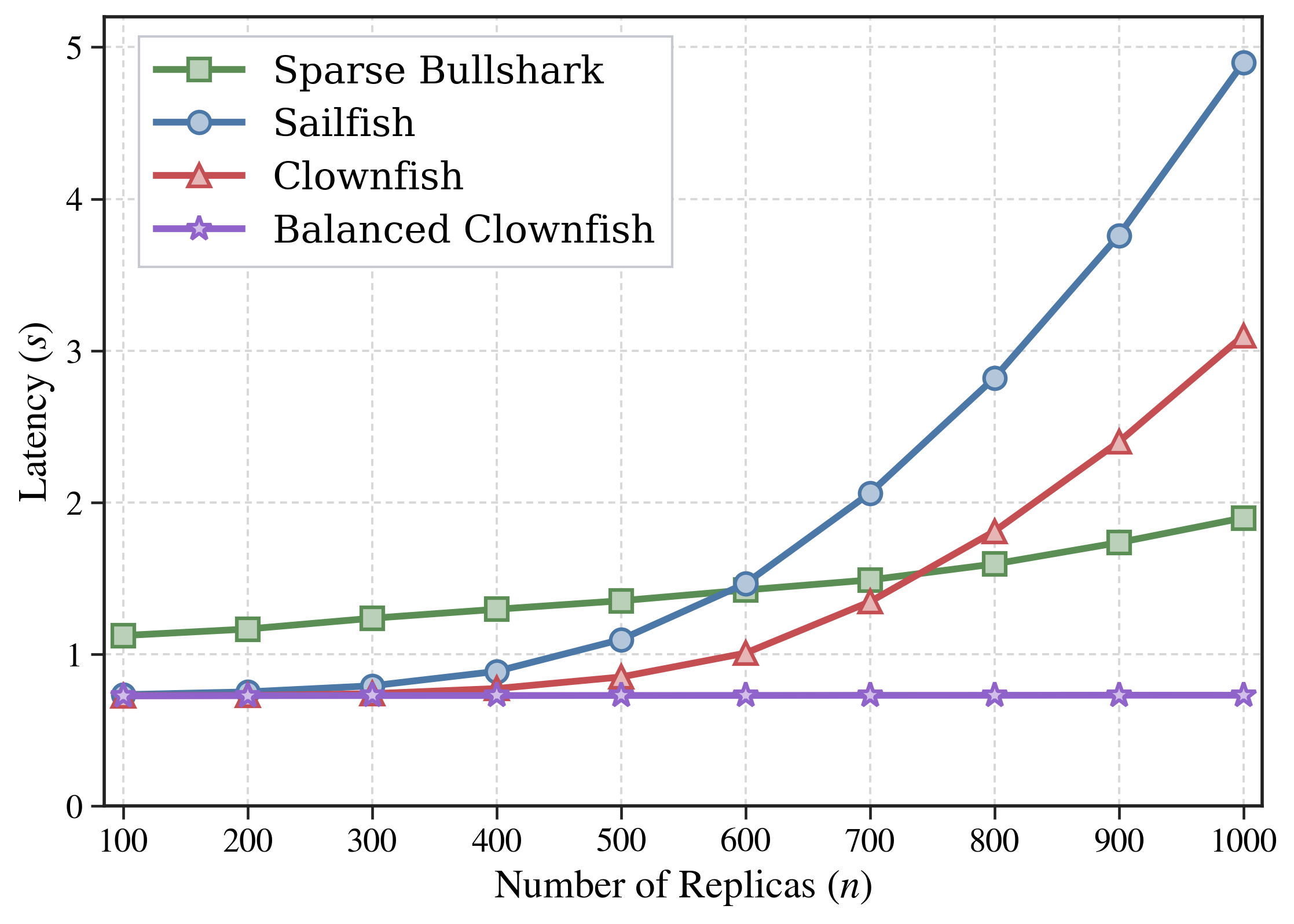}
        \caption{Latency}
        \label{fig:single_latency}
    \end{subfigure}
    \hfill
    \begin{subfigure}[t]{0.33\textwidth}
        \centering
        \includegraphics[width=\linewidth]{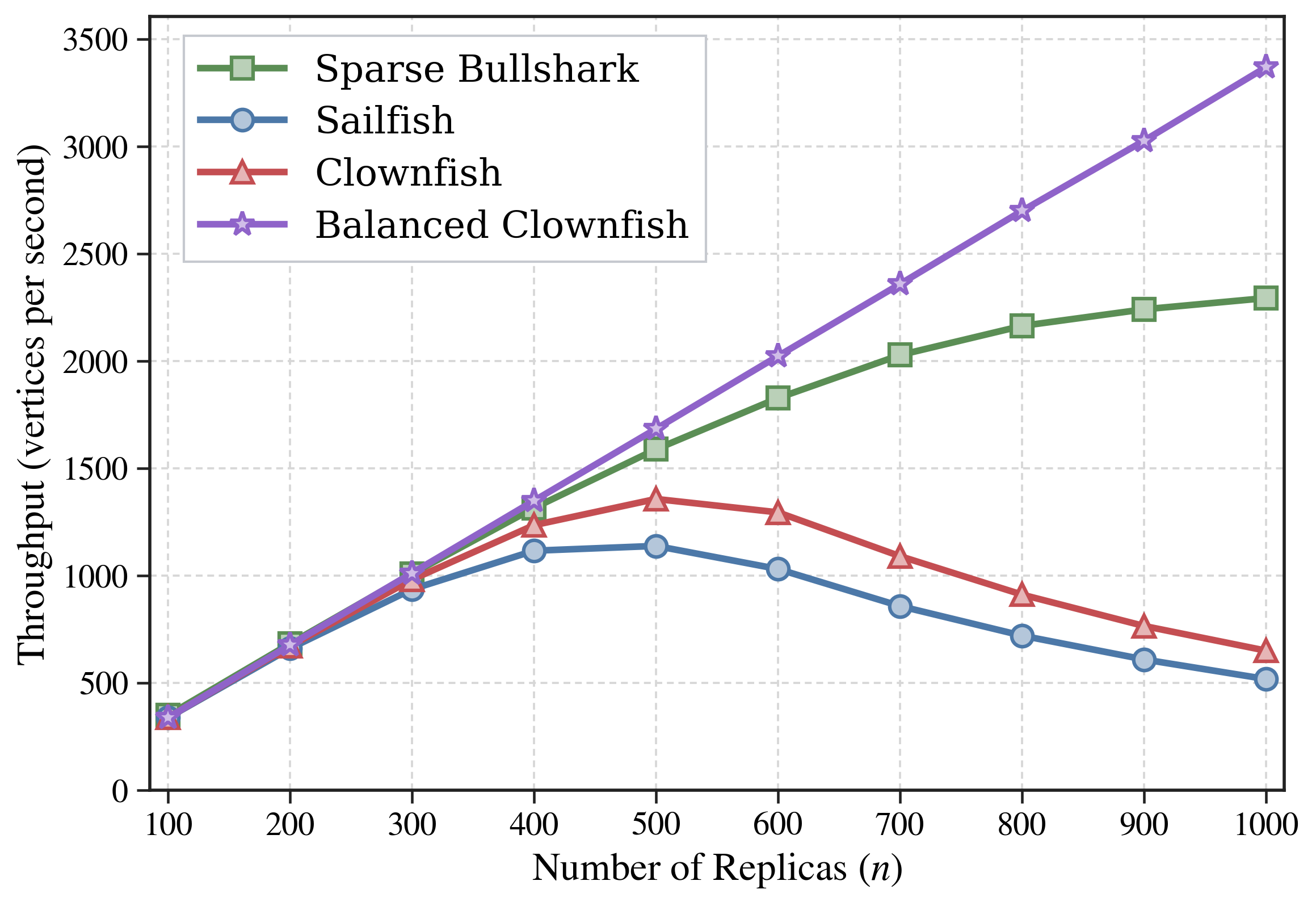}
        \caption{Throughput}
        \label{fig:single_throughput}
    \end{subfigure}
    \hfill
    \begin{subfigure}[t]{0.33\textwidth}
        \centering
        \includegraphics[width=\linewidth]{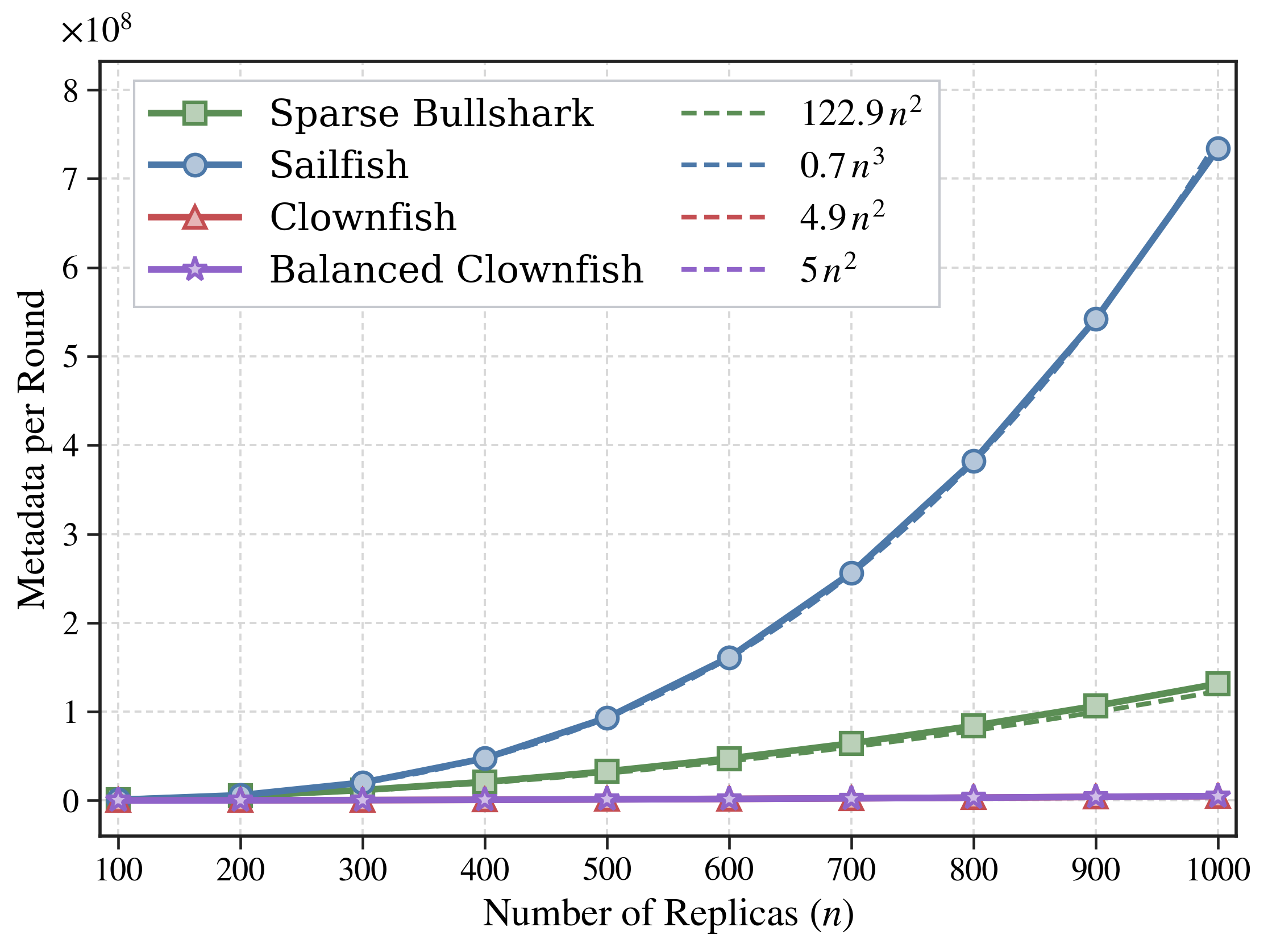}
        \caption{Metadata communication}
        \label{fig:single_metadata}
    \end{subfigure}

    \caption{Performance comparison under failure-free case with varying numbers of replicas.}
    \label{fig:scalability}
\end{figure*}

\noindentparagraph{\textbf{Implementation details.}}
For simulation, we modify the Rust-based discrete-event simulator used by Sparse Bullshark~\cite{anoprenko2025dags} to implement \name and Multi-leader \name. Consistent with the original implementations of Sailfish and Sparse Bullshark, we use Narwhal's CBC~\cite{danezis2022narwhal} to minimize communication overhead. We focus on metadata communication and only consider consensus over empty blocks without transactions. To ensure a fair comparison, we also re-implement Sailfish, Multi-leader Sailfish, and Sparse Bullshark in the simulator according to their protocol descriptions\footnote{We set the sample size (number of strong edges) of Sparse Bullshark to 128 to achieve safety comparable to that of Sailfish and \name, while preserving sparsity.}.

For deployment, we implement the above protocols in Go\footnote{Since our deployment considers medium-scale systems with $n<\lambda$, we use Bullshark instead of Sparse Bullshark, as Sparse Bullshark cannot be sparsified in this regime without affecting safety.}. We implement the primary-worker architecture of Narwhal~\cite{danezis2022narwhal}. Specifically, workers are responsible for broadcasting transactions from clients in batches. After receiving $n-f$ acknowledgments from other workers, a worker provides the digest of the corresponding batch to its primary. The primary then uses these digests as payloads in the consensus protocol.



\noindentparagraph{\textbf{Experimental setup.}}
For simulation, we set the maximum number of replicas to 1000. To capture the bandwidth bottlenecks that arise from large volumes of metadata transmission, we limit the network bandwidth of each replica to 1Gbps. To make latency more realistic, we simulate replicas evenly distributed across five geo-distributed regions. We use historical RTT measurements from Google Cloud~\cite{GoogleCloudPerfDashboard} and model the latency of each directed region pair with a normal distribution. We measure latency as the average time from vertex creation to its commitment. Since transactions are not modeled in the simulator, throughput is measured as the average number of committed vertices per second, which should ideally scale linearly with the number of replicas. Each simulation runs for 180 simulated seconds. 

We conduct the deployment experiments on geo-distributed servers across five AWS regions: us-east-1, us-west-1, eu-west-2, ap-northeast-1, and ap-southeast-2. Replicas are evenly distributed across these regions. The round-trip time between different regions ranges from 59ms to 264ms. Each replica runs on an AWS EC2 \texttt{c5a.2xlarge} instance with 8 vCPUs and 16 GB of memory, running Ubuntu 22.04. We limit the bandwidth of each machine to 100Mbps to emulate settings where communication is affected by metadata transmission. Each replica consists of one primary and one worker, which share the same bandwidth limit. The worker generates dummy transactions of 256 random bytes according to the specified rate. Each experiment runs for 30 seconds after several warm-up rounds with 50 replicas. The timeout parameter is set to 2.5 seconds. End-to-end latency is measured as the average time from transaction creation to commitment. Throughput is measured as the number of committed transactions per second.

\noindentparagraph{\textbf{Performance of \name under failure-free case.}}
We first compare the scalability of \name against other protocols under the common case where there are no failures. Since the leader in vanilla \name can still become a bottleneck due to the large size of leader vertices, we optimize it using Balanced Multicast from~\cite{alhaddad2022balanced}. Specifically, we replace the first broadcast of each leader vertex with Balanced Multicast, which uses error-correcting codes to distribute the leader's sending load evenly across all replicas. Notably, this optimization does not reduce the total communication overhead and may introduce additional message delay. It is therefore ineffective for Sparse Bullshark and Sailfish, where the sending load is already balanced because all vertices have the same metadata size.

In simulation, we vary the number of replicas from 100 to 1000 and measure the latency, throughput, and metadata volume per round of each protocol. The results are shown in \Cref{fig:scalability}, where \name and Balanced \name denote vanilla \name and \name optimized with Balanced Multicast, respectively.

In Figure~\ref{fig:scalability}\subref{fig:single_latency}, all protocols experience higher latency as the number of replicas increases. When $n\le300$, \name and Sailfish have nearly identical latency and both outperform Sparse Bullshark. However, Sailfish's latency increases rapidly once $n>300$, as its cubic metadata communication causes the bandwidth bottleneck to appear earlier. Vanilla \name delays this bottleneck until $n>500$ and suffers a smaller increase, since only the leader carries a large vertex. Nevertheless, this leader bottleneck eventually makes vanilla \name slower than Sparse Bullshark when $n>700$, despite its lower metadata complexity. Balanced \name removes this bottleneck by distributing the leader's sending load, and therefore achieves the best scalability: its latency remains almost unchanged up to $n=1000$ (from 724ms to 729ms). 

The throughput results in Figure~\ref{fig:scalability}\subref{fig:single_throughput} show the same trend. Sailfish and vanilla \name first benefit from having more vertices per round, but their throughput eventually drops as metadata transmission saturates the network. Sparse Bullshark also begins to reach a turning point at $n=1000$. In contrast, Balanced \name continues to scale almost linearly with the number of replicas. 

\begin{figure}[h]
    \centering
    \includegraphics[width=0.85\linewidth]{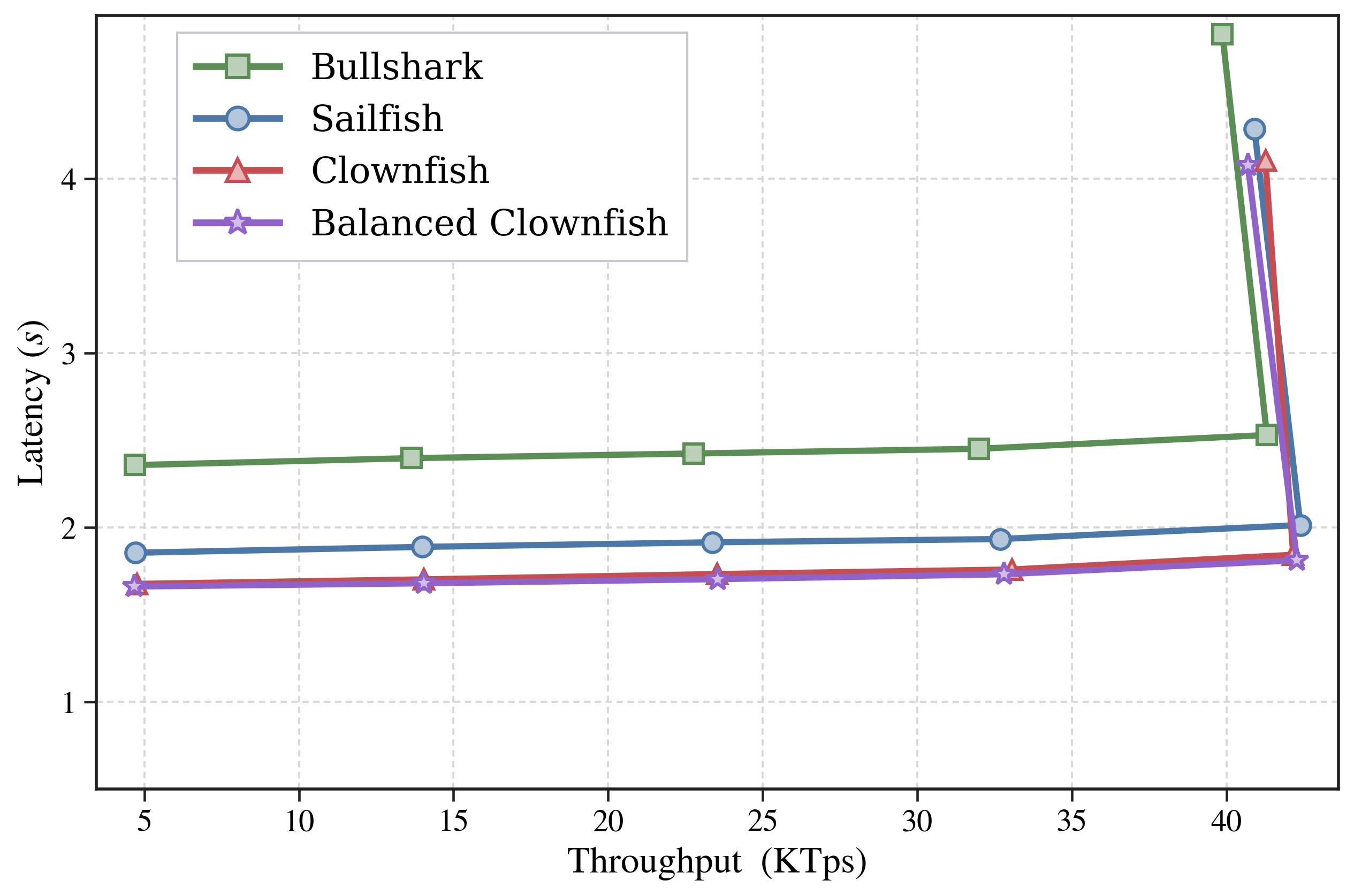}
    \caption{Throughput / latency at $n=50$ without failures.}
    \label{fig:single_l}
\end{figure}

\begin{figure}[h]
    \centering
    \includegraphics[width=0.85\linewidth]{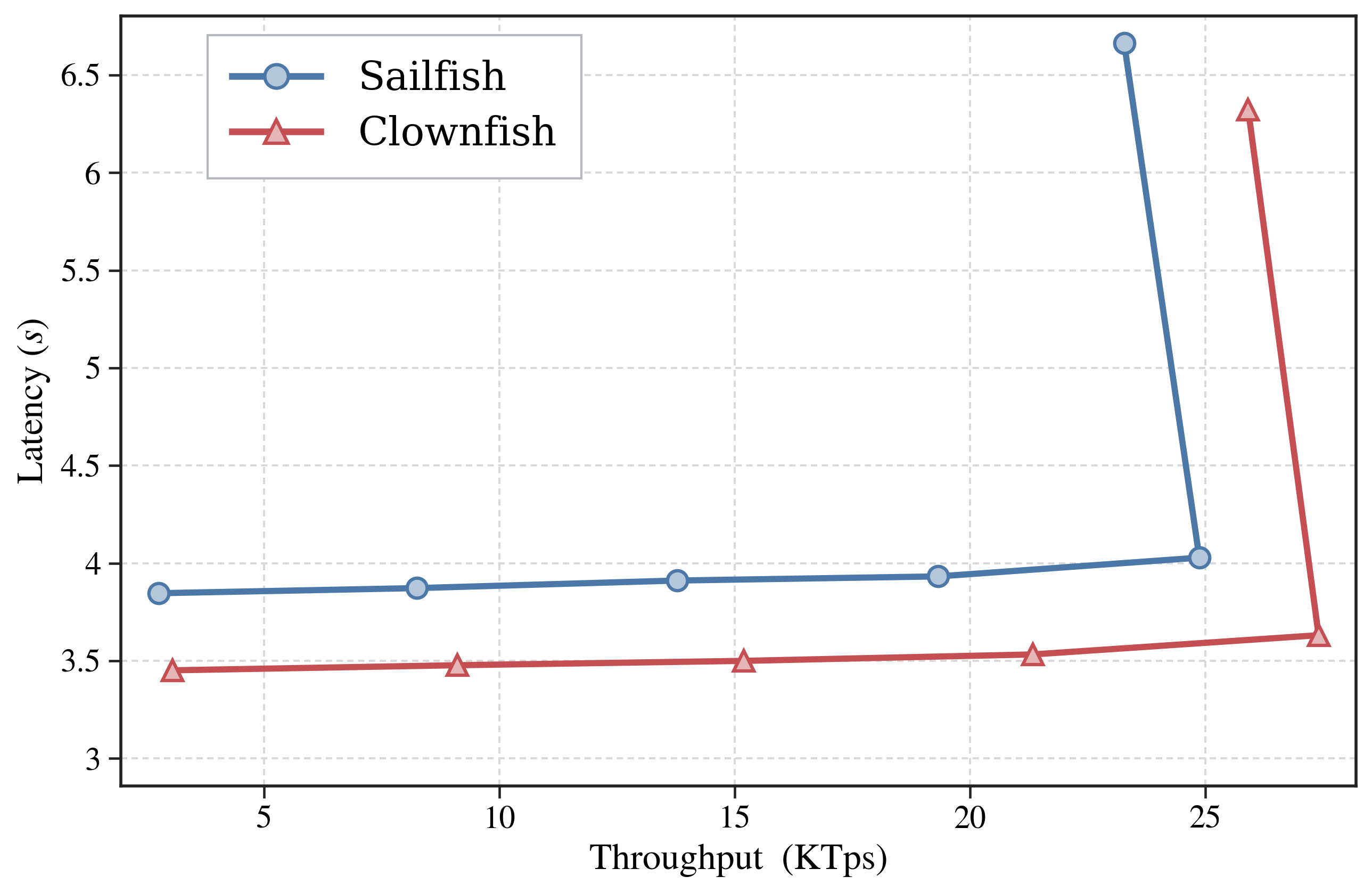}
    \caption{Throughput / latency at $n=50$ with 16 crash failures.}
    \label{fig:crash_l}
\end{figure}

Figure~\ref{fig:scalability}\subref{fig:single_metadata}  summarizes the per-round metadata communication of different protocols, where each signature is counted as one unit of metadata. The dashed lines show polynomial fits of appropriate degree. Consistent with the theoretical analysis, Sailfish exhibits cubic communication complexity, while Sparse Bullshark and \name both exhibit quadratic communication complexity, with Sparse Bullshark incurring an additional $\Theta(\lambda)$ factor (here $\lambda=128$). Due to the use of error-correcting codes, Balanced \name incurs slightly higher communication overhead than vanilla \name.

We further evaluate the throughput-latency performance of the protocols in deployment, with the results shown in \Cref{fig:single_l}. Under the given bandwidth limit, all protocols achieve comparable throughput in the tens of thousands, while \name and Balanced \name achieve similar latency, both lower than Sailfish and Bullshark. The latency gap of about 10\% (180ms) between \name and Sailfish comes from the saved metadata processing and communication overhead. These results show that metadata reduction provides practical latency benefits even at medium system sizes. This advantage is expected to become more pronounced as the ratio between metadata communication volume and available bandwidth increases, and can extend to larger-scale systems and higher-bandwidth settings.

\noindentparagraph{\textbf{Performance of \name under crash failures.}}
We evaluate the throughput-latency performance of \name and Sailfish in the presence of $f=\left\lfloor \frac{n-1}{3} \right\rfloor$ crash failures. We distribute the crashed replicas evenly across the five regions. We use round-robin leader rotation, where a crash leader appears every three rounds.

As depicted in Figure~\ref{fig:crash_l}, the average latency of both protocols is substantially higher than in the failure-free case, due to the need of waiting for all honest replicas as well as the timeout caused by crash leaders. Before saturation, Sailfish's latency stays around 3850ms to 4050ms, whereas \name maintains a latency of around 3450ms to 3650ms. The enlarged latency gap comes from \name's optimized round-advancement rule: replicas in \name only need one round of message dissemination after a timeout, whereas Sailfish requires two sequential rounds. As a result, non-leader vertices in \name incur one less additional message delay.

\begin{figure}[h]
    \centering

    \begin{subfigure}[t]{0.8\linewidth}
        \centering
        \includegraphics[width=\linewidth]{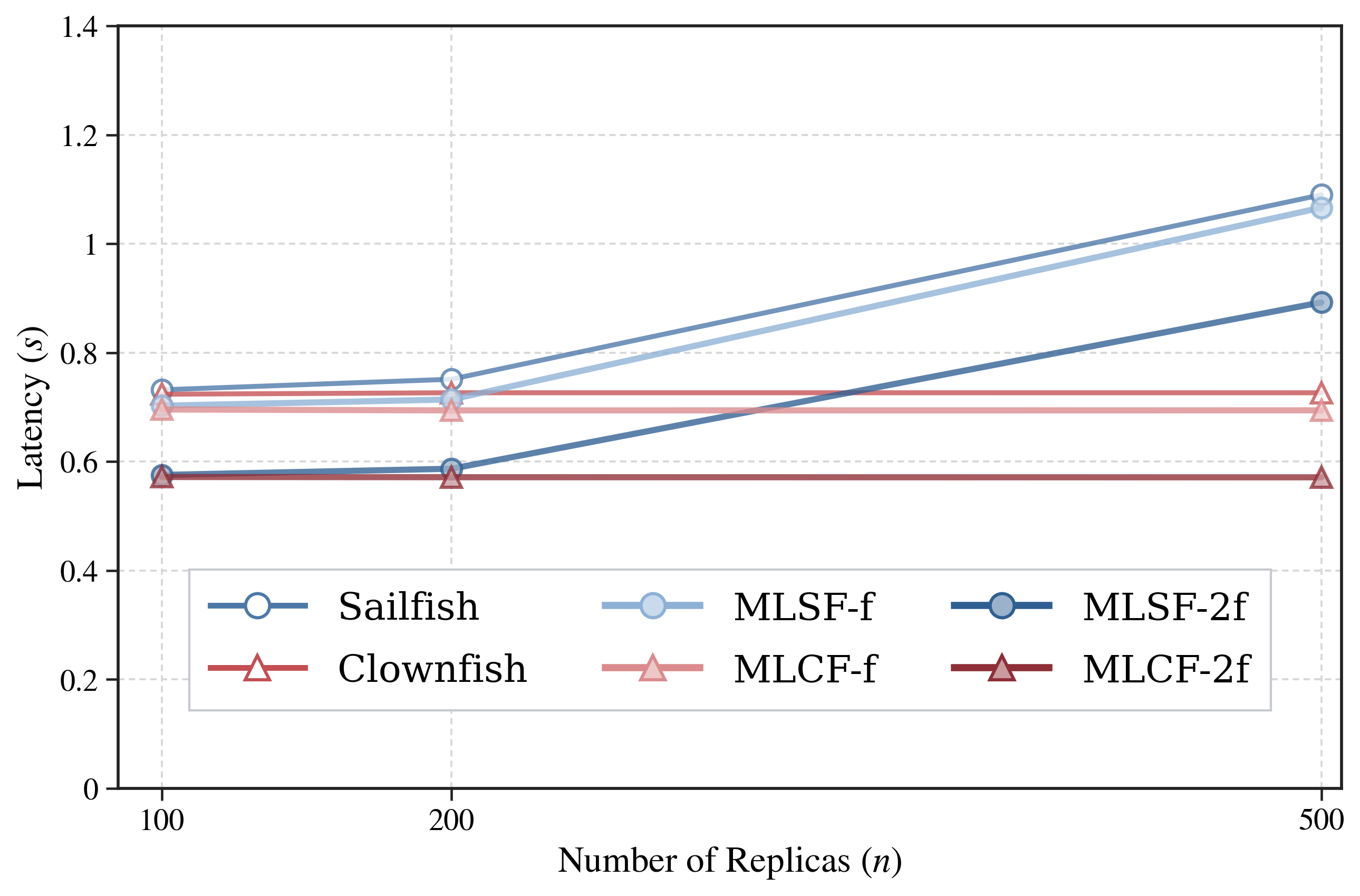}
        \caption{Latency result of simulation.}
        \label{fig:multi_sim}
    \end{subfigure}

    \begin{subfigure}[t]{0.8\linewidth}
        \centering
        \includegraphics[width=\linewidth]{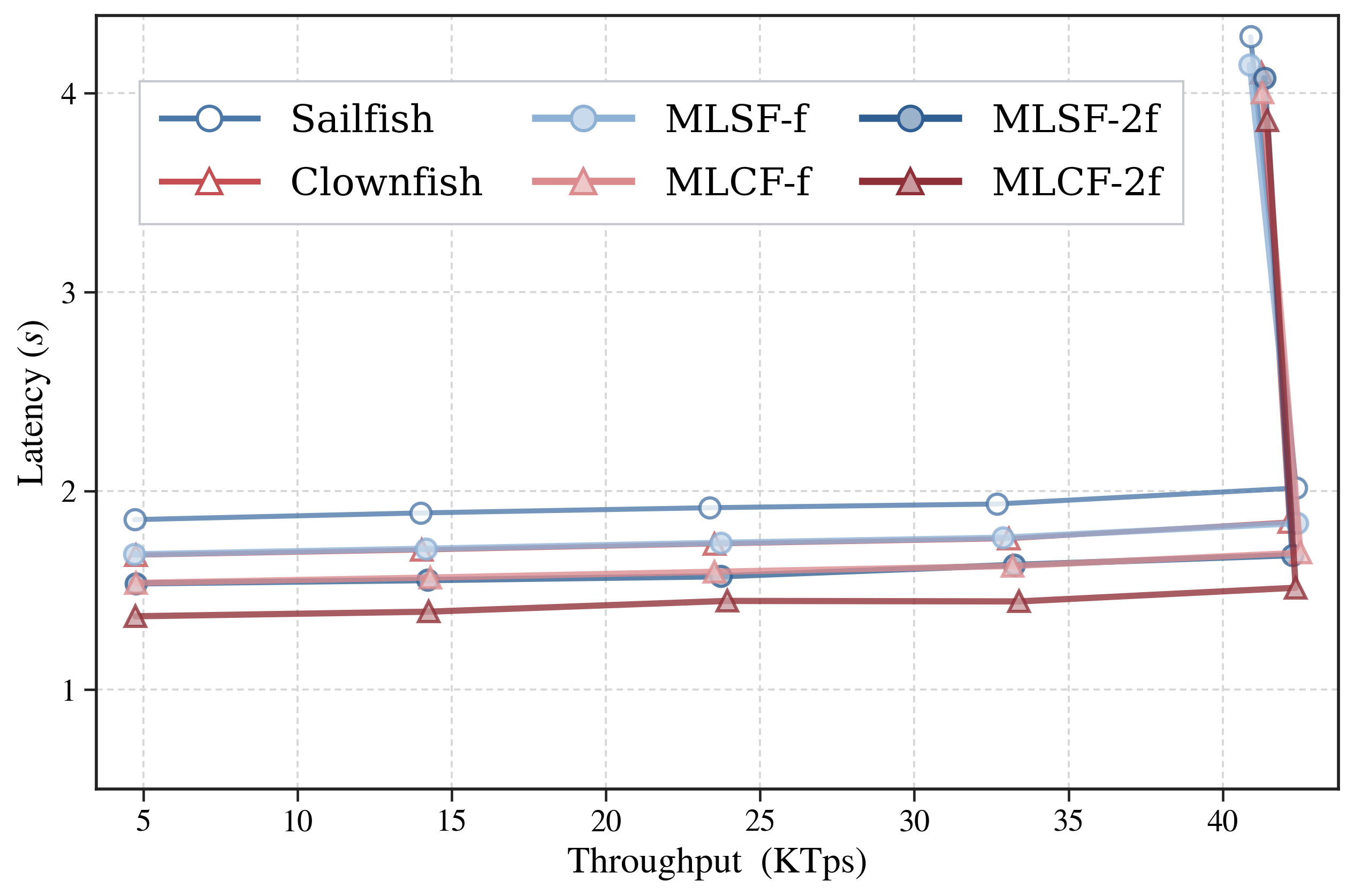}
        \caption{Throughput / latency result of deployment.}
        \label{fig:multi_deploy}
    \end{subfigure}

    \caption{Performance of Multi-leader protocols at $n=50$ without failures.}
    \label{fig:multi_l}
\end{figure}

\noindentparagraph{\textbf{Performance of Multi-leader \name under failure-free case.}}
We compare the performance of Multi-leader \name and Multi-leader Sailfish in the failure-free case (Sparse Bullshark does not support multiple leaders). We adopt the same implementation strategy as Multi-leader Sailfish~\cite{shrestha2025sailfish}, where replicas wait for all leader vertices before entering the next round. Figure~\ref{fig:multi_l}\subref{fig:multi_sim} shows the simulation latency under different numbers of replicas and leaders, where MLCF denotes Multi-leader \name and MLSF denotes Multi-leader Sailfish. Figure~\ref{fig:multi_l}\subref{fig:multi_deploy} shows the throughput-latency performance of these protocols in deployment.

As shown in the figures, Multi-leader \name reduces latency in a manner similar to Multi-leader Sailfish, with larger improvements as the number of leaders increases. More importantly, Multi-leader \name preserves the scalability advantage of \name: its latency remains relatively stable as the number of replicas grows. The deployment results further show that Multi-leader \name reduces average latency in practice while retaining the advantage brought by lower communication overhead.
\section{Discussion}\label{s:discussion}




In \name, leader vertices carry more metadata compared to non-leader vertices. To balance the workload, one could employ a load-balanced broadcast specifically for leader vertices~\cite{alhaddad2022balanced}, as in our implementation. Another orthogonal optimization involves utilizing a leader reputation mechanism~\cite{cohen2022aware, tsimos2024hammerhead} to select the fastest leaders. Since \name requires leaders to perform more work compared to other replicas, the benefits of such an approach would be particularly pronounced. Finally, an interesting direction is extending \name to a signature-free variant. Sailfish++~\cite{shrestha2025optimistic} provides a signature-free implementation of Sailfish; their methodology could be directly adapted to \name. We also aim to explore more efficient and concise implementation approaches in future work.

 
\section{Related Work}\label{s:related_work}

\noindentparagraph{\textbf{Leader-based BFT.}}
Numerous BFT protocols over the past two decades have adopted the leader-based paradigm~\cite{castro1999practical, kotla2007zyzzyva, buchman2018latest, yin2019hotstuff, chan2023simplex}. This paradigm relies on a leader to disseminate data and drive consensus progress, utilizing a view-change mechanism to periodically rotate the leader. Building upon the optimal latency achieved by PBFT~\cite{castro1999practical} and the linear complexity achieved by HotStuff~\cite{yin2019hotstuff}, subsequent protocols aim to achieve a superior latency-communication tradeoff~\cite{gelashvili2022jolteon, malkhi2023hotstuff}, as well as to enhance performance under specific scenarios~\cite{jalalzai2023fast, gueta2019sbft, giridharan2023beegees}. However, the primary drawbacks of leader-based BFT are the complexity of view change~\cite{naor2021cogsworth, bravo2022making} and the throughput bottleneck imposed by a single leader.

\noindentparagraph{\textbf{Towards high throughput BFT.}}
To improve throughput, it is crucial to utilize the available network bandwidth for data transmission. One approach involves running multiple leader-based BFT instances in parallel~\cite{stathakopoulou2019mir, stathakopoulou2022state, gupta2021rcc}. Since distinct replicas can serve as leaders across different instances and broadcast data concurrently, this results in a substantial throughput improvement in good case. However, this method is susceptible to stragglers and necessitates complex coordination for combining outputs and view changes~\cite{lyu2025ladon}.

Another strategy involves decoupling data transmission from consensus logic~\cite{yang2022dispersedledger, giridharan2024autobahn, gao2022dumbo, danezis2022narwhal}. This approach recognizes that consensus serves solely as an ordering mechanism, whereas data availability can be achieved in a fully parallel and asynchronous manner. Consequently, replicas can allocate bandwidth to the underlying data dissemination layer while executing the consensus protocol. 

\noindentparagraph{\textbf{DAG-based BFT.}}
DAG-based protocols offer a natural approach to achieving high throughput. In this paradigm, all replicas concurrently broadcast vertices containing both data payloads and references, collectively constructing a DAG structure. By interpreting specific references as votes and leveraging path relationships, replicas can achieve a total ordering of the data without incurring additional communication overhead.

DAG-based BFT protocols are initially designed for asynchronous networks. Hashgraph~\cite{baird2016swirlds} and Aleph~\cite{gkagol2019aleph} introduce unstructured and structured DAG constructions, respectively. Subsequently, a series of protocols (DAG-Rider~\cite{keidar2021all}, Tusk~\cite{danezis2022narwhal}, and Bullshark~\cite{spiegelman2022bullshark}) introduce more efficient designs, serving as templates for subsequent DAG-based protocols. More recently, works such as GradedDAG~\cite{dai2023gradeddag} and MAHI-MAHI~\cite{jovanovic2025mahi} have further reduced the latency of asynchronous DAG protocols by employing different broadcast primitives and commit rules. To circumvent the FLP impossibility result~\cite{fischer1985impossibility}, these protocols necessitate the use of a common coin for retrospective leader election.

Starting with Bullshark~\cite{spiegelman2022bullshark_p}, DAG-based protocols begin to focus on efficient designs under partial synchrony networks. Bullshark introduces designated leaders and timeout to provide better performance during periods of synchrony. Since Bullshark designates a leader vertex every two rounds and requires all vertices to use RBC, a latency of at least 2 RBCs is necessary to commit a vertex.

To improve latency, Shoal~\cite{spiegelman2024shoal} introduces the concept of ``pseudo-pipelining'', interpreting the DAG as multiple Bullshark instances to achieve the effect of having a leader vertex in every round. Shoal also proposes employing a leader reputation mechanism~\cite{cohen2022aware, tsimos2024hammerhead} to select the fastest leaders. Shoal++~\cite{arun2025shoal++} further interprets the DAG as more Bullshark instances to achieve multiple leaders per round. Additionally, Shoal++ extends Bullshark's commit rules to reduce the commit latency of leader vertices. To support a leader vertex in every round while maintaining safety, Sailfish~\cite{shrestha2025sailfish} introduced additional consensus messages (such as $\NoVote$), along with extra requirements for leader vertices. Furthermore, Sailfish also provides faster commit rules and support for multiple leaders per round.

The aforementioned protocols focus on latency rather than communication complexity. They rely on batching $\Omega(n)$ transactions within each vertex to amortize the high communication overhead. Sparse Bullshark~\cite{anoprenko2025dags} highlights the limitations of this approach and reduces metadata communication complexity by requiring each vertex to include only $O(\lambda)$ edges. We share a similar objective; however, we adopt a strategy that differentiates the treatment of leader vertices from that of non-leader vertices. Another concurrent work, Angelfish~\cite{yu2025angelfish}, offers a more flexible proposal mechanism building upon Sailfish. In their protocol, replicas can freely choose to execute RBC or merely cast a vote. 
Consequently, the protocol's communication complexity depends on the number of replicas performing RBC, which is distinct from our method.

\noindentparagraph{\textbf{Uncertified DAG.}}
Recently, several works have proposed DAG-based protocols that operate without relying on RBC or CBC. Prominent examples include Cordial Miner~\cite{keidar2023cordial} and Mysticeti~\cite{babel2025mysticeti}, which rely on best-effort broadcast for vertex dissemination. While this approach achieves ideal latency in good cases, the inherent tradeoff is the requirement for immediate data fetching to retrieve missing data. This process inevitably incurs additional latency and communication overhead under failures. 
BBCA-CHAIN~\cite{malkhi2024bbca} proposes that leaders employ BBCA (a single-shot PBFT) for broadcasting to facilitate rapid commitment, while other replicas utilize best-effort broadcast. 
Its philosophy of hybridizing DAG and leader-based BFT paradigms provides a new perspective on protocol design.

\bibliographystyle{ACM-Reference-Format}
\bibliography{section/ref}

@inproceedings{spiegelman2022bullshark,
  title={Bullshark: Dag bft protocols made practical},
  author={Spiegelman, Alexander and Giridharan, Neil and Sonnino, Alberto and Kokoris-Kogias, Lefteris},
  booktitle={Proceedings of the 2022 ACM SIGSAC Conference on Computer and Communications Security},
  pages={2705--2718},
  year={2022}
}

@article{spiegelman2022bullshark_p,
  title={Bullshark: The partially synchronous version},
  author={Spiegelman, Alexander and Giridharan, Neil and Sonnino, Alberto and Kokoris-Kogias, Lefteris},
  journal={arXiv preprint arXiv:2209.05633},
  year={2022}
}

@inproceedings{arun2025shoal++,
  title={Shoal++: High Throughput $\{$DAG$\}$$\{$BFT$\}$ Can Be Fast and Robust!},
  author={Arun, Balaji and Li, Zekun and Suri-Payer, Florian and Das, Sourav and Spiegelman, Alexander},
  booktitle={22nd USENIX Symposium on Networked Systems Design and Implementation (NSDI 25)},
  pages={813--826},
  year={2025}
}

@inproceedings{spiegelman2024shoal,
  title={Shoal: Improving dag-bft latency and robustness},
  author={Spiegelman, Alexander and Arun, Balaji and Gelashvili, Rati and Li, Zekun},
  booktitle={International Conference on Financial Cryptography and Data Security},
  pages={92--109},
  year={2024},
  organization={Springer}
}

@inproceedings{shrestha2025sailfish,
  title={Sailfish: Towards improving the latency of dag-based bft},
  author={Shrestha, Nibesh and Shrothrium, Rohan and Kate, Aniket and Nayak, Kartik},
  booktitle={2025 IEEE Symposium on Security and Privacy (SP)},
  pages={1928--1946},
  year={2025},
  organization={IEEE}
}

@article{yu2025angelfish,
  title={Angelfish: Leader, DAG, or Anywhere in Between},
  author={Yu, Qianyu and Losa, Giuliano and Shrestha, Nibesh and Wang, Xuechao},
  journal={arXiv preprint arXiv:2509.15847},
  year={2025}
}

@inproceedings{keidar2021all,
  title={All you need is dag},
  author={Keidar, Idit and Kokoris-Kogias, Eleftherios and Naor, Oded and Spiegelman, Alexander},
  booktitle={Proceedings of the 2021 ACM Symposium on Principles of Distributed Computing},
  pages={165--175},
  year={2021}
}

@inproceedings{danezis2022narwhal,
  title={Narwhal and tusk: a dag-based mempool and efficient bft consensus},
  author={Danezis, George and Kokoris-Kogias, Lefteris and Sonnino, Alberto and Spiegelman, Alexander},
  booktitle={Proceedings of the Seventeenth European Conference on Computer Systems},
  pages={34--50},
  year={2022}
}

@article{anoprenko2025dags,
  title={DAGs for the Masses},
  author={Anoprenko, Michael and Tonkikh, Andrei and Spiegelman, Alexander and Kuznetsov, Petr},
  journal={arXiv preprint arXiv:2506.13998},
  year={2025}
}

@inproceedings{malkhi2024bbca,
  title={Bbca-chain: Low latency, high throughput bft consensus on a dag},
  author={Malkhi, Dahlia and Stathakopoulou, Chrysoula and Yin, Maofan},
  booktitle={International Conference on Financial Cryptography and Data Security},
  pages={51--73},
  year={2024},
  organization={Springer}
}

@inproceedings{keidar2023cordial,
  title={Cordial Miners: Fast and Efficient Consensus for Every Eventuality},
  author={Keidar, Idit and Naor, Oded and Poupko, Ouri and Shapiro, Ehud},
  booktitle={37th International Symposium on Distributed Computing (DISC 2023)},
  pages={26--1},
  year={2023},
  organization={Schloss Dagstuhl--Leibniz-Zentrum f{\"u}r Informatik}
}

@inproceedings{babel2025mysticeti,
  title={MYSTICETI: Reaching the Latency Limits with Uncertified DAGs},
  author={Babel, Kushal and Chursin, Andrey and Danezis, George and Kichidis, Anastasios and Kokoris-Kogias, Lefteris and Koshy, Arun and Sonnino, Alberto and Tian, Mingwei},
  booktitle={Network and Distributed Systems Security Symposium (NDSS)},
  year={2025}
}

@article{baird2016swirlds,
  title={The swirlds hashgraph consensus algorithm: Fair, fast, byzantine fault tolerance},
  author={Baird, Leemon},
  journal={Swirlds Tech Reports SWIRLDS-TR-2016-01, Tech. Rep},
  volume={34},
  pages={9--11},
  year={2016}
}

@inproceedings{gkagol2019aleph,
  title={Aleph: Efficient atomic broadcast in asynchronous networks with byzantine nodes},
  author={G{\k{a}}gol, Adam and Le{\'s}niak, Damian and Straszak, Damian and {\'S}wi{\k{e}}tek, Micha{\l}},
  booktitle={Proceedings of the 1st ACM Conference on Advances in Financial Technologies},
  pages={214--228},
  year={2019}
}

@inproceedings{jovanovic2025mahi,
  title={Mahi-mahi: Low-latency asynchronous bft dag-based consensus},
  author={Jovanovic, Philipp and Kokoris-Kogias, Lefteris and Kumara, Bryan and Sonnino, Alberto and Tennage, Pasindu and Zablotchi, Igor},
  booktitle={2025 IEEE 45th International Conference on Distributed Computing Systems (ICDCS)},
  pages={549--559},
  year={2025},
  organization={IEEE}
}

@inproceedings{dai2023gradeddag,
  title={GradedDAG: An asynchronous DAG-based BFT consensus with lower latency},
  author={Dai, Xiaohai and Zhang, Zhaonan and Xiao, Jiang and Yue, Jingtao and Xie, Xia and Jin, Hai},
  booktitle={2023 42nd International Symposium on Reliable Distributed Systems (SRDS)},
  pages={107--117},
  year={2023},
  organization={IEEE}
}

@article{bracha1987asynchronous,
  title={Asynchronous Byzantine agreement protocols},
  author={Bracha, Gabriel},
  journal={Information and computation},
  volume={75},
  number={2},
  pages={130--143},
  year={1987},
  publisher={Elsevier}
}

@inproceedings{abraham2021good,
  title={Good-case latency of byzantine broadcast: A complete categorization},
  author={Abraham, Ittai and Nayak, Kartik and Ren, Ling and Xiang, Zhuolun},
  booktitle={Proceedings of the 2021 ACM Symposium on Principles of Distributed Computing},
  pages={331--341},
  year={2021}
}

@inproceedings{das2021asynchronous,
  title={Asynchronous data dissemination and its applications},
  author={Das, Sourav and Xiang, Zhuolun and Ren, Ling},
  booktitle={Proceedings of the 2021 ACM SIGSAC Conference on Computer and Communications Security},
  pages={2705--2721},
  year={2021}
}

@inproceedings{alhaddad2022balanced,
  title={Balanced byzantine reliable broadcast with near-optimal communication and improved computation},
  author={Alhaddad, Nicolas and Das, Sourav and Duan, Sisi and Ren, Ling and Varia, Mayank and Xiang, Zhuolun and Zhang, Haibin},
  booktitle={Proceedings of the 2022 ACM Symposium on Principles of Distributed Computing},
  pages={399--417},
  year={2022}
}

@inproceedings{shrestha2025optimistic,
  title={Optimistic, signature-free reliable broadcast and its applications},
  author={Shrestha, Nibesh and Yu, Qianyu and Kate, Aniket and Losa, Giuliano and Nayak, Kartik and Wang, Xuechao},
  booktitle={Proceedings of the 2025 ACM SIGSAC Conference on Computer and Communications Security},
  pages={3780--3794},
  year={2025}
}

@inproceedings{yin2019hotstuff,
  title={HotStuff: BFT consensus with linearity and responsiveness},
  author={Yin, Maofan and Malkhi, Dahlia and Reiter, Michael K and Gueta, Guy Golan and Abraham, Ittai},
  booktitle={Proceedings of the 2019 ACM symposium on principles of distributed computing},
  pages={347--356},
  year={2019}
}

@article{malkhi2023hotstuff,
  title={Hotstuff-2: Optimal two-phase responsive bft},
  author={Malkhi, Dahlia and Nayak, Kartik},
  journal={Cryptology ePrint Archive},
  year={2023}
}

@inproceedings{castro1999practical,
  title={Practical byzantine fault tolerance},
  author={Castro, Miguel and Liskov, Barbara and others},
  booktitle={OSDI},
  series={99},
  number={1999},
  pages={173--186},
  year={1999}
}

@article{buchman2018latest,
  title={The latest gossip on BFT consensus},
  author={Buchman, Ethan and Kwon, Jae and Milosevic, Zarko},
  journal={arXiv preprint arXiv:1807.04938},
  year={2018}
}

@inproceedings{chan2023simplex,
  title={Simplex consensus: A simple and fast consensus protocol},
  author={Chan, Benjamin Y and Pass, Rafael},
  booktitle={Theory of Cryptography Conference},
  pages={452--479},
  year={2023},
  organization={Springer}
}

@inproceedings{gelashvili2022jolteon,
  title={Jolteon and ditto: Network-adaptive efficient consensus with asynchronous fallback},
  author={Gelashvili, Rati and Kokoris-Kogias, Lefteris and Sonnino, Alberto and Spiegelman, Alexander and Xiang, Zhuolun},
  booktitle={International conference on financial cryptography and data security},
  pages={296--315},
  year={2022},
  organization={Springer}
}

@article{jalalzai2023fast,
  title={Fast-hotstuff: A fast and robust bft protocol for blockchains},
  author={Jalalzai, Mohammad M and Niu, Jianyu and Feng, Chen and Gai, Fangyu},
  journal={IEEE Transactions on Dependable and Secure Computing},
  volume={21},
  number={4},
  pages={2478--2493},
  year={2023},
  publisher={IEEE}
}

@inproceedings{giridharan2023beegees,
  title={Beegees: stayin'alive in chained bft},
  author={Giridharan, Neil and Suri-Payer, Florian and Ding, Matthew and Howard, Heidi and Abraham, Ittai and Crooks, Natacha},
  booktitle={Proceedings of the 2023 ACM Symposium on Principles of Distributed Computing},
  pages={233--243},
  year={2023}
}

@inproceedings{kotla2007zyzzyva,
  title={Zyzzyva: speculative byzantine fault tolerance},
  author={Kotla, Ramakrishna and Alvisi, Lorenzo and Dahlin, Mike and Clement, Allen and Wong, Edmund},
  booktitle={Proceedings of twenty-first ACM SIGOPS symposium on Operating systems principles},
  pages={45--58},
  year={2007}
}

@inproceedings{gueta2019sbft,
  title={SBFT: A scalable and decentralized trust infrastructure},
  author={Gueta, Guy Golan and Abraham, Ittai and Grossman, Shelly and Malkhi, Dahlia and Pinkas, Benny and Reiter, Michael and Seredinschi, Dragos-Adrian and Tamir, Orr and Tomescu, Alin},
  booktitle={2019 49th Annual IEEE/IFIP international conference on dependable systems and networks (DSN)},
  pages={568--580},
  year={2019},
  organization={IEEE}
}

@article{naor2021cogsworth,
  title={Cogsworth: Byzantine view synchronization},
  author={Naor, Oded and Baudet, Mathieu and Malkhi, Dahlia and Spiegelman, Alexander},
  year={2021},
  publisher={Metagov}
}

@article{bravo2022making,
  title={Making byzantine consensus live},
  author={Bravo, Manuel and Chockler, Gregory and Gotsman, Alexey},
  journal={Distributed Computing},
  volume={35},
  number={6},
  pages={503--532},
  year={2022},
  publisher={Springer}
}

@article{stathakopoulou2019mir,
  title={Mir-bft: High-throughput bft for blockchains},
  author={Stathakopoulou, Chrysoula and David, Tudor and Vukolic, Marko},
  journal={arXiv preprint arXiv:1906.05552},
  volume={92},
  year={2019}
}

@inproceedings{stathakopoulou2022state,
  title={State machine replication scalability made simple},
  author={Stathakopoulou, Chrysoula and Pavlovic, Matej and Vukoli{\'c}, Marko},
  booktitle={Proceedings of the Seventeenth European Conference on Computer Systems},
  pages={17--33},
  year={2022}
}

@inproceedings{gupta2021rcc,
  title={Rcc: Resilient concurrent consensus for high-throughput secure transaction processing},
  author={Gupta, Suyash and Hellings, Jelle and Sadoghi, Mohammad},
  booktitle={2021 IEEE 37th International Conference on Data Engineering (ICDE)},
  pages={1392--1403},
  year={2021},
  organization={IEEE}
}

@inproceedings{lyu2025ladon,
  title={Ladon: High-Performance Multi-BFT Consensus via Dynamic Global Ordering},
  author={Lyu, Hanzheng and Xie, Shaokang and Niu, Jianyu and Feng, Chen and Zhang, Yinqian and Beschastnikh, Ivan},
  booktitle={Proceedings of the Twentieth European Conference on Computer Systems},
  pages={226--242},
  year={2025}
}

@inproceedings{yang2022dispersedledger,
  title={$\{$DispersedLedger$\}$:$\{$High-Throughput$\}$ byzantine consensus on variable bandwidth networks},
  author={Yang, Lei and Park, Seo Jin and Alizadeh, Mohammad and Kannan, Sreeram and Tse, David},
  booktitle={19th USENIX Symposium on Networked Systems Design and Implementation (NSDI 22)},
  pages={493--512},
  year={2022}
}

@inproceedings{giridharan2024autobahn,
  title={Autobahn: Seamless high speed BFT},
  author={Giridharan, Neil and Suri-Payer, Florian and Abraham, Ittai and Alvisi, Lorenzo and Crooks, Natacha},
  booktitle={Proceedings of the ACM SIGOPS 30th Symposium on Operating Systems Principles},
  pages={1--23},
  year={2024}
}

@inproceedings{gao2022dumbo,
  title={Dumbo-ng: Fast asynchronous bft consensus with throughput-oblivious latency},
  author={Gao, Yingzi and Lu, Yuan and Lu, Zhenliang and Tang, Qiang and Xu, Jing and Zhang, Zhenfeng},
  booktitle={Proceedings of the 2022 ACM SIGSAC Conference on Computer and Communications Security},
  pages={1187--1201},
  year={2022}
}

@book{cachin2011introduction,
  title={Introduction to reliable and secure distributed programming},
  author={Cachin, Christian and Guerraoui, Rachid and Rodrigues, Lu{\'\i}s},
  year={2011},
  publisher={Springer Science \& Business Media}
}

@article{dwork1988consensus,
  title={Consensus in the presence of partial synchrony},
  author={Dwork, Cynthia and Lynch, Nancy and Stockmeyer, Larry},
  journal={Journal of the ACM (JACM)},
  volume={35},
  number={2},
  pages={288--323},
  year={1988},
  publisher={ACM New York, NY, USA}
}

@article{fischer1985impossibility,
  title={Impossibility of distributed consensus with one faulty process},
  author={Fischer, Michael J and Lynch, Nancy A and Paterson, Michael S},
  journal={Journal of the ACM (JACM)},
  volume={32},
  number={2},
  pages={374--382},
  year={1985},
  publisher={ACM New York, NY, USA}
}

@article{boneh2004short,
  title={Short signatures from the Weil pairing},
  author={Boneh, Dan and Lynn, Ben and Shacham, Hovav},
  journal={Journal of cryptology},
  volume={17},
  number={4},
  pages={297--319},
  year={2004},
  publisher={Springer}
}

@inproceedings{cohen2022aware,
  title={Be aware of your leaders},
  author={Cohen, Shir and Gelashvili, Rati and Kogias, Lefteris Kokoris and Li, Zekun and Malkhi, Dahlia and Sonnino, Alberto and Spiegelman, Alexander},
  booktitle={International Conference on Financial Cryptography and Data Security},
  pages={279--295},
  year={2022},
  organization={Springer}
}

@inproceedings{tsimos2024hammerhead,
  title={Hammerhead: Leader reputation for dynamic scheduling},
  author={Tsimos, Giorgos and Kichidis, Anastasios and Sonnino, Alberto and Kokoris-Kogias, Lefteris},
  booktitle={2024 IEEE 44th International Conference on Distributed Computing Systems (ICDCS)},
  pages={1377--1387},
  year={2024},
  organization={IEEE}
}

@inproceedings{alqahtani2021bottlenecks,
  title={Bottlenecks in blockchain consensus protocols},
  author={Alqahtani, Salem and Demirbas, Murat},
  booktitle={2021 IEEE International Conference on Omni-Layer Intelligent Systems (COINS)},
  pages={1--8},
  year={2021},
  organization={IEEE}
}

@inproceedings{blackshear2024sui,
  title={Sui lutris: A blockchain combining broadcast and consensus},
  author={Blackshear, Sam and Chursin, Andrey and Danezis, George and Kichidis, Anastasios and Kokoris-Kogias, Lefteris and Li, Xun and Logan, Mark and Menon, Ashok and Nowacki, Todd and Sonnino, Alberto and others},
  booktitle={Proceedings of the 2024 on ACM SIGSAC Conference on Computer and Communications Security},
  pages={2606--2620},
  year={2024}
}

@misc{suiscan,
  author       = {{SuiScan}},
  title        = {Suiscan: Transaction Blocks Explorer},
  year         = {2026},
  url          = {https://suiscan.xyz/mainnet/txs/tx-blocks},
  lastaccessed = {Accessed: 2026-01-31}
}

@misc{aptoscan,
  author       = {{Aptoscan}},
  title        = {Aptoscan: Blocks Explorer},
  year         = {2026},
  url          = {https://aptoscan.com/blocks},
  lastaccessed         = {Accessed: 2026-01-31}
}

@online{GoogleCloudPerfDashboard,
  author       = {{Google Cloud}},
  title        = {Performance Dashboard Overview},
  year         = {2026},
  url          = {https://docs.cloud.google.com/network-intelligence-center/docs/performance-dashboard/concepts/overview},
  lastaccessed = {March 24, 2026}
}

\appendix

\section{Additional Illustrations}\label{a:illustration}

\Cref{fig:fast_vote} provides an illustration of the fast-vote mechanism in \Cref{ss:protocol}. The left side illustrates the liveness violation. Since replica \node{3} skips from round $r$ to round $r+2$ and the Byzantine replica \node{4} does not reference $LV_r$, $LV_r$ fails to receive sufficient votes to be committed. Similarly, as replica \node{1} skips from round $r+1$ to round $r+3$ (or higher), $LV_{r+1}$ also fails to be committed. The right side shows the efficacy of the fast-vote. Despite skipping rounds, replicas \node{3} and \node{1} still cast fast-votes for $LV_{r}$ and $LV_{r+1}$ respectively, enabling them to be committed.

\begin{figure*}[t]
	\centering
	\includegraphics[width=0.9\textwidth]{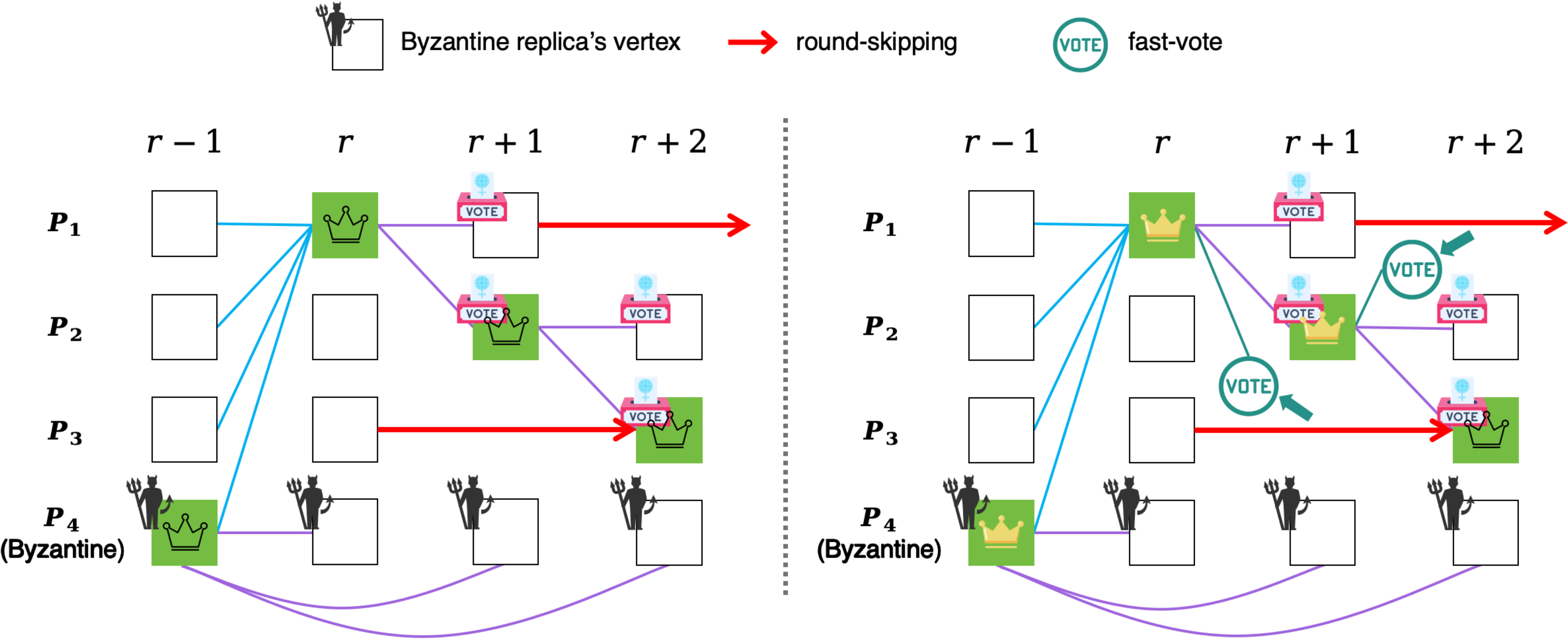    }
	\caption{Illustration of \name's fast-vote mechanism.}
\label{fig:fast_vote}
\end{figure*}

\begin{figure*}[t]
	\centering
	\includegraphics[width=0.75\textwidth]{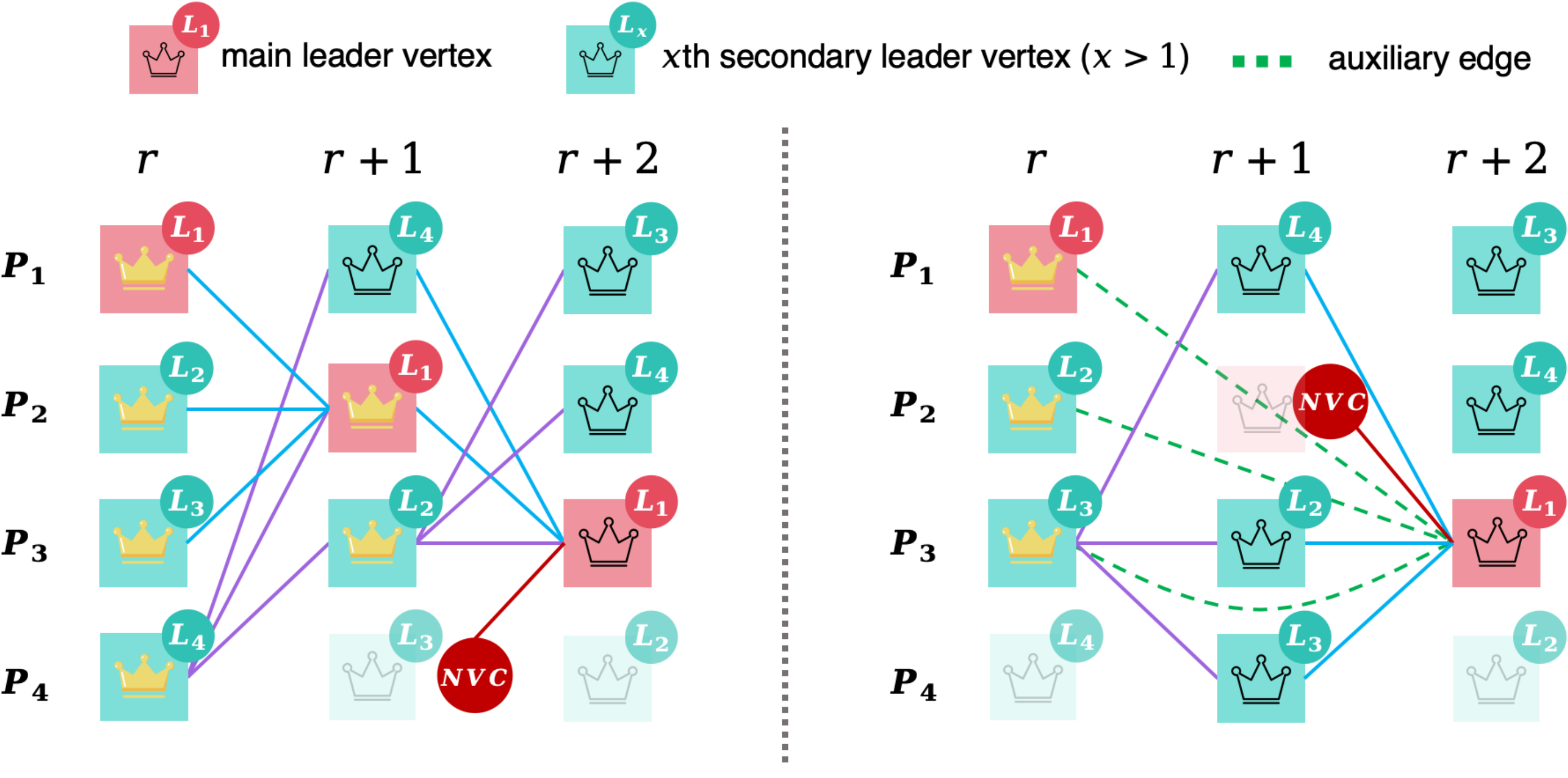}
	\caption{Illustration of Multi-leader \name. }
\label{fig:multi_leader}
\end{figure*}

\Cref{fig:multi_leader} provides an illustration of Multi-leader \name. All replicas are configured as leaders. In the left side, according to the commit rule, all leader vertices in round $r$ are directly committed. In round $r+1$, since $L_3$ has not been delivered, only $L_1$ and $L_2$ are directly committed. The main leader $L_1$ in round $r+2$ need to provide $NVC_{r+1}^{L_3}$. Due to the presence of $LV_{r+1}$, auxiliary edges are omitted. In the right side, no leader vertex in round $r+1$ can be directly committed in the absence of $LV_{r+1}$. The main leader $L_1$ in round $r+2$ is required to provide $NVC_{r+1}$ and connect to the leader vertices $L_1,L_2$ and $L_3$ in round $r$ via auxiliary edges.
\section{\name with Consistent Broadcast}\label{a:protocol_CBC}

We consider the CBC implementation presented in Narwhal~\cite{danezis2022narwhal} (referred to as Narwhal's CBC), which is adopted by many DAG-based protocols. Specifically, to broadcast a vertex $v$ for round $r$, the sender \node{i} first broadcasts $m\coloneqq\tuple{v,r,i}_i$ to all replicas. Upon receiving the valid $m$, a replica \node{j} signs it and sends $\tuple{hash(v),r,i}_j$ back to \node{i}. Once \node{i} collects at least $2f+1$ valid signatures, it aggregates them into a certificate $cert(v)$ and broadcasts it to all replicas. As described in Narwhal~\cite{danezis2022narwhal}, if a replica receives $cert(v)$ but lacks the content of $v$, it can fetch it from other replicas via random pulling.

We provide a detailed specification of CBC-based \name in \Cref{ss:protocol_C}, analyze its theoretical properties in \Cref{ss:analysis_C}, and finally discuss the multi-leader version in \Cref{ss:multi_leader_C}.

\subsection{Protocol Description}\label{ss:protocol_C}

The DAG structure of CBC-based \name is the same as basic \name. Since the data fetching process is off the critical path, we omit the details of data fetching from the consensus protocol description. Briefly, any replica that misses the data associated with a specific certificate can retrieve it from other replicas within an expected constant number of requests (guaranteed by the fact that at least $f+1$ honest replicas possess the data). Unless stated otherwise, we adopt the notation established in basic \name. We denote the availability certificate corresponding to $LV_r$ as $LC_r$.

\begin{algorithm*}[!t]
    \caption{CBC-based \name's pseudocode for replica \node{i}}
    \label{alg:protocol_C}
    \begin{algorithmic}[1]
    \footnotesize

    \Statex \textbf{Local variables:}
    \StateX struct vertex $v$; struct fast-vote $fv$
    {\color{magenta}
    \StateX struct new-round $nr$:
    \Comment{The struct of a new-round message for the leader}
    \StateXX $nr.round$ - the round to be entered by $nr$
    \StateXX $nr.source$ - the replica that sends $nr$
    \StateXX $nr.certs$ - a set of certificates in round $\geq nr.round-1$
    \Comment{Ensure the leader of $nr.round$ can receive $2f+1$ certificates}
    }
    \StateX $DAG_i[]$ - An array of sets of {\color{magenta}certificates} (indexed by rounds)
    \StateX $lastLeader$ - The most recent delivered leader vertex{\color{magenta}'s certificate}
    \vspace{-1em}
    \setlength{\columnsep}{25pt}
    \begin{multicols}{2}

    \Upon {{\color{magenta}\Call{c\_deliver$_i$}{$v, r, p$}}} \label{line:c_deliver}
    \color{magenta}
    \Comment{Receive $v$'s certificate}
    \State set $vc \gets v$'s certificate
    \normalcolor
    \If {{\color{magenta}is\_valid$(vc)$}}
    \color{magenta} 
    \State $DAG_i[r] \gets DAG_i[r]\cup\{vc\}$
    \label{line:c_deliver_start}
    \If {missing the content of $v$}
    \State start the data fetching of $v$
    \label{line:c_deliver_end}
    \EndIf
    \normalcolor
    \EndIf
    \EndUpon

    \Upon {${\color{magenta}|\bigcup_{r'\ge r-1}DAG_i[r'] \text{ from different sources}|\geq 2f+1}~\land$
    \StateX (${\color{magenta} LC_r}$ or $NVC_r \text{ is received })\text{ for } r \geq round$}
    \label{line:advance_C}
    \If {$r>round$}
    \Comment{Skip to a higher round (rule 1)}
    \State \Call{send\_fast\_vote}{$round+1,r$}
    \EndIf
    \color{magenta}
    \State $nr \gets$ \Call{create\_new\_round}{$r+1$}
    \label{line:newround_1}
    \Comment{Explicit round synchronization}
    \State send $\tuple{\NewRound, nr}_i$ to $L_{r+1}$ \label{line:newround_2}
    \If {$LC_r$ is received}
    \State broadcast $LC_r$
    \label{line:LC_sync}
    \EndIf
    \normalcolor
    \State \Call{advance\_round}{$r+1$}
    \EndUpon

    \Upon {receiving a set $\M$ of $\geq f+1$ first messages {\color{magenta}for round $\ge r+1$ from different sources} $\land$ ({\color{magenta}$LC_r$} or $NVC_r$ is received) for $r \geq round$} \label{line:jumping_C}
    \If {$r > round$}
    \Comment{Skip to a higher round (rule 2)}
    \State \Call{send\_fast\_vote}{$round+1,r$}
    \EndIf
    \State \Call{advance\_round}{$r+1$}
    \EndUpon

    \color{magenta}
    \Upon {receiving $LC_r$ for $r>round$} \label{line:jumping_LC}
    \Comment{Skip to a higher round (rule 3)}
    \State \Call{send\_fast\_vote}{$round+1,r$} and \Call{advance\_round}{$r+1$}
    \EndUpon
    \normalcolor
    
    \Upon {{\color{magenta}receiving $LC_r$} for $r>lastLeader.round$} \label{line:receive_leader_C}
    \If {haven't sent $\tuple{\NoVote,r}$ before}
    \State $lastLeader \gets {\color{magenta}LC_r}$ \label{line:last_leader_C}
    \EndIf
    \EndUpon

    \Procedure{broadcast\_vertex}{$r$}
    \State $v \gets$ \Call{create\_new\_vertex}{$r$}
    \State \Call{try\_add\_to\_dag}{$v$}
    \color{magenta}
    \State {\Call{c\_bcast$_i$}{$v, r$}}
    \normalcolor
    \EndProcedure

    \columnbreak

    {\color{magenta}
    \Procedure{create\_new\_round}{$r$}
    \State $nr.round \gets r$
    \State $nr.source \gets \node{i}$
    \State $nr.certs \gets \bigcup_{r'\ge r-1} DAG_i[r']$ from different sources
    \label{line:cert_C}
    \State \Return $nr$
    \EndProcedure
    }

    \Procedure{create\_new\_vertex}{$r$}
    \State $(v.round, v.source, v.block) \gets (r, \node{i}, blocksToPropose.\text{dequeue})$
    \State \Call{set\_leader\_edge}{$v$}
    \If {$\node{i} = L_r$}
    \color{magenta}
    \State $v.strongEdges \gets \bigcup_{r'\ge r-1}DAG_i[r']$ from different sources \label{line:strong_edge_C}
    \normalcolor
    \If {$\not\exists v' \in DAG_i[r-1]: v'.source=L_{r-1}$}
    \State $v.nvc \gets NVC_{r-1}$
    \EndIf
    \State \Call{set\_weak\_edges}{$v, r$}
    \EndIf
    \State \Call{set\_self\_edges}{$v$}
    \State \Return $v$
    \EndProcedure

    \Procedure{send\_fast\_vote}{$start, end$}
    \State start a new $timer$
    \For {$r'=start$ up to $end$ simultaneously}
    \If {{\color{magenta}\Call{c\_deliver$_i$}{$*,r'-1,L_{r'-1}$}} before $\Timeout$ $\land$ $NVC_{r'-1}$ is not
    \StateXX received $\land$ haven't sent $\tuple{\NoVote,r}_i$ before}
    \label{line:fv_condition_C}
    \State $fv \gets$ \Call{create\_fast\_vote}{$r'$}
    \State broadcast $\tuple{fv, r'}$
    \EndIf
    \EndFor
    \EndProcedure 

    \Procedure {advance\_round}{$r$}
    {\color{magenta}
    \If {\node{i}$=L_r$}
    \State \textbf{wait until} $|\bigcup_{r'\ge r-1}DAG_i[r']\text{ from different}$ $\text{sources}| \geq 2f+1$ 
    \label{line:leader_advance_C}
    \EndIf
    }
    \State $round\gets r$
    \State \Call{broadcast\_vertex}{$round$}
    {\color{magenta}
    \Upon {receiving $2f+1$ first messages for $r' \ge round$} \label{line:timer}
    \State start $timer$ for $round$
    \Comment{Delay the timer}
    \EndUpon
    }
    \EndProcedure

    \algstore{break_C}
    \end{multicols}
    \vspace{-2em}

    \end{algorithmic}
\end{algorithm*}

The pseudocode of CBC-based \name are presented in \Cref{alg:protocol_C}. To minimize redundancy, we present only the lines of code that differ from basic \name. Text highlighted in magenta explicitly denotes the modifications relative to basic \name.

\noindentparagraph{\textbf{Leveraging certificates.}}
By employing CBC, a vertex is considered delivered immediately upon the receipt of its corresponding certificate (Line~\ref{line:c_deliver}). Due to network asynchrony or Byzantine senders, a replica might possess only the certificate of a vertex. Prior to the completion of data fetching, the replica may remain unaware of the transactions and references within the vertex. We reiterate that not hinder the protocol from committing the vertex or advancing to subsequent rounds.

A crucial distinction between CBC-based \name and basic \name is that a replica can add a vertex into its local DAG prior to retrieving its content and causal history (Lines~\ref{line:c_deliver_start}--\ref{line:c_deliver_end}). The sole trade-off is that the ordering of its causal history (which is off the critical path) may need additional waiting time, since some references might be temporarily unknown.

\noindentparagraph{\textbf{Explicit round synchronization.}}
Given that CBC lacks \Cref{property:RBC}, ensuring round synchronization in CBC-based \name is more challenging than in basic \name. First, we mandate that all replicas broadcast the certificate $LC_r$ upon delivering it (Line~\ref{line:LC_sync}). This serves to prevent a Byzantine leader from disrupting round synchronization by selectively sending certificates to only a subset of replicas (since replicas are required to await the leader vertex). Merely relying on the broadcasting of $LC_r$ and $NVC_r$ is insufficient to ensure that $L_{r+1}$ enters round $r+1$, since a valid $LV_{r+1}$ is required to reference at least $2f+1$ vertices from round $\geq r$ via strong edges (Line~\ref{line:leader_advance_C}). The reason for considering vertices in rounds $>r$ will be explained in the following paragraphs. The primary issue here is that Byzantine replicas can refuse to send certificates to $L_{r+1}$. Even if an honest replica has received $2f+1$ certificates from rounds $\ge r$ and entered round $r+1$, we cannot guarantee that $L_{r+1}$ has also received these certificates (as some may originate from Byzantine senders). To address this problem, we mandate that any replica possessing a sufficient number of certificates sends a $\NewRound$ message to $L_{r+1}$ (Lines~\ref{line:newround_1}--\ref{line:newround_2}). A $\NewRound$ message for round $r+1$ encapsulates a set of at least $2f+1$ certificates from rounds $\ge r$ (Line~\ref{line:cert_C}). Leveraging these certificates, $L_{r+1}$ can successfully create a valid new vertex (Line~\ref{line:strong_edge_C}).

The round-skipping rules in CBC-based \name also differ from basic \name (Lines~\ref{line:advance_C}, \ref{line:jumping_C}, \ref{line:jumping_LC}). The intuition behind this is that CBC does not guarantee the timely synchronization of local DAG views across replicas, even after GST. By permitting replicas to consider CBC from higher rounds, we prevent honest replicas from being stalled due to the inability to timely retrieve certificates from Byzantine senders (see the proof of \Cref{lemma:synchronization_C} and \Cref{lemma:synchronization_ML_C} to understand its specific role). 

Under these rules, a scenario may arise where the total number of vertices in a specific round falls below $2f+1$. Assuming that $f$ Byzantine replicas have not broadcast vertices for round $r$. Upon receiving $LC_r$ and $f+1$ certificates in round $r$, these Byzantine replicas (provided they are not $L_{r+1}$) can immediately broadcast vertices for round $r+1$. If these vertices complete CBC before $f$ lagging honest replicas enter round $r$, honest replicas will advance to round $r+1$ in accordance with Line~\ref{line:advance_C}. This results in only $f+1$ vertices existing in round $r$. Consequently, we permit $LV_{r+1}$ to reference vertices from rounds higher than $r$ via strong edges to satisfy the $2f+1$ threshold.

\noindentparagraph{\textbf{Delay the round timer.}}
Existing DAG-based protocols (including basic \name) initiate the timer immediately upon entering a round~\cite{shrestha2025sailfish, spiegelman2022bullshark}. In CBC-based \name, this approach may cause a long timeout duration. The intuition is that when the first honest replica enters round $r+1$, the remaining replicas are only guaranteed to enter round $r$ within $\Delta$ (after GST). Consequently, in the worst-case, the timeout duration must account for the latency required for lagging honest replicas to complete the CBC for round $r$ and subsequently advance to round $r+1$. To reduce the timeout duration, a replica in CBC-based \name defers the initialization of the timer until it has received the \textit{first message} of CBC from $2f+1$ distinct replicas for the corresponding round or higher (Line~\ref{line:timer}). This strategy effectively leverages the presence of at least $f+1$ first messages from honest replicas, which guarantees that the remaining replicas can directly enter round $r+1$ within $\Delta$ (by Line~\ref{line:jumping_C}). We analyze the benefits of this design in \Cref{ss:analysis_C}.

\subsection{Correctness and Efficiency Analysis}\label{ss:analysis_C}

The correctness proof is presented in Appendix~\ref{a:correctness_C}

\noindentparagraph{\textbf{Communication complexity.}}
When using Narwhal's CBC, the size of each reference is $O(\lambda)$ (assuming threshold signatures). Consequently, the delivery of a leader vertex incurs a communication overhead of $O(\lambda n^2)$, whereas the delivery of a non-leader vertex incurs $O(\lambda n)$. The all-to-all broadcasting of the leader certificate necessitates $O(\lambda n^2)$ additional communication. Furthermore, each replica is required to transmit a $\NewRound$ message of size $O(\lambda n)$ to the leader, resulting in a total communication overhead of $O(\lambda n^2)$. In summary, the total per-round metadata communication complexity of CBC-based \name is $O(\lambda n^2)$.

Although data fetching is off the critical path, it impacts the communication complexity in bad case. Specifically, $O(n)$ replicas may request a total of $O(\lambda n)$ data (the sum of metadata contained in all vertices) from $O(n)$ other replicas per round. This results in a communication complexity of $O(\lambda n^3)$. This scenario manifests only in the presence of $O(n)$ Byzantine replicas and upon consecutive failures of the random pull. As noted in previous studies, such a bad case rarely occurs in practice~\cite{arun2025shoal++, babel2025mysticeti}. Even if the bad case occurs, the communication complexity of CBC-based \name remains lower than $O(\lambda n^4)$ complexity incurred by other protocols under identical conditions.

\noindentparagraph{\textbf{Latency analysis.}}
In CBC-based \name, the post-GST commit latency for an honest leader vertex is one CBC latency plus $1\delta$. In good case, the commit latency for non-leader vertices incurs an additional CBC latency. For Narwhal's CBC, these corresponding latencies equal to $4\delta$ and $7\delta$, respectively.

Next, we evaluate the additional latency incurred by a Byzantine leader. Analogous to basic \name, each Byzantine leader incurs an additional latency of $4\Delta+\delta$. Regarding Sailfish, we assume it operates under RBC properties with parameters $(k_1, k_2)=(3, 1)$\footnote{In Narwhal's CBC, only honest senders satisfy this property.}. Under this favorable assumption, Sailfish incurs latencies of $4\Delta+2\delta$ and $5\Delta+2\delta$ in the presence of a single Byzantine leader and consecutive Byzantine leaders, respectively. A comprehensive comparison with other protocols is provided in \Cref{table:comparison}.

Finally, we analyze the benefits yielded by delaying the timer, which is not captured in \Cref{table:comparison}. In the absence of this design, setting $\tau=4\Delta$ would be insufficient. According to the proof of \Cref{lemma:synchronization_C}, the first replica to enter a new round would require to wait for an additional CBC latency for round synchronization. Consequently, $\tau$ must be set to $7\Delta$. After delaying the timer, although the interval between entering round and initiating the timer still necessitates waiting for round synchronization, this wait becomes \textit{responsive}, incurring a delay proportional to the actual network latency $\delta$. This is significantly shorter than the timeout duration (proportional to $\Delta$), given that typically $\delta \ll \Delta$.

\subsection{Multi-leader \name with Consistent Broadcast}\label{ss:multi_leader_C}

With the same objective as Multi-leader \name, we extend CBC-based \name to a multi-leader variant. However, given that \Cref{property:CBC} of CBC is weaker than \Cref{property:RBC} of RBC, implementing a multi-leader version that guarantees both safety and liveness proves to be more intricate. In the subsequent description, we delineate the specific challenges encountered and the corresponding solutions. The pseudocode for CBC-based Multi-leader \name is presented in \Cref{alg:protocol_ML_C}, with the differences relative to both CBC-based \name and Multi-leader \name highlighted in magenta.

\begin{figure}[h]
	\centering
	\includegraphics[width=0.4\textwidth]{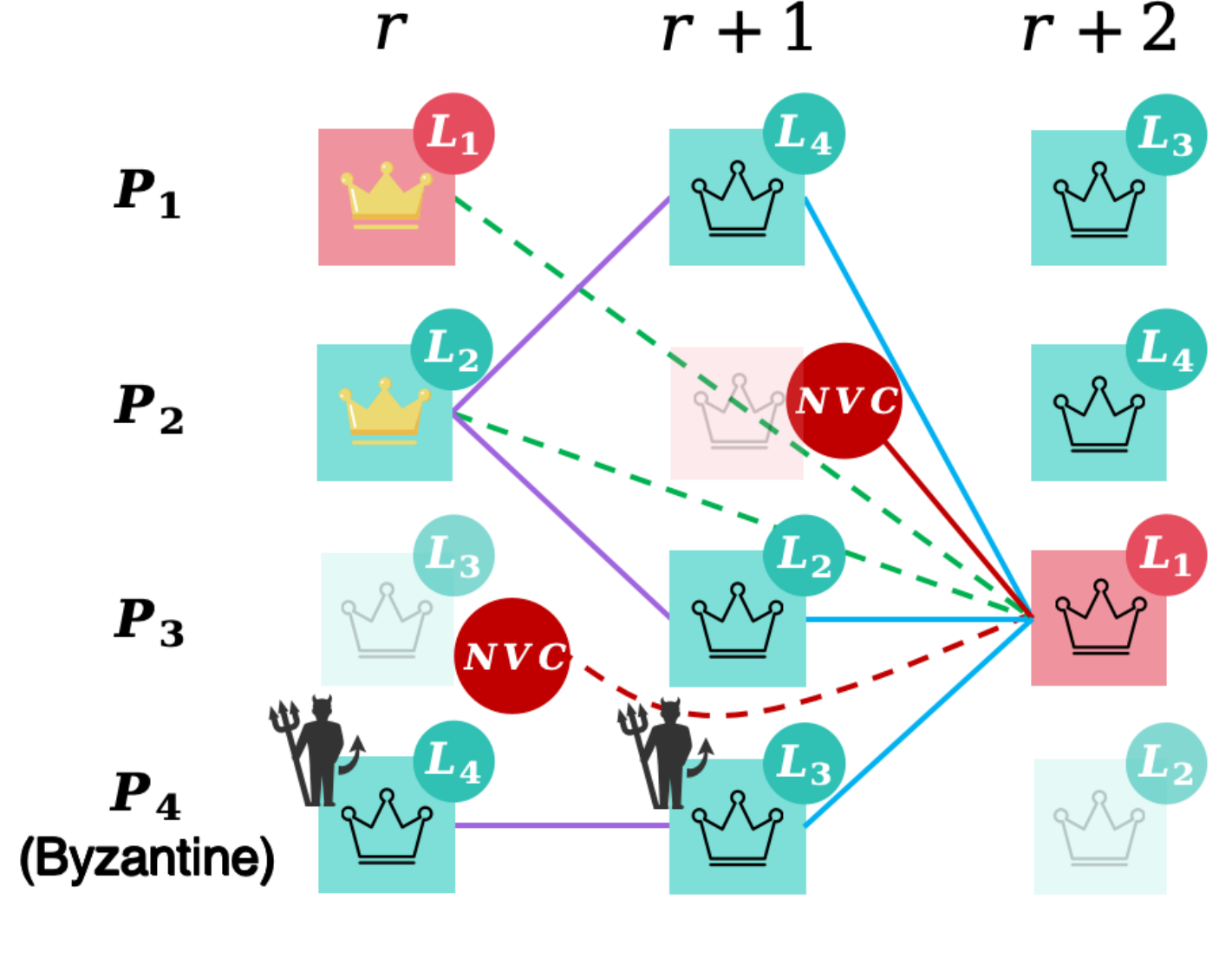}
	\caption{Illustration of CBC-based Multi-leader \name. Due to the absence of $LV_{r+1}$, $LV_{r+2}$ is required to provide $NVC_{r+1}$ and auxiliary edges. Among the vertices in round $r+1$ referenced via $LV_{r+2}$'s strong edges, $L_3$ (originating from \node{4}) contains the highest leader edge. The Byzantine replica \node{4} references $L_4$ of round $r$ via this leader edge, even though $L_3$ in round $r$ does not exist. Consequently, $LV_{r+2}$ is required to connect to $L_1$ and $L_2$ in round $r$ via auxiliary edges, and provide the $NVC_r^{\node{3}}$ to prove that $L_3$ has not been committed.}
\label{fig:multi_leader_C}
\end{figure}

\noindentparagraph{\textbf{The key challenges without reliable broadcast.}}
Recall that in Multi-leader \name, a leader edge implicitly represents a vote for all leaders with indices preceding it (in the same round). Since the properties of RBC guarantee that these leaders will be delivered by all replicas (within a fixed time after GST), within the underlying RBC protocol, all replicas implicitly validate the leader edge by awaiting the delivery of all preceding leader vertices in the same round. Consequently, for any delivered vertex, all leaders represented by its leader edge are guaranteed to be deliverable by the main leader referencing said vertex.

However, in the context of CBC, relying solely on the leader edge proves inefficient. Even considering a leader edge within a vertex created by an honest replica, a certificate with lower index may originate from a Byzantine leader. If we persist in requiring replicas to await all certificates to validate the leader edge in the underlying CBC protocol, Byzantine replicas can exploit this to stall the critical path of the consensus. This is primarily because replicas would need to first acquire these certificates via data fetching, incurring additional latency and communication overhead.

\begin{algorithm*}[!t]
    \caption{CBC-based Multi-leader \name's pseudocode for replica \node{i}}
    \label{alg:protocol_ML_C}
    \begin{algorithmic}[1]
    \footnotesize
    \algrestore{break_C}

    \Statex \textbf{Local variables:}
    \StateX struct vertex $v$:
    \Comment{The struct of a vertex in the DAG}
    \StateXX $v.nvc$ - a $\NoVote$ certificate of a leader vertex in $v.round-1$
    {\color{magenta}
    \StateXX $v.nvc'$ - a $\NoVote$ certificate of a secondary leader vertex in round $< v.round-1$
    \Comment{Only main leader vertices need to contain}
    }

    \vspace{-1em}
    \setlength{\columnsep}{25pt}
    \begin{multicols}{2}

    \Upon {\Call{c\_deliver$_i$}{$v, r, p$}}
    \Comment{Receive $v$'s certificate}
    \State set $vc \gets v$'s certificate
    \color{magenta}
    \State set $lc \gets v.leaderEdge$
    \label{line:wait_leader_edge}
    \normalcolor
    \If {is\_valid$(vc)$ {\color{magenta}$\land$ is\_valid$(lc)$}}
    \label{line:valid_leader_edge}
    \State $DAG_i[r] \gets DAG_i[r]\cup\{vc\}$
    \If {missing the content of $v$}
    \State start the data fetching of $v$
    \EndIf
    \EndIf
    \EndUpon

    \Procedure{create\_new\_vertex}{$r$}
    \State $v.round \gets r$
    \State $v.source \gets \node{i}$
    \State $v.block \gets blocksToPropose.\text{dequeue}$
    \State \Call{set\_leader\_edge}{$v$}
    \If {$\node{i} = L_r$}
    \State $v.strongEdges \gets \bigcup_{r'\ge r-1}DAG_i[r']$
    \If {$\exists v' \in DAG_i[r-1]: v'.source=L_{r-1}$}
    \For {$\ell \in \ML_{r-1}$}
    \If {$\nexists v' \in DAG_i[r-1]:v'.source=\ell$}
    \State $v.nvc\gets NVC_{r-1}^{\ell}$
    \State \textbf{break}
    \EndIf
    \EndFor
    \Else
    \State $v.nvc \gets NVC_{r-1}$
    \State $v^* \gets \arg\max_{v'\in \bigcup_{r'\ge r-1}DAG_i[r']}\{v'.leaderEdge\}$
    \State $(r^*, x^*) \gets (v^*.round,v^*.index)$
    \color{magenta}
    \For {$\ell \in \ML_{r^*}[:x^*]$}
    \label{line:start_aux}
    \If {$\nexists v'\in DAG_i[r^*]:v'.source=\ell$}
    \State $v.nvc' \gets NVC_{r^*}^{\ell}$
    \State $l^* \gets l.index-1$
    \State \textbf{break}
    \EndIf
    \EndFor
    \State $v.auxEdges \gets \{$ \Call{get\_leader\_vertex}{$r^*,x$} for
    \StateXXX $x=1$ to $\min(l^*,x^*)\}$
    \label{line:end_aux}
    \EndIf
    \normalcolor
    \State \Call{set\_weak\_edges}{$v, r$}
    \EndIf
    \State \Call{set\_self\_edges}{$v$}
    \State \Return $v$
    \EndProcedure

    \columnbreak

    \Procedure{create\_new\_round}{$r$}
    \State $nr.round \gets r$ $nr.source \gets \node{i}$
    \State $nr.certs \gets$ {\color{magenta}$\bigcup_{r'\ge r-1} DAG_i[r']$ $\bigcup$
    \StateX $\{v'.leaderEdge\, | \, v'\in\bigcup_{r'\ge r-1} DAG_i[r']\}$}
    \label{line:cert_ML_C}
    \State \Return $nr$
    \EndProcedure

    \Procedure {advance\_round}{$r$}
    \color{magenta}
    \If {$lastLeader.round < r-1$}
    \State $(r^*, x^*)\gets (lastLeader.round, lastLeader.index)$
    \For {$\forall \ell \in \ML_{r^*}[:x^*]$, $\forall r''=r^*+1$ to $r$} \Comment{additional certificates}
    \State send \Call{get\_vertex}{$\ell,r^*$} to $L_{r''}$
    \label{line:send_cert}
    \EndFor
    \For {$\forall \ell \in \ML_{r^*}[x^*+1:]$, $\forall r''=r^*+1$ to $r$}
    \label{line:start_send_novote}
    \Comment{additional $\NoVote$s}
    \State send $\tuple{\NoVote, \ell,r^*}$ to $L_{r''}$
    \EndFor
    \For {$\forall r'=r^*+1$ to $r-2,\forall \ell \in \ML_{r'}:$, $\forall r''=r^*+1$ to $r$}
    \State send $\tuple{\NoVote, \ell,r'}$ to $L_{r''}$
    \EndFor
    \label{line:end_send_novote}
    \EndIf

    \normalcolor
    \For {$\ell\in \ML_{r-1}$}
    \Comment{iterative over $\ML_{r-1}$ in order}
    \If {$\nexists v' \in DAG_i[r-1]:v'.source=\ell$}
    \State send $\tuple{\NoVote, \ell, r-1}_i$ to $L_r$
    \If {$lastLeader=LC_{r-1} \land l.index>2$}
    \State $lastLeader \gets$ \Call{get\_leader}{$r-1, \ell.index-1$}
    \EndIf
    \EndIf
    \EndFor
    \If {\node{i}$=L_r$}
    \State \textbf{wait until} $|\bigcup_{r'\ge r-1}DAG_i[r']\text{ from different}$ $\text{sources}| \geq 2f+1$ 
    \label{line:leader_advance_ML_C}
    \EndIf
    \If {$\node{i} = L_{r}$  and \Call{c\_deliver$_i$}{$LV_{r-1},r-1,L_{r-1}$}}
    \State \textbf{wait until} $((\exists v' \in DAG_i[r-1]: v'.source=\ell)$ for $\forall \ell \in$
    \StateXX $\ML_{r-1}[:x])$ $\land$ $(NVC_{r-1}^{\ell'}$ is received for $\ell'=\ML_{r-1}[x+1])$
    \label{line:wait_NVC}
    \EndIf
    \If {$\node{i} = L_r$ and $NVC_{r-1}$ is received}
    \State $v^* \gets \arg\max_{v'\in \bigcup_{r'\ge r-1}DAG_i[r']}\{v'.leaderEdge\}$
    \State $(r^*,x^*)\gets (v^*.round, v^*.index)$
    \label{line:find_highest}
    \color{magenta}
    \State \textbf{wait until} $(\exists x<x^*:((\exists v^* \in DAG_i[r^*]:$ $v^*.source=\ell)$ for 
    \StateXX $\forall \ell \in \ML_{r^*}[:x])$ and $(NVC_{r^*}^{\ell'}$
    \StateXX  is received for $\ell'=\ML_{r^*}[x+1]))$ 
    \StateXX \textbf{or} $((\exists v^* \in$ $DAG_i[r^*]$: $v^*.source=\ell)$ for $\forall \ell \in \ML_{r^*}[:x^*])$
    \label{line:wait_MLC}
    \normalcolor
    \EndIf
    \State $round\gets r$; start $timer$ for $round$
    \State \Call{broadcast\_vertex}{$round$}
    \EndProcedure
    \end{multicols}
    \vspace{-2em}

    \end{algorithmic}
\end{algorithm*}

To address this problem, a straightforward approach is to require replicas to validate only the leader edge in the underlying CBC process. While this strategy eliminates additional waiting periods and does not require the inclusion of additional edges in the vertex, it introduces the risk of the non-existence of certificates with lower indices\footnote{For instance, a Byzantine replica might intentionally reference a leader edge with a high index, while a Byzantine leader with a lower index doesn't initiate the CBC process.}. In such scenarios, the main leader might be unable to provide the requisite auxiliary edge (in accordance with the rules of Multi-leader \name). This poses a significant challenge to ensuring the existence of a leader path.

\noindentparagraph{\textbf{Solve the challenges through additional components.}}
Within the underlying CBC protocol, we require replicas to validate \textit{only} the leader edge and await the receipt of the main leader's certificate of the same round before deeming the vertex valid (which is not explicitly depicted in the pseudocode). The rationale for awaiting the latter is to prevent Byzantine replicas from referencing a leader vertex in a round where the main leader vertex does not exist, as such a leader edge is invalid. Since a single certificate cannot reflect the leader edge of its corresponding vertex, the main leader might be unable to obtain the leader edges of the referenced vertices (potentially used for creating auxiliary edges). To address this, we require replicas to additionally include the leader edges of vertices in the $\NewRound$ message (Line~\ref{line:cert_ML_C}). Correspondingly, the delivery rule for a vertex is also modified to require the receipt of both the vertex and the certificate of its leader edge (Lines~\ref{line:wait_leader_edge}--\ref{line:valid_leader_edge}).

Building upon these, we modify the requirements for \textit{auxiliary edges}. Regarding $LV_r$, we consider the highest leader edge (with round $r^*$ and index $x^*$) within all vertices referenced by its strong edges, as in Multi-leader \name (Line~\ref{line:find_highest}). Given that such an edge may originate from a Byzantine replica and it may not send $\NewRound$ messages as required, $L_r$ does not necessarily possess all certificates for $\ML_{r^*}[:x^*]$. In this case, we merely mandate that the main leader includes a $NVC$ with an index $\leq x^*$ to prove that subsequent leader vertices cannot be directly committed (Lines~\ref{line:start_aux}--\ref{line:end_aux}, \ref{line:wait_MLC}). \Cref{fig:multi_leader_C} provides a illustration of this scenario.

To ensure that the main leader receives the $NVC$s and certificates to create vertex, we mandate that replicas transmit additional messages to the main leader upon entering a new round (even via skipping). Specifically, replicas are required to send certificates for all leader vertices preceding $lastLeader$ within its respective round (Line~\ref{line:send_cert}). Furthermore, replicas must send $\NoVote$ messages for all leaders after $lastLeader$ (Lines~\ref{line:start_send_novote}--\ref{line:end_send_novote}). Note that we require replicas to transmit these messages to multiple main leaders after $lastLeader$, rather than solely to the main leader of the new round. Intuitively, this aims to prevent any main leader from failing to meet the requirements for entering a new round, which could further block the round synchronization (see \Cref{lemma:synchronization_ML_C}). We demonstrate in the proofs that the main leader is capable of leveraging these messages to establish valid $auxEdges$ and create a new vertex.

We provide the complete correctness proof for CBC-based Multi-leader \name in Appendix~\ref{a:correctness_ML_C}.

\noindentparagraph{\textbf{Communication complexity.}}
The number of references within a vertex is identical to that in CBC-based \name. Upon entering a new round, the total number of certificates sent by each replica to the main leader remains $O(n)$. In the good or average case, the number of $\NoVote$ and certificates sent by each replica is $O(n)$. This is because the average number of rounds between $lastLeader$ and the current round is $O(1)$. In the bad case, the number of $\NoVote$ and certificates sent by each replica may reach $O(n^2)$. This arises because, in the event of $O(n)$ consecutive failures, the number of rounds between $lastLeader$ and the current round extends to $O(n)$. The communication overhead of this scenario is equivalent to requiring replicas to broadcast all ($O(n)$) $\NoVote$ messages and certificates, resulting in a total communication complexity of $O(\lambda n^3)$ per round.

In summary, the per-round metadata communication complexity of CBC-based Multi-leader \name is $O(\lambda n^2)$ in the good (or average) case and $O(\lambda n^3)$ in the bad case. This remains consistent with CBC-based \name.

\noindentparagraph{\textbf{Latency analysis.}}
In the best-case scenario, the commit latency for every leader vertex is $4\delta$. In a non-optimal case, the main leader vertex requires an additional $\delta$ to gather $NVC$s and certificates. Consequently, the direct commit latency for all leader vertices becomes $5\delta$. Analogous to the reasoning presented for Multi-leader \name, provided that $x > \frac{n-f+3}{3}$ leader vertices are directly committed, the average commit latency of CBC-based Multi-leader \name outperforms that of CBC-based \name.
\section{Proofs}\label{a:correctness}

\subsection{Correctness Proof of \name}\label{a:correctness_R}

For the sake of brevity, throughout all our proofs, we implicitly use the following fact guaranteed by the properties of RBC and CBC.

\begin{fact}\label{fact:ne}
Any vertex appearing at the same position (i.e., the same round and the same source) within the DAG views of all honest replicas must be identical.
\end{fact}

\noindentparagraph{\textbf{The core of safety.}}
As highlighted in the discussion of main challenges, the cornerstone of the protocol's safety is ensuring that a path exists between a leader vertex committed by any honest replica and all subsequent leader vertices. We formalize this property through the following lemma.

\begin{lemma}\label{lemma:leader_path}
If an honest replica \node{i} directly commits the leader vertex $LV_r$, then for any valid leader vertex $LV_{r'}$ in round $r' > r$, there exists a leader path from $LV_{r'}$ to $LV_r$.
\end{lemma}

\begin{proof}
Since \node{i} has directly committed $LV_r$, there must exist at least $2f+1$ votes (either first messages or fast-votes) for $LV_r$. Among these, at least $f+1$ originate from honest replicas. Let $\mathcal{H}$ denote this set of honest replicas. According to Lines~\ref{line:receive_leader}--\ref{line:last_leader} of the protocol, the round associated with the $lastLeader$ variable for any replica in $\mathcal{H}$ must satisfy $lastLeader.round \ge r$. Furthermore, these replicas will not broadcast a $\tuple{\NoVote, r}$ message (by Line~\ref{line:novote} and Line~\ref{line:fv_condition}). Consequently, any leader edge created by these replicas in subsequent rounds is guaranteed to reference a leader vertex from a round no less than $r$. We now proceed by induction on $r'$.

\noindent\textbf{Case $r' = r+1$:} Since $|\mathcal{H}| \ge f+1$, the property of quorum intersection guarantees that an $NVC_r$ cannot be formed. Therefore, $LV_{r+1}$ must reference $LV_r$ via a strong edge.

\noindent\textbf{Case $r' \ge r+2$:} We assume the inductive hypothesis holds for all leader vertices in rounds $r''$ such that $r < r'' < r'$. Since $LV_{r'}$ references at least $2f+1$ vertices from round $r' - 1$ via strong edges, quorum intersection implies that at least one of these vertices originates from a replica in $\mathcal{H}$. We denote one such replica as \node{k}. The vertex created by \node{k} in round $r' - 1$ references a leader vertex $LV_{r^*}$ via a leader edge. Based on the property of set $\mathcal{H}$ derived earlier, we have $r^* \ge r$. If $r^* = r$, the leader path is directly established. If $r^* > r$, the inductive hypothesis guarantees that a leader path exists between $LV_{r^*}$ and $LV_r$. By transitivity, a leader path must exist between $LV_{r'}$ and $LV_r$.
\end{proof}

\noindentparagraph{\textbf{The core of liveness.}}
Given that only leader vertices can be committed, the key of liveness lies in showing that a leader vertex created by an honest replica after GST will receive sufficient votes. We formalize this property through two lemmas. Specifically, \Cref{lemma:synchronization_R} guarantees round synchronization, and \Cref{lemma:leader_commit_R} ensures that an honest leader vertex is directly committed.

\begin{lemma}\label{lemma:synchronization_R}
Let $t$ be a time after GST. If the first honest replica \node{i} enters round $r$ at time $t$, then all honest replicas will enter round $r$ or higher by time $t + k_2\Delta$.
\end{lemma}

\begin{proof}
Since \node{i} is the first honest replica to enter round $r$, it could not have advanced via round-skipping. Consequently, \node{i} must have delivered $2f+1$ vertices from round $r-1$. By \Cref{property:RBC}, all honest replicas are guaranteed to deliver these vertices within $k_2\Delta$. We proceed by considering the two cases regarding how \node{i} entered round $r$.

If \node{i} enters round $r$ by delivering $LV_{r-1}$, then all honest replicas will deliver $LV_{r-1}$ within $k_2\Delta$ and subsequently enter round $r$ (or higher). If \node{i} enters round $r$ via $NVC_{r-1}$, it must have broadcast $NVC_{r-1}$ by time $t$. Consequently, all honest replicas will receive $NVC_{r-1}$ within $\Delta$. Given that $k_2 \ge 1$ holds for all known RBC protocols, it follows that all honest replicas will enter round $r$ within $k_2\Delta$ (unless they have already advanced to a higher round).
\end{proof}

\begin{lemma}\label{lemma:leader_commit_R}
If the first honest replica enters round $r$ at time $t$ after GST, and the timeout parameters are configured such that $\tau=(k_1+k_2)\Delta$ and $\tau'=k_2\Delta$, then $LV_r$ is guaranteed to be directly committed before $t+\tau+\Delta$ provided that $L_r$ is honest.
\end{lemma}

\begin{proof}
The first honest replica to enter round $r$ initiates the timer at time $t$. By Lemma~\ref{lemma:synchronization_R}, all honest replicas (including $L_r$) are guaranteed to enter round $r$ or higher by time $t + k_2\Delta$.

It is obvious that no honest replica will broadcast $\tuple{\NoVote, r}$ before $t+\tau$. Consequently, $NVC_r$ cannot exist before $t+\tau$. The definition of $t$ also ensures that $NVC$ corresponding to higher round cannot exist before $t+\tau$. Thus, honest replicas cannot enter a round higher than $r$ via $NVC$.

Consider an honest replica \node{i}. If \node{i} has not entered a round higher than $r$ by time $t+\tau$, then by \Cref{property:RBC}, it is guaranteed to deliver at least $2f+1$ vertices from round $r$ (including $LV_r$) by time $t+\tau$. Subsequently, it will enter round $r+1$ and create a vertex referencing $LV_r$ via leader edge. If \node{i} enters round $r+1$ before $t+\tau$, since $NVC_r$ does not exist, the vertex it creates for round $r+1$ must reference $LV_r$. The last case is that \node{i} enters a round higher than $r+1$ before $t+\tau$. If it does not enter round $r+1$ (i.e., it skips round $r+1$), it must have invoked the $\mathsf{send\_fast\_vote}$ procedure covering round $r+1$. Since the first honest replica to enter round $> r+1$ must have delivered $LV_r$ (given the non-existence of $NVC_r$ and $NVC_{r+1}$), it follows from \Cref{property:RBC} that \node{i} will deliver $LV_r$ within $\tau'=k_2\Delta$ and consequently broadcast a fast-vote for it.

In summary, all honest replicas will cast a vote for $LV_r$ (via either a first message or a fast-vote) by time $t+\tau$. This guarantees that $LV_r$ is directly committed by all honest replicas before $t+\tau+\Delta$.
\end{proof}

Next, we prove that \name satisfies the four properties required by BAB (as defined in \Cref{ss:dag_bft}). We say that a leader vertex $LV$ is \textit{directly committed} by \node{i} if \node{i} invokes \Call{commit\_leader}{$LV$}. A leader vertex is \textit{indirectly committed} if it is pushed onto the $leaderStack$. Furthermore, we say that \node{i} \textit{consecutively directly commits} leader vertices $LV_r$ and $LV_{r'}$ if \node{i} directly commits both $LV_r$ and $LV_{r'}$ without directly committing any leader vertex between $r$ and $r'$.

The \textit{Integrity} property is straightforward:

\begin{theorem}\label{thm:integrity_R}
\name satisfies \textbf{Integrity}.
\end{theorem}

\begin{proof}
An honest replica \node{i} outputs \Call{a\_deliver$_i$}{$v.block, v.round,$\\$v.source$} only if vertex $v$ is already in $DAG_i$. This implies that \node{i} has already output \Call{r\_deliver$_i$}{$v, v.round, v.source$}. Consequently, the theorem follows directly from the \textit{Integrity} property of RBC.
\end{proof}

We now proceed to establish the \textit{Total order} property by leveraging \Cref{lemma:leader_path}.

\begin{lemma}\label{lemma:single_order}
If an honest replica \node{i} directly commits leader vertex $LV_r$, and an honest replica \node{j} directly commits leader vertex $LV_{r'}$ with $r' \ge r$, then \node{j} either directly or indirectly commits $LV_{r}$.
\end{lemma}

\begin{proof}
The case where $r'=r$ is trivial. When $r'>r$, by \Cref{lemma:leader_path}, there exists a leader path from $LV_{r'}$ to $LV_r$. Based on the pseudocode for commit\_leader, \node{j} indirectly commits $LV_r$ if it has not yet directly committed any $LV_{r''}$ such that $r \leq r'' < r'$. Otherwise, by an inductive argument, \node{j} has already committed $LV_r$ when directly committing $LV_{r''}$.
\end{proof}

\begin{lemma}\label{lemma:pair_order}
If an honest replica \node{i} consecutively directly commits leader vertices $LV_{r_i}$ and $LV_{r_i'}$ with $r_i' > r_i$, and an honest replica \node{j} consecutively directly commits leader vertices $LV_{r_j}$ and $LV_{r_j'}$ with $r_j' > r_j$, then \node{i} and \node{j} commit the same leader vertices between rounds $\max(r_i, r_j)$ and $\min(r_i', r_j')$ in the same order.
\end{lemma}

\begin{proof}
Without loss of generality, we only need to consider the non-trivial case where $r_i \leq r_j < r_i' \leq r_j'$. By \Cref{lemma:single_order}, both \node{i} and \node{j} will commit $LV_{r_i'}$. 
Based on the pseudocode for commit\_leader and \Cref{lemma:leader_path}, upon committing $LV_{r_i'}$, both \node{i} and \node{j} will indirectly commit all leader vertices that have a leader path to $LV_{r_i'}$ from rounds $r_i'$ down to $r_j$. Due to the deterministic logic of commit\_leader, \node{i} and \node{j} will commit these vertices in the same order.
\end{proof}

\begin{theorem}\label{thm:total_order_R}
\name satisfies \textbf{Total order}.
\end{theorem}

\begin{proof}
By inductively applying \Cref{lemma:pair_order} to every pair of honest replicas, we can deduce that all honest replicas commit same leader vertices in the same order. Based on the pseudocode for order\_vertices, all honest replicas will traverse the causal histories of these leader vertices in the same deterministic manner, and consequently a\_deliver all vertices in the same order.
\end{proof}

Next, we rely on \Cref{lemma:synchronization_R} and \Cref{lemma:leader_commit_R} to establish the \textit{Validity} property.

\begin{lemma}\label{lemma:keep_higher}
All honest replicas continuously enter higher rounds.
\end{lemma}

\begin{proof}
Assume that all honest replicas have entered round $r$ or higher. If an honest replica \node{i} enter round $r'>r$ at any time, then due to the finiteness of GST and \Cref{lemma:synchronization_R}, all honest replicas will eventually enter round $r'$ or higher. Otherwise, if all honest replicas remain in round $r$, they will invoke \Call{broadcast\_vertex}{$r$} and eventually deliver at least $2f+1$ vertices from round $r$. If any honest replica r\_delivers $LV_r$ before the timeout, then all honest replicas will satisfy the condition to enter round $r+1$ upon delivering $LV_r$. Otherwise, all honest replicas will broadcast $\tuple{\NoVote, r}$ and receive $NVC_r$, which also satisfies the condition to enter round $r+1$.
\end{proof}

\begin{theorem}\label{thm:validity}
\name satisfies \textbf{Validity}.
\end{theorem}

\begin{proof}
Suppose an honest replica \node{i} calls \Call{a\_bcast}{$b$}, which implies that it pushes $b$ into the $blockToPropose$ queue. By \Cref{lemma:keep_higher}, \node{i} continuously enters higher rounds and creates new vertices. Consequently, \node{i} will eventually create a vertex $v$ in some round $r$ that contains $b$. Due to the \textit{Validity} property of RBC, all honest replicas will eventually add $v$ to the DAG.

According to \Cref{lemma:leader_commit_R}, all honest leader vertices created after GST will be committed by all honest replicas. Facilitated by the weak edges, an honest leader will eventually create a leader vertex that has a path to $v$, and this leader vertex will be committed. Based on the pseudocode for order\_vertices, all honest replicas will eventually output \Call{a\_deliver}{$b, r, \node{i}$}.
\end{proof}

Finally, we directly establish the \textit{Agreement} property:

\begin{theorem}\label{thm:agreement}
\name satisfies \textbf{Agreement}.
\end{theorem}

\begin{proof}
If an honest replica \node{i} outputs \Call{a\_deliver$_i$}{$v.block, v.round,$\\$v.source$}, then $v$ must be in the causal history of some committed leader vertex $LV_r$. Based on the reasoning in \Cref{thm:validity}, any honest replica \node{j} will eventually commit a leader vertex $LV_{r'}$ with $r' \geq r$. According to the pseudocode for order\_vertices and \Cref{lemma:leader_path}, \node{j} will eventually output \Call{a\_deliver$_j$}{$v.block, v.round, v.source$} when traversing the causal history of $LV_r$.
\end{proof}
\subsection{Correctness Proof of CBC-based \name}\label{a:correctness_C}

We first give the key arguments of safety and liveness. 

\noindentparagraph{\textbf{The key safety argument.}}
Observe that the commit rules and the format of leader vertices in CBC-based \name- are basically identical to those in basic \name. The sole distinction lies in the fact that references take the form of certificates and may have higher rounds, which provides same guarantees. Consequently, the proof remains identical to \Cref{lemma:leader_path} and we omit it here.

\noindentparagraph{\textbf{The key liveness argument.}}
Without \Cref{property:RBC}, we must re-establish two key liveness lemmas. The core intuition relies on leveraging the explicit round synchronization and round-skipping rules. We begin by outlining the properties satisfied by CBC\cite{danezis2022narwhal} after GST.

\begin{property}\label{property:CBC}
Let $t$ be a time after GST. If an honest replica consistently broadcasts a vertex $v$ at time $t$, then all honest replicas will receive the certificate of $v$ by time $t + 3\Delta$.
\end{property}

The following two lemmas correspond to \Cref{lemma:synchronization_R} and \Cref{lemma:leader_commit_R} in basic \name.

\begin{lemma}\label{lemma:synchronization_C}
Let $t$ be a time after GST. If the first honest replica \node{i} enters round $r$ at time $t$, then all honest replicas will enter round $r$ or higher by time $t + 4\Delta$.
\end{lemma}

\begin{proof}
Since \node{i} is the first honest replica to enter round $r$, it could not have advanced via skipping. Consequently, \node{i} must have received certificates for at least $2f+1$ vertices from different sources in rounds $\ge r-1$. Among these, at least $f+1$ originate from honest replicas (let $\mathcal{H}$ denote this set). By virtue of Line~\ref{line:LC_sync} and the $\NoVote$ mechanism, all replicas are guaranteed to receive either $LC_{r-2}$ or $NVC_{r-2}$ and either $LC_{r-1}$ or $NVC_{r-1}$ within $\Delta$. Given that \node{i} sends a $\NewRound$ message to $L_r$, $L_r$ is guaranteed to obtain at least $2f+1$ certificates from different sources in rounds $\ge r-1$ and satisfy the conditions specified in Line~\ref{line:leader_advance_C} within $\Delta$. 

Furthermore, all replicas will receive the first messages from $\mathcal{H}$ within $\Delta$. Since the first honest replica to enter round $r-1$ must have done so before time $t$ and sent a $\NewRound$ message to $L_{r-1}$, $L_{r-1}$ is also guaranteed to satisfy the conditions specified in Line~\ref{line:leader_advance_C} within $\Delta$. In conjunction with Line~\ref{line:jumping_C}, this ensures that all honest replicas will enter round $r-1$ (or higher) within $\Delta$. After a subsequent latency of $3\Delta$, all honest replicas will complete the CBC for rounds $\ge r-1$ and advance to round $r$ (or higher) by Line~\ref{line:advance_C}.
\end{proof}

\begin{lemma}\label{lemma:leader_commit_C}
If the first honest replica enters round $r$ at time $t$ after GST, and the timeout parameters are configured such that $\tau=4\Delta$ and $\tau'=\Delta$, then $LV_r$ is guaranteed to be directly committed before $t+10\Delta$ provided that $L_r$ is honest.
\end{lemma}

\begin{proof}
We begin by establishing a bound on the time at which the timer for round $r$ is started. By \Cref{lemma:synchronization_C}, all honest replicas are guaranteed to enter round $r$ or higher by time $t + 4\Delta$. If no honest replica has advanced to a higher round by time $t + 5\Delta$, it is certain that all honest replicas will receive $2f+1$ \textit{first messages} for round $r$ and start the timer for round $r$ by time $t + 5\Delta$. Conversely, consider the scenario where the first honest replica does advance to a higher round. This replica must have delivered $2f+1$ vertices from rounds $\ge r$, implying that it must have initiated the timer prior to this advancement. 

Let $t^* \le t + 5\Delta$ denote the time when the first honest replica \node{i} initiates the timer for round $r$. Since the set of $2f+1$ first messages received by \node{i} for round $\ge r$ must contain at least $f+1$ messages from honest replicas, it follows from Line~\ref{line:jumping_C} that all honest replicas will enter round $r$ or a higher round by time $t^* + \Delta$.

Analogous to \Cref{lemma:leader_commit_R}, no honest replica can enter a higher round via $NVC$ before $t^*+\tau$. We proceed by considering an arbitrary honest replica \node{j} and discussing the following cases:

If \node{j} has not entered a round higher than $r$ by time $t^*+\tau$, then by \Cref{property:CBC} and the preceding argument, it is guaranteed to receive at least $2f+1$ certificates from rounds $\geq r$ (including $LC_r$) by $t^*+\tau$. Consequently, \node{j} will enter round $r+1$ and create a vertex referencing $LC_r$. If \node{j} enters round $r+1$ before $t^*+\tau$, given the non-existence of $NVC_r$, the vertex it creates for round $r+1$ must reference $LC_r$. If \node{j} enters a round higher than $r+1$ before $t^*+\tau$ and skips round $r+1$, it must have invoked the $\mathsf{send\_fast\_vote}$ procedure covering round $r+1$. 
Since the first honest replica to enter round $r+2$ possesses $LC_r$ and has not sent $\tuple{\NoVote, r}$, \node{j} will receive $LC_r$ within $\tau'=\Delta$ and subsequently broadcast a fast-vote for it.

In summary, all honest replicas will cast a vote for $LV_r$ (either via a first message or a fast-vote) by time $t^*+\tau$. This guarantees that $LV_r$ is directly committed by all honest replicas before $t^*+\tau+\Delta \leq t+10\Delta$.
\end{proof}

Since CBC also satisfies the \textit{Integrity} property, CBC-based \name satisfies the \textit{Integrity} property, and the proof is therefore omitted. 

\begin{theorem}\label{thm:integrity_C}
CBC-based \name satisfies \textbf{Integrity}.
\end{theorem}

Regarding the \textit{Total order} property, based on the argument in \Cref{ss:analysis_C}, all lemmas and their corresponding proofs from basic \name can be directly applied to CBC-based \name. Consequently, we directly obtain the following theorem:

\begin{theorem}\label{thm:total_order_C}
CBC-based \name satisfies \textbf{Total order}.
\end{theorem}

By substituting the application of \Cref{lemma:synchronization_R} with \Cref{lemma:synchronization_C} in the proof of \Cref{lemma:keep_higher}, we obtain the following lemma.

\begin{lemma}\label{lemma:keep_higher_C}
All honest replicas continuously enter higher rounds.
\end{lemma}

By substituting the application of \Cref{lemma:keep_higher} with \Cref{lemma:keep_higher_C}, and the application of \Cref{lemma:leader_commit_R} with \Cref{lemma:leader_commit_C} in the proof of \Cref{thm:validity}, we establish the \textit{Validity} property, and further obtain the \textit{Agreement} property.

\begin{theorem}\label{thm:validity_C}
CBC-based \name satisfies \textbf{Validity}.
\end{theorem}

\begin{theorem}\label{thm:agreement_C}
CBC-based \name satisfies \textbf{Agreement}.
\end{theorem}
\subsection{Correctness Proof of Multi-leader \name}\label{a:correctness_ML_R}

For rounds $r, r'$ and indices $x, x'$, we define the ordering $(r', x') \succeq (r, x)$ if and only if $(r' = r \land x' \ge x) \lor (r' > r)$ is true. When we refer to a leader edge $v^*$ as being the ``highest'' within a specific set of leader edges, it implies that the tuple $(v^*.round, v^*.index)$ associated satisfies the aforementioned relation $\succeq$ with respect to the corresponding tuples of all other leader edges in the set.

We say that a leader vertex $\MLV_r[x]$ is \textit{directly committed} by \node{i} if \node{i} invokes \Call{commit\_leader}{$r$} and $\MLV_r[x]$ is present in the corresponding $\CLS$ (Line~\ref{line:CLS}). 
A leader vertex is \textit{indirectly committed} if it is pushed onto the $leaderStack$ within $\CMV$ (Line~\ref{line:CMV}).

The argument for the \textit{Integrity} property is identical to \Cref{thm:integrity_R} and is therefore omitted.

\begin{theorem}\label{thm:integrity_ML_R}
Multi-leader \name satisfies \textbf{Integrity}.
\end{theorem}

Next, we establish the \textit{Total order} property. The overall approach aligns with the proof for \name.

\begin{lemma}\label{lemma:leader_path_ML_R}
If an honest replica \node{i} directly commits a leader vertex $\MLV_r[x]$, then for any valid main leader vertex $LV_{r'}$ in round $r' > r$, there exists a leader path from $LV_{r'}$ to $\MLV_r[x]$.
\end{lemma}

\begin{proof}
Since \node{i} directly commits $\MLV_r[x]$, there must exist at least $2f+1$ votes for $\MLV_r[x]$. These votes comprise first messages from vertices in round $r+1$ with leader edges higher than $(r,x)$, as well as fast-votes. The latter exist only for $LV_r$. Among these, at least $f+1$ originate from honest replicas. Let $\mathcal{H}$ denote this set of honest replicas. According to Lines~\ref{line:last_leader_1} and \ref{line:last_leader_2} of the protocol, the $lastLeader$ variable held by any replica in $\mathcal{H}$ must satisfy $(lastLeader.round, lastLeader.index) \succeq (r,x)$. Furthermore, these replicas will not broadcast $\tuple{\NoVote, \ML_r[y], r}$ for any $y \leq x$. Consequently, any leader edge $l$ of the vertices created by these replicas in subsequent rounds is guaranteed to satisfy $(l.round, l.index) \succeq (r,x)$. We now proceed by induction on $r'$.

\noindent\textbf{Case $r' = r+1$:} 
Since $|\mathcal{H}| \ge f+1$, the property of quorum intersection guarantees that for any $\ell \in \ML_r[:x]$, an $NVC_r^{\ell}$ cannot be formed. Therefore, if a valid leader vertex $LV_{r+1}$ exists, it must reference $\MLV_r[x]$ via a strong edge.

\noindent\textbf{Case $r' \ge r+2$:} 
We assume the inductive hypothesis holds for all leader vertices in rounds $r''$ such that $r < r'' < r'$. Since $LV_{r'}$ references at least $2f+1$ vertices from round $r'-1$ via strong edges, quorum intersection implies that at least one of these referenced vertices originates from a replica in $\mathcal{H}$. We denote one such replica as \node{k}. Based on the property of set $\mathcal{H}$ derived earlier, the leader edge $l$ of the vertex created by \node{k} in round $r' - 1$ satisfies $(l.round, l.index) \succeq (r, x)$. If $LV_{r'}$ references $LV_{r'-1}$ via strong edge (Lines~\ref{line:start_exist}--\ref{line:end_exist}), then by inductive hypothesis and transitivity, a leader path exists from $LV_{r'}$ to $\MLV_r[x]$. If $LV_{r'}$ does not reference $LV_{r'-1}$ (Lines~\ref{line:start_nexist}--\ref{line:end_nexist}), the highest leader edge $v^*$ of round $r'-1$ vertices referenced by $LV_{r'}$ must satisfy $(v^*.round, v^*.index) = (r^*, x^*) \succeq (r, x)$ (since it is higher than or equal to \node{k}'s leader edge). If $r^* = r$, $LV_{r'}$ must connect to $\MLV_r[x]$ via $auxEdges$. If $r^* > r$, $LV_{r'}$ must connect to $LV_{r^*}$ via $auxEdges$. By the inductive hypothesis and transitivity, a leader path exists between $LV_{r'}$ and $\MLV_r[x]$.
\end{proof}

By substituting the application of \Cref{lemma:leader_path} with \Cref{lemma:leader_path_ML_R} in the proof of \Cref{lemma:single_order}, and based on the pseudocode for commit\_leaders, we directly obtain the following lemma:

\begin{lemma}\label{lemma:single_order_main}
If an honest replica \node{i} directly commits main leader vertex $LV_r$, and an honest replica \node{j} directly commits main leader vertex $LV_{r'}$ with $r' \ge r$, then \node{j} either directly or indirectly commits $LV_r$.
\end{lemma}

\begin{lemma}\label{lemma:single_order_all}
If an honest replica \node{i} directly commits leader vertices $\MLV_r[:x]$ (with $x>0$), and an honest replica \node{j} directly commits main leader vertex $LV_{r'}$ with $r' > r$, then \node{j} either directly or indirectly commits leader vertices $\MLV_r[:x]$.
\end{lemma}

\begin{proof}
By \Cref{lemma:leader_path_ML_R}, any main leader vertex in a round higher than $r$ has a leader path to $\MLV_r[:x]$. Based on the pseudocode for commit\_leaders, \node{j} will indirectly commit $\MLV_r[:x]$ if it has not yet directly committed any $LV_{r''}$ such that $r < r'' < r'$. Otherwise, by an inductive argument, \node{j} has already committed $\MLV_r[:x]$ when directly committing $LV_{r''}$.
\end{proof}

\begin{lemma}\label{lemma:pair_order_ML_R}
If an honest replica \node{i} consecutively directly commits the leader vertices in rounds $r_i$ and $r_i'$, and an honest replica \node{j} consecutively directly commits the leader vertices in rounds $r_j$ and $r_j'$, then \node{i} and \node{j} commit the same leader vertices between rounds $\max(r_i, r_j)$ and $\min(r_i', r_j')$ in the same order.
\end{lemma}

\begin{proof}
Without loss of generality, we only need to consider the non-trivial case where $r_i \le r_j < r_i' \le r_j'$. Suppose that \node{i} directly commits $\MLV_{r_i'}[:x]$ in round $r_i'$. 
If $r_j' = r_i'$, by \Cref{lemma:single_order_main}, \node{j} commits at least $LV_{r_i'}$. Otherwise, by \Cref{lemma:single_order_all}, \node{j} indirectly commits $\MLV_{r_i'}[:x]$. Based on the pseudocode for commit\_leaders and \Cref{lemma:leader_path_ML_R}, the remainder of the proof is identical to that in \Cref{lemma:pair_order}.
\end{proof}

By substituting the application of \Cref{lemma:pair_order} with \Cref{lemma:pair_order_ML_R} in the proof of \Cref{thm:total_order_R}, we obtain the \textit{Total order} property:

\begin{theorem}\label{thm:total_order_ML_R}
Multi-leader \name satisfies \textbf{Total order}.
\end{theorem}

Next, we rely on following lemmas to establish the \textit{Validity} property.

\begin{lemma}\label{lemma:synchronization_ML_R}
Let $t$ be a time after GST. If the first honest replica \node{i} enters round $r$ at time $t$, then all honest replicas will enter round $r$ or higher by time $t + 2k_2\Delta$.
\end{lemma}

\begin{proof}
We distinguish two cases based on whether the replica is the main leader.

(\rmnum{1}) Replicas other than $L_r$. Following the same argumentation as in \Cref{lemma:synchronization_R}, all honest replicas are guaranteed to deliver $2f+1$ vertices from round $r-1$, including either $LV_{r-1}$ or $NVC_{r-1}$, within $k_2\Delta$. Consequently, all replicas other than $L_r$ will enter round $r$ (or higher) by time $t+k_2\Delta$.

(\rmnum{2}) The main leader $L_r$. Based on the argumentation in \Cref{lemma:leader_commit_R}, if there exists an honest replica that enters round $r' > r$ before $t+2k_2\Delta < t+\tau$, it must have delivered $LV_r$. This implies that $L_r$ has already entered round $r$ by time $t+2k_2\Delta$. Otherwise, by Case (\rmnum{1}), any honest replica will send $\tuple{\NoVote, \ell, r-1}$ to $L_r$ before $t+k_2\Delta$ for all $\ell \in \ML_{r-1}$ if it has not delivered the corresponding leader vertex by that time. If no honest replica has delivered the vertex of $\ell$ in round $r-1$ by $t+k_2\Delta$, $L_r$ will obtain $NVC_{r-1}^{\ell}$ by $t+(k_2+1)\Delta$. Conversely, by \Cref{property:RBC}, $L_r$ will deliver the vertex of $\ell$ by $t+2k_2\Delta$. Therefore, for any $\ell \in \ML_{r-1}$, $L_r$ receives either the corresponding leader vertex or the $NVC$ by $t+2k_2\Delta$. If $L_r$ has not delivered $LV_{r-1}$ by $t+k_2\Delta$, let $(r^*, x^*)$ denote the round and index of the highest leader edge among the $2f+1$ vertices delivered by \node{i} at time $t$. By the property of RBC, $L_{r}$ can deliver $\MLV_{r^*}[:x^*]$ by $t+k_2\Delta$. In summary, $L_r$ can satisfy the conditions in Lines~\ref{line:start_exist}--\ref{line:end_nexist} by $t+2k_2\Delta$ and enter round $r$ (or higher).
\end{proof}

\begin{lemma}\label{lemma:leader_commit_ML_R}
If the first honest replica enters round $r$ at time $t$ after GST, and the timeout parameters are configured such that $\tau=(k_1+2k_2)\Delta$ and $\tau'=k_2\Delta$, then $LV_r$ is guaranteed to be directly committed before $t+\tau+2\Delta$ provided that $L_r$ is honest.
\end{lemma}

\begin{proof}
We present only the parts of the proof that differ from \Cref{lemma:leader_commit_R}. First, we invoke \Cref{lemma:synchronization_ML_R} in place of \Cref{lemma:synchronization_R}. This ensures that all honest replicas enter round $r$ (or higher) by time $t+2k_2\Delta$. 

Consider an honest replica \node{i}. If \node{i} $\neq L_{r+1}$, the analysis remains identical to that in \Cref{lemma:leader_commit_R}. If \node{i} $= L_{r+1}$, we must modify the analysis for the scenario where \node{i} has not entered a round higher than $r$ by time $t+\tau$. In this scenario, given the absence of $LV_{r+1}$ and $NVC_{r+1}$, no honest replica will enter a round higher than $r+1$ prior to $t+\tau$. Consequently, all honest replicas will enter round $r+1$ and send $\NoVote$ messages for round $r$ to \node{i} ($L_{r+1}$) by time $t+\tau$. Thus, \node{i} receives these $\NoVote$ messages by $t+\tau+\Delta$ and enters round $r+1$ in accordance with Line~\ref{line:end_exist}. At this point, \node{i} will create a vertex referencing $LV_r$.

In summary, all honest replicas will cast a vote for $LV_r$ by time $t+\tau+\Delta$. This guarantees that $LV_r$ is directly committed by all honest replicas before $t+\tau+2\Delta$.
\end{proof}

\begin{lemma}\label{lemma:keep_higher_ML_R}
All honest replicas continuously enter higher rounds.
\end{lemma}

\begin{proof}
We first follow the same reasoning as in \Cref{lemma:keep_higher}, substituting the application of \Cref{lemma:synchronization_R} with \Cref{lemma:synchronization_ML_R}. At this point, all honest replicas except $L_{r+1}$ can enter round $r+1$. By the finiteness of GST and \Cref{lemma:synchronization_ML_R}, $L_{r+1}$ will eventually enter round $r+1$ as well.
\end{proof}

By substituting the application of \Cref{lemma:keep_higher} with \Cref{lemma:keep_higher_ML_R}, and the application of \Cref{lemma:leader_commit_R} with \Cref{lemma:leader_commit_ML_R} in the proof of \Cref{thm:validity}, we establish the \textit{Validity} property:

\begin{theorem}\label{thm:validity_ML_R}
Multi-leader \name satisfies \textbf{Validity}.
\end{theorem}

By substituting the application of \Cref{lemma:leader_path} with \Cref{lemma:leader_path_ML_R} in the proof of \Cref{thm:agreement}, we obtain the \textit{Agreement} property:

\begin{theorem}\label{thm:agreement_ML_R}
Multi-leader \name satisfies \textbf{Agreement}.
\end{theorem}

Finally, we demonstrate through the following theorem that in the optimistic scenario where all replicas cast votes for a certain number of secondary leaders, these secondary leaders can also be directly committed. This suffices to fulfill the design objectives of the multi-leader protocol.

\begin{theorem}\label{claim:secondary_commit_ML_R}
If the first honest replica enters round $r$ at time $t$ after GST and all leaders within $\ML_r[:x]$ are honest, then in the optimistic scenario where all replicas vote for $\MLV_r[:x]$, the leader vertices in $\MLV_r[:x]$ are guaranteed to be directly committed before $t+\tau+2\Delta$.
\end{theorem}

\begin{proof}
Obviously, no honest replica can enter a round higher than $r$ via $NVC_r$ prior to $t+\tau$. By \Cref{lemma:synchronization_ML_R}, all honest replicas will enter round $r$ (or higher) by time $t+2k_2\Delta$.

Consider an honest replica \node{i}. If \node{i} enters a round higher than $r+1$ before $t+\tau$, then some honest replica must have delivered $2f+1$ vertices from round $r+1$. By the inductive hypothesis, the leader edges of these vertices must have a round of $r$ and an index of at least $x$. If \node{i} enters round $r+1$ before $t+\tau$, then, under the assumption, its vertex for round $r+1$ must provide a leader edge with round $r$ and index $\ge x$. If \node{i} does not enter a round higher than $r$ before $t+\tau$, then according to the assumption, it is guaranteed to deliver $2f+1$ vertices from round $r$ (including $\MLV_r[:x]$) by time $t+\tau$. Combining this with the reasoning in \Cref{lemma:leader_commit_ML_R}, \node{i} will enter round $r+1$ by $t+\tau+\Delta$ and create a vertex referencing $\MLV_r[y]$ where $y \ge x$.

In summary, all honest replicas will cast a vote for $\MLV_r[:x]$ (via either a first message or a fast-vote) by time $t+\tau+\Delta$. This guarantees that $\MLV_r[:x]$ are directly committed by all honest replicas before $t+\tau+2\Delta$.
\end{proof}
\subsection{Correctness Proof of CBC-based Multi-leader \name}\label{a:correctness_ML_C}

The \textit{Integrity} property remains straightforward:

\begin{theorem}\label{thm:integrity_ML_C}
CBC-based Multi-leader \name satisfies \textbf{Integrity}.
\end{theorem}

Next, we proceed to prove the \textit{Total order} property.

\begin{lemma}\label{lemma:leader_path_ML_C}
If an honest replica \node{i} directly commits a leader vertex $\MLV_r[x]$, then for any valid main leader vertex $LV_{r'}$ in round $r' > r$, there exists a leader path from $LV_{r'}$ to $\MLV_r[x]$.
\end{lemma}

\begin{proof}
Observe that the commit rules and the format of leader vertices in CBC-based Multi-leader \name are basically identical to those in basic Multi-leader \name. The sole distinction lies in the fact that references take the form of certificates, and the strong edges of a main leader vertex may reference certificates from rounds higher than the previous round. This distinction does not impact the quorum intersection argument and the conditions satisfied by the highest leader edge as presented in \Cref{lemma:leader_path_ML_R}. Consequently, with the exception that describing the references of $LV_{r'}$ requires considering $2f+1$ vertices from distinct replicas in rounds $\geq r'-1$, the remainder of the logic is entirely consistent with \Cref{lemma:leader_path_ML_R}. 
\end{proof}

By applying \Cref{lemma:leader_path_ML_C} and following the same reasoning as in basic Multi-leader \name, we directly establish the \textit{Total order} property:

\begin{theorem}\label{thm:total_order_ML_C}
CBC-based Multi-leader \name satisfies \textbf{Total order}.
\end{theorem}

Next, we rely on following lemmas to establish the \textit{Validity} property. The proof is notably more intricate, primarily due to the difficulties introduced by CBC's weaker property.

\begin{lemma}\label{lemma:synchronization_ML_C}
Let $t$ be a time after GST. If the first honest replica \node{i} enters round $r$ at time $t$, then all honest replicas will enter round $r$ or higher by time $t + 6\Delta$.
\end{lemma}

\begin{proof}
We distinguish two cases based on whether the replica is the main leader.

(\rmnum{1}) Replicas other than $L_r$. Following the same argumentation as in \Cref{lemma:synchronization_C}, \node{i} is guaranteed to have received at least $2f+1$ certificates for rounds $\ge r-1$, as well as either $LC_{r-1}$ or $NVC_{r-1}$. Among these, at least $f+1$ originate from honest replicas (let $\mathcal{H}$ denote this set). Within $\Delta$, all honest replicas (with the exception of $L_{r-1}$) will receive the first messages from $\mathcal{H}$ along with either $LC_{r-2}$ or $NVC_{r-2}$, subsequently entering round $r-1$ or higher. 

Since the first honest replica to enter round $r-1$ must have done so before time $t$ and sent a $\NewRound$ message to $L_{r-1}$, $L_{r-1}$ is also guaranteed to satisfy the conditions specified in Line~\ref{line:leader_advance_ML_C} within $\Delta$. We now demonstrate that $L_{r-1}$ will satisfy the conditions stipulated in Line~\ref{line:wait_NVC} or Line~\ref{line:wait_MLC} by time $t+2\Delta$, thereby entering round $r-1$ (or higher).

Based on the preceding argument, at least $2f+1$ honest replicas will initiate the message transmission described in Lines~\ref{line:start_send_novote}--\ref{line:end_send_novote} for rounds $\ge r-1$ by time $t+\Delta$. Consider the values of $lastLeader$ held by each honest replica at the exact moment it initiates the transmission. Let $(\hat{r}, \hat{x})$ denote the highest round and index among these values. Consider all vertices that can be referenced by $LV_{r-1}$ via strong edges by time $t+\Delta$. Let $(r^*, x^*)$ denote the round and index of the highest leader edge contained within these vertices (which may originate from a Byzantine replica). It suffices to consider the case where $r^* < r-1$ and $\hat{r}<r-1$, since otherwise $L_{r-1}$ would have already completed the CBC or could advance to a higher round via $LC_{r^*}$ or $LC_{\hat{r}}$ (by Line~\ref{line:jumping_LC}). In this scenario, all honest replicas will send all $\NoVote$ messages and certificates for round $r-2$ to $L_{r-1}$ before time $t+\Delta$, in accordance with Lines~\ref{line:start_send_novote}--\ref{line:end_send_novote}. Consequently, for any $\ell \in \ML_{r-2}$, $L_{r-1}$ is guaranteed to receive either the certificate or the $NVC_{r-2}^{\ell}$.

If $r^* = r-2$, based on the preceding argument, $L_{r-1}$ satisfies the condition in Line~\ref{line:wait_NVC} by time $t+2\Delta$. The same condition holds if $r^* < r-2$ and $\hat{r} = r-2$\footnote{In fact, when $\hat{r} > r^*$, the set of vertices referenced by strong edges might not include the honest replica holding the highest $lastLeader$. However, we emphasize that by referencing certificates for the higher $\ML_{\hat{r}}[:\hat{x}]$, $L_{r-1}$ satisfies a condition strictly stronger than that in Line~\ref{line:wait_MLC}. This is because $LV_{\hat{r}}$ is sufficient to provide the safety guarantee for establishing a leader path to the leader vertex in round $r^*$. We omitted this scenario—where certificates from a round higher than the highest leader edge are provided—from the pseudocode to maintain clarity.}. If $r^* < r-2$ and $\hat{r} < r-2$, we require further case analysis. When $(r^*, x^*) \succeq (\hat{r}, \hat{x})$, $L_{r-1}$ will collect all corresponding certificates or $NVC$s for $\ML_{r^*}[:x^*]$ by time $t+2\Delta$. Thus, $L_{r-1}$ satisfies the condition in Line~\ref{line:wait_MLC} by $t+2\Delta$. When $(\hat{r}, \hat{x}) \succeq (r^*, x^*)$, $L_{r-1}$ will collect all certificates for $\ML_{\hat{r}}[:\hat{x}]$ by time $t+2\Delta$, thereby satisfying the condition in Line~\ref{line:wait_MLC}.

In summary, all honest replicas will initiate the CBC for rounds $\ge r-1$ by $t+2\Delta$ and complete it by $t+5\Delta$. Consequently, all replicas other than $L_r$ will enter round $r$ (or higher) by time $t+5\Delta$.

(\rmnum{2}) The main leader $L_r$. The $\NewRound$ message sent by \node{i} to $L_r$ is guaranteed to arrive by time $t+\Delta$. According to Case (\rmnum{1}), all honest replicas will send all relevant $\NoVote$ messages or certificates to $L_r$ by time $t+5\Delta$. Similar to the reasoning presented regarding $L_{r-1}$, $L_r$ will satisfy the conditions in Line~\ref{line:wait_NVC} or Line~\ref{line:wait_MLC} and enter round $r$ (or higher) by time $t+6\Delta$.
\end{proof}

\begin{lemma}\label{lemma:leader_commit_ML_C}
If the first honest replica enters round $r$ at time $t$ after GST, and the timeout parameters are configured such that $\tau=5\Delta$ and $\tau'=\Delta$, then $LV_r$ is guaranteed to be directly committed before $t+14\Delta$ provided that $L_r$ is honest.
\end{lemma}

\begin{proof}
We present only the parts of the proof that differ from \Cref{lemma:leader_commit_C}. First, we invoke \Cref{lemma:synchronization_ML_C} in place of \Cref{lemma:synchronization_C}. Based on the same argumentation as in \Cref{lemma:leader_commit_C}, the first honest replica \node{i} is guaranteed to initiate the timer for round $r$ by time $t^* \leq t+7\Delta$. Since the set of $2f+1$ first messages received by \node{i} for round $\ge r$ must contain at least $f+1$ messages from honest replicas, in accordance with the round-skipping rule, all honest replicas other than $L_r$ will enter round $r$ (or a higher round) by time $t^* + \Delta$. 
If any of these replicas enters a round higher than $r$, then, given that $NVC_r$ cannot exist prior to $t^*+\tau$, $L_r$ must have already completed the CBC for round $r$. Otherwise, $L_r$ will receive sufficient $\NoVote$ messages and certificates by time $t^*+2\Delta$, subsequently entering round $r$.

The subsequent reasoning remains consistent with \Cref{lemma:leader_commit_ML_R}, with the exception that $t$ is substituted by $t^*$.

In summary, all honest replicas will cast a vote for $LV_r$ (via either a first message or a fast-vote) by time $t^*+\tau+\Delta$. This guarantees that $LV_r$ is directly committed by all honest replicas before $t^*+\tau+2\Delta \leq t+14\Delta$.
\end{proof}

By substituting the application of \Cref{lemma:synchronization_ML_C} with \Cref{lemma:synchronization_ML_R} in the proof of \Cref{lemma:keep_higher_ML_R}, we obtain the following lemma.

\begin{lemma}\label{lemma:keep_higher_ML_C}
All honest replicas continuously enter higher rounds.
\end{lemma}

By substituting the application of \Cref{lemma:keep_higher} with \Cref{lemma:keep_higher_ML_C}, and the application of \Cref{lemma:leader_commit_R} with \Cref{lemma:leader_commit_ML_C} in the proof of \Cref{thm:validity}, we establish the \textit{Validity} property:

\begin{theorem}\label{thm:validity_ML_C}
CBC-based Multi-leader \name satisfies \textbf{Validity}.
\end{theorem}

By substituting the application of \Cref{lemma:leader_path} with \Cref{lemma:leader_path_ML_C} in the proof of \Cref{thm:agreement}, we obtain the \textit{Agreement} property:

\begin{theorem}\label{thm:agreement_ML_C}
CBC-based Multi-leader \name satisfies \textbf{Agreement}.
\end{theorem}

The final theorem corresponds to the optimistic scenario described in \Cref{claim:secondary_commit_ML_R}. Under the given assumptions, certificates from honest leaders are guaranteed to be received by all honest replicas within $\Delta$. As the modifications to the remainder of the proof relative to \Cref{lemma:leader_commit_ML_C} are analogous to those in the RBC case, we omit the proof.

\begin{theorem}\label{claim:secondary_commit_ML_C}
If the first honest replica enters round $r$ at time $t$ after GST and all leaders within $\ML_r[:x]$ are honest, then in the optimistic scenario where all replicas vote for $\MLV_r[:x]$, the leader vertices in $\MLV_r[:x]$ are guaranteed to be directly committed before $t+14\Delta$.
\end{theorem}

\end{document}